\documentclass[a4paper,onecolumn,11pt,accepted=2025-01-06]{quantumarticle} 
\pdfoutput=1

\usepackage[margin=1in]{geometry}

\usepackage{authblk}
\usepackage{array, makecell, cellspace}
\usepackage[utf8]{inputenc}
\usepackage[english]{babel}
\usepackage[T1]{fontenc}
\usepackage{amsmath}
\usepackage{hyperref}
\usepackage{enumerate}

\usepackage{lipsum}

\usepackage{epsfig}  
\usepackage{braket}
\usepackage{amsfonts}
\usepackage{xcolor}
\usepackage{amsthm}
\usepackage{bbold}
\usepackage{enumitem}   
\usepackage{ amssymb }
\usepackage{multicol}
\usepackage{amstext} 
\usepackage{array}   
\usepackage{tabu}
\usepackage{framed}
\usepackage{subfig}
\usepackage{float}
\usepackage{cancel}
\usepackage{tikz-cd,mathtools}
\usepackage{abraces}
\usepackage{tikz}
\usepackage{tikz-3dplot}
\usepackage{forest}
\usepackage{thmtools, thm-restate}
\usepackage{multirow}

\usepackage[square,numbers,sort&compress]{natbib}

\DeclareMathOperator{\Tr}{Tr}

\newtheorem{theorem}{Theorem}[section]
\newtheorem{proposition}{Proposition}[section]

\newtheorem{definition}{Definition}[section]
\newtheorem{condition}{Condition}[section]

\newtheorem{example}{Example}[section]

\newcommand{\x}{\mathbf{x}}
\newcommand{\re}[1]{ \text{Re} \left\{ #1 \right\}}
\newcommand{\tr}[1]{ \text{Tr} \left[ #1 \right]}

\newcommand{\average}[1]{\langle #1 \rangle}

\newcommand{\ketbra}[2]{\ket{#1} \hspace{-0.1cm} \bra{#2}}
\newcommand{\thetab}{{\boldsymbol{\theta}}}
\newcommand{\id}{\text{id}}
\newcommand{\rank}{\text{rank}}

\newcommand{\norm}[1]{\lVert #1 \rVert }
\newcommand{\normone}[1]{\lVert #1 \rVert _1 }

\newcommand{\nul}{\text{null}}
\newcommand\subsetsim{\mathrel{\substack{\textstyle\subset\\[-0.2ex]\textstyle\sim}}}
\newcommand{\asum}[1]{\smashoperator{\sum_{#1}}}

\DeclareMathOperator*{\Motimes}{\text{\raisebox{0.25ex}{\scalebox{0.8}{$\bigotimes$}}}}
\DeclareMathOperator\supp{supp}

\begin{document}

\title{Private and Robust States for Distributed Quantum Sensing}   

\author[1,2,3,4]{Lu\'{i}s Bugalho}
\author[4]{Majid Hassani}
\author[1,2,3]{Yasser Omar}
\author[4]{Damian Markham}

\affil[1]{Instituto Superior Técnico, Universidade de Lisboa, Portugal}
\affil[2]{Physics of Information and Quantum Technologies Group, Centro de Física e Engenharia de Materiais Avançados (CeFEMA), Portugal}
\affil[3]{PQI -- Portuguese Quantum Institute, Portugal}
\affil[4]{Sorbonne Université, CNRS, LIP6, 4 Place Jussieu, Paris F-75005, France}

\date{Jan 8, 2025}
\maketitle

\begin{abstract}
Distributed quantum sensing enables the estimation of multiple parameters encoded in spatially separated probes. While traditional quantum sensing is often focused on estimating a single parameter with maximum precision, distributed quantum sensing seeks to estimate some function of multiple parameters that are only locally accessible for each party involved. 
In such settings, it is natural to not want to give away more information than is necessary. To address this, we use the concept of privacy with respect to a function, ensuring that only information about the target function is available to all the parties, and no other information.
We define a measure of privacy (essentially how close we are to this condition being satisfied) and show it satisfies a set of naturally desirable properties of such a measure. Using this privacy measure, we identify and construct entangled resource states that ensure privacy for a given function under different resource distributions and encoding dynamics, characterized by Hamiltonian evolution.
For separable and parallel Hamiltonians, we prove that the GHZ state is the only private state for certain linear functions, with the minimum amount of required resources, up to SLOCC. Recognizing the vulnerability of this state to particle loss, we create families of private states, that remain robust even against loss of qubits, by incorporating additional resources.
We then extend our findings to different resource distribution scenarios and Hamiltonians, resulting in a comprehensive set of private and robust states for distributed quantum estimation. These results advance the understanding of privacy and robustness in multi-parameter quantum sensing.
\end{abstract}

\tableofcontents

\section{Introduction}
Quantum metrology and quantum sensing have been longstanding areas of interest in quantum information research \cite{Giovannetti2006,Sidhu2019}. In particular, they are known to surpass the classical bounds for the estimation physical parameters in experiments and reach the Heisenberg scaling \cite{Giovannetti2006,Degen2017,Hou2021,Liu2021e,Barbieri2022,Li2022}.
The idea is that a probe state interacts with some physical process, which encodes parameters into the quantum state of the probe, measuring thereafter and using the outcomes to perform an estimation the parameters.
In the case of single parameter estimation, one is interested in estimating one parameter of a system with the maximum precision \cite{Paris2009,Degen2017}. However, in the multi-parameter scenario, not only is precision important, but also what information is available. The latter has been directly linked to entangled probes across the sites. With the ongoing development of quantum networks, both over large distances \cite{Wehner2018,Proctor2017,Yehia2022} and local-area types of networks \cite{Sekatski2019}, distributed quantum sensing has emerged as a promising area of research, taking the multiparameter scenario to applications with near-term potential \cite{Knott2016,Proctor2017,Eldredge2018,Proctor2018d,Qian2019,Zhang2020,Qian2020}. Moreover, there have been some experiments demonstrating the quantum advantage for this setting \cite{Liu2021e,Cao2022}. 

Distributed quantum sensing has numerous use cases, such as clock synchronization protocols \cite{Komar2014,Nichol2022}, optical interferometry proposals \cite{Khabiboulline2018}, some preliminary work on gravity and dark-matter experiments  \cite{Leveque2021b,Conlon2022,Brady2022,Alonso2022} and physical implementation of sensors capable of measuring such quantities \cite{Fang2016,Greve2022}. One can also find different approaches to solve the optimal estimation of sensing parameters, such as variational quantum sensing \cite{Koczor2020a,Meyer2021,Kaubruegger2023}, covert sensing \cite{Hao2022,Kasai2022}, error-corrected \cite{Shettell2021} and error-mitigated \cite{Yamamoto2022} quantum sensing, and controlled enhanced quantum sensing \cite{Liu2017a}.

The framework for distributed quantum sensing has been well-established as quantum sensor networks \cite{Knott2016,Proctor2017,Proctor2018d,Eldredge2018,Qian2019,Rubio2020,Rubio2020a,Qian2020,Bringewatt2021a}, where each quantum node holds a set of resources (in our case, qubits) and has access to a local parameter $\theta_\mu$ (encoded by some local Hamiltonian; see Fig.~\ref{fig:network}). We are interested in when they want to estimate not the local parameters themselves, but some function of these parameters $f(\{\theta_\mu\}) \equiv f(\thetab)$. 
More recently, security concerns have been addressed in several works \cite{Kasai2022,Shettell2022b,Moore2023}, along with the critical concept of privacy \cite{Shettell2022a}. 
On the one hand, in a network setting there may be eavesdroppers who want to gain unauthorised access to the value of the parameters, or even simply disrupt or corrupt the estimation process \cite{Kasai2022,Shettell2022b,Moore2023}. On the other hand, some parties who are involved in the estimation itself, may like to gain more information than they should - for example the value of other parties local parameters themselves. This is addressed by the notion of privacy, which ensures  that, only information about a target estimator can be obtained \cite{Shettell2022a}. 
Who has access to what information about the parameters typically depends on the initial state used in the estimation scenario, prompting the natural question of identifying private states. 
In \cite{Shettell2022a} privacy is defined for a specific function (the average value of the parameters), and the questions of general functions, the optimality of the states, and how to deal with noise were left open. 

In this work we develop a broad framework for privacy in networks of quantum sensors. We tackle these questions by finding a way to define and quantify privacy through the analysis of the quantum Fisher information (QFI) matrix, a key metric in a quantum estimation scenario \cite{Shettell2020c,Liu2019,Sidhu2019}. From here we verify that the amount of resources and the control over the dynamics that encode the parameters determine the functions of parameters available at hand. We are able to find a general expression for the QFI matrix of any arbitrary quantum pure state, starting with stabilizer states, in the multi-parameter scenario. We use these results to prove which ones are the only private states, $i.e.$ states where the information available is only the target function. Our findings relate the amount of resources and the properties of the encoding dynamics to the functions available at hand for an estimation scenario. In particular, for local and separable encoding dynamics, it will be required a minimum amount of distributed resources to even be able to estimate a target function privately. We generalize our findings for local but not necessarily separable dynamics, by analyzing the eigenvalues of the local Hamiltonians. In creating these families of private states, we then question their robustness against several types of common Pauli noise, and for qubit-loss, which is know to hinder quantum sensing. We find that some of our private states retain information even after particle loss, while at the same time, remaining private.  

This paper is structured as follows: we start in Section \ref{section:estimation} by analysing the methods commonly used to tackle the problem of multiparameter sensing and define exactly our setup problem as a network problem. This encompasses defining three things beforehand: \textit{(i)} what are the target functions of parameters; \textit{(ii)} defining what are the resources and introducing a notation to deal with them, and finally \textit{(iii)} creating a measure for privacy from a set of premises. 
In Section \ref{section:buildingprivatestates} we address the problem of finding private states. We start by introducing important concepts and notations to deal with the upcoming problem. Then we present statements about the privacy of states both for stabilizer states and arbitrary pure states for local and separable Hamiltonians, generalizing then to arbitrary local Hamiltonians. 
We provide conditions for a given function to be private and find the unique minimum set of private states, give some examples and then provide relaxations of optimality to allow robustness.
Finally, in Section~\ref{section:robustness}, we analyze the effect of different noise types over our found families of private states. Over the manuscript we try to be consistent notation wise, using greek letters for nodes, latin letters for individual qubits, boldsymbols to denote a higher dimensional object, wether it be a partition on a set, a vector of parameters or a larger operator decomposed by smaller operators.

\begin{figure}[h!]
\begin{center}
\includegraphics[width=0.8\columnwidth]{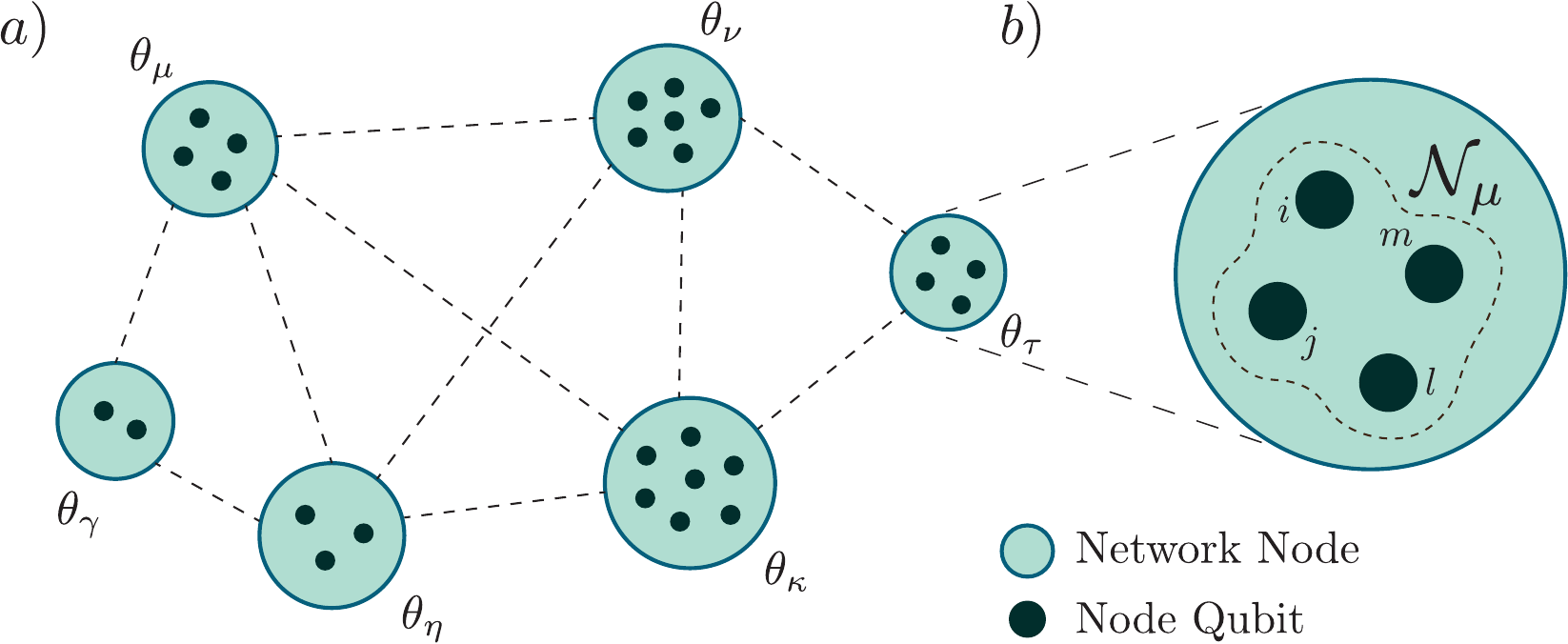} \caption{Distributed sensing scenario, consisting of \emph{a)} a network of quantum nodes, capable of distributing entangled states, where \emph{b)} each of the nodes holds their own sets of qubits $\mathcal{N}_\mu$, which can be seen as resources for quantum sensing.}
\label{fig:network}
\end{center}
\end{figure}

\section{Estimation Scenarios and Privacy \label{section:estimation}}

There are three important inputs of an estimation scenario: \textit{(i)} the initial state, \textit{(ii)} the encoding dynamics and \textit{(iii)} the measurement (see Fig.~\ref{fig:circuit}). 
Consider first the single parameter case. The initial state is transformed via the encoding dynamics to state $\rho_\theta$, which is then measured giving results $x$, with probability $p(x|\theta)$, from which one would like to estimate the parameter theta.
In the frequentist approach to the estimation problem, we are interested in retrieving information of the parameter over successive measurements. In this scenario, the overall precision of the parameter is bounded by the Fisher information (FI), $F(\theta)$, and the number of measurements made $m$ \cite{Paris2009,Sidhu2019}, in the assymptotic regime:
\begin{equation}\label{eq:fi}
\begin{aligned}
    \Delta \theta^2 &\geq \frac{1}{m F(\theta)} \\
    F(\theta) &= \int p(\x | \theta) \left(\partial_\theta \ln p(\x | \theta)\right)^2 d\x .
\end{aligned}
\end{equation}

When the estimation scenario involves a quantum initial state and quantum dynamics, the Fisher information becomes its quantum counterpart. The probability distribution is now given by the Born rule $p(\x | \theta) = \Tr [\Pi_\x \rho_\theta]$, where the set $\{ \Pi_\x \}_\x$ denotes a POVM satisfying $\sum_{\x}\Pi_\x=\mathbb{1}$. The quantum version of the FI can then be written as:
\begin{equation}
\begin{aligned}
	\mathcal{I}(p(\x | \theta)) &= \int p(\x | \theta) \left(\partial_\theta \ln p(\x | \theta)\right)^2 d\x\\
	&= \int \frac{1}{p(\x | \theta)} \left(\partial_\theta  p(\x | \theta)\right)^2 d\x \\
	&= \int \frac{1}{\Tr [\Pi_\x \rho_\theta]} \left(\partial_\theta  \Tr [\Pi_\x \rho_\theta]\right)^2 d\x \\
	&= \int \frac{\Tr [\Pi_\x \dot{\rho}_\theta]^2}{\Tr [\Pi_\x \rho_\theta]} d\x  ,
 \end{aligned}
\end{equation}
where $\dot\rho _{\theta}=\frac{\partial\rho_{\theta}}{\partial\theta}$. By maximising over all POVMs which do not explicitly depend on $\theta$, the quantum Fisher information (QFI), $\mathcal{Q}(\theta)$ appears by implicitly defining an additional operator named the symmetric logarithmic derivative (SLD), and Eq.~\ref{eq:fi} becomes:

\begin{equation}\label{eq:qfi}
\begin{aligned}
    \Delta \theta^2 &\geq \frac{1}{m \mathcal{Q}(\theta)} \\
	\mathcal{Q}( \theta) &= \tr{ \rho_\theta L_{\rho_\theta}^2} , \quad 	\partial_\theta \rho_\theta = \frac{\rho_\theta L_\theta + L_\theta \rho_\theta}{2},
\end{aligned}
\end{equation}
which is often denoted by quantum Crámer-Rao bound. Note that achieving this bound might only be achievable by implementing POVMs that depend on the target $\theta$ \cite{Paris2009}, requiring adaptive schemes. However, for some of the private states we will find, this should not pose a problem, as the optimal POVM is known to be $\theta$-independent.
The QFI depends on the state of the system immediately before measurement, meaning it will depend on the initial state of the system and the dynamics of the encoding (and additional sources of noise which we leave for discussion on the robustness section of this work). Eq.~\ref{eq:qfi} has a simpler counterpart for the case of pure states, as in this case the SLD is simply given by $L_\theta = 2\partial_\theta \rho_\theta$:

\begin{equation}
\begin{aligned}
	\mathcal{Q} (\rho_\theta) = 4 \re{ \braket{\partial_\theta \psi_\theta | \partial_\theta \psi_\theta} - \braket{\partial_\theta \psi_\theta | \psi_\theta}\braket{ \psi_\theta | \partial_\theta  \psi_\theta} }.
\end{aligned}
\label{eq:puresingle}
\end{equation}

In quantum sensing, and in particular in the distributed scenario, the goal is often to estimate a function of parameters that are spatially distributed. This means that the dynamics that encode the parameters are in spatially separated locations and the resources only have access to their own local dynamics. The natural way of addressing this problem is in a network, where each channel encoding one parameter is associated with one node. While the structure of the network impacts the initial state distribution, it does not impact the overall estimation scenario. Given a graph $G = (V,E)$, in a distributed scenario we consider a set of nodes $\{ \mu \in V(G) = 1,2,\cdots,k\}$, each of them with access to some quantum dynamics described by a quantum channel $\Lambda_{\theta_\mu}(\rho)$ (if unitary, this is described by a Hamiltonian), which encodes its own parameter $\theta_\mu$, as we outline in Fig.~\ref{fig:network}. 

In this case, the QFI becomes more complex as it involves multiple parameters. Given that the Crámer-Rao bound is a bound on the variance, in the multi-parameter scenario, the variance becomes a covariance matrix. Therefore, it makes sense that the QFI is also a matrix. The description follows naturally from Eq.~\ref{eq:fi} by substituting $\partial_{\theta} \rightarrow \partial_{\theta_\mu}, \partial_{\theta_\nu}$. In this way, the SLD is now defined with respect to one parameter $L_{\theta_\mu}$ and the Crámer-Rao bound and QFI matrix entries are given by \cite{Proctor2017,Sidhu2019,Liu2019,Goldberg2021}:

\begin{equation}
\begin{aligned}
    COV(\thetab) &\geq \frac{1}{m}\boldsymbol{\mathcal{Q}}(\thetab)^{-1} ,\\
	\boldsymbol{\mathcal{Q}}_{\mu\nu} (\rho_{\thetab}) &=\tr{  \rho_{\thetab} \frac{ L_{\theta_\mu} L_{\theta_\nu} + L_{\theta_\nu} L_{\theta_\mu} }{2} }  , \quad 	\partial_{\theta_\mu} \rho_\theta = \frac{\rho_\theta L_{\theta_\mu} + L_{\theta_\mu} \rho_\theta}{2},
\end{aligned}
\label{eq:cramermulti}
\end{equation}
where $A\geq B$ is meant in the sense that $A-B$ is a semi-definite positive matrix. Expanding for pure states where $\rho_{\thetab} = \ketbra{\psi_{\thetab}}{\psi_{\thetab}}$, just as in Eq. \ref{eq:puresingle} using $L_{\theta_\mu} = 2 \partial_{\theta_\mu} \rho_{\thetab}$, one can verify that the QFI simplifies to:
\begin{align}
	\boldsymbol{\mathcal{Q}}_{\mu\nu} (\rho_{\thetab}) &=4 \re{ \braket{\partial_{\theta_\mu} \psi_{\thetab} | \partial_{\theta_\nu} \psi_{\thetab} }  - \braket{ \partial_{\theta_\mu} \psi_{\thetab} | \psi_{\thetab}}\braket{ \psi_{\thetab} | \partial_{\theta_\nu} \psi_{\thetab}}  }.
\label{eq:qfipuremulti}
\end{align}
Since each node only has access to its own parameters in distributed scenarios, it is reasonable to assume the encoding dynamics are local and unitary (see Fig.~\ref{fig:circuit}). Under these assumptions over the encoding dynamics, there are still multiple scenarios one can analyze (see Fig.~\ref{fig:circuit}). Let us detail the most general one given in Fig.~\ref{fig:circuit}c :

\begin{equation}
\begin{aligned}
	\Lambda_{\theta_\mu}(\rho) &= e^{-i \theta_\mu \boldsymbol{H}_\mu} , \quad \boldsymbol{H}_\mu = \sum_{P_j \in \mathcal{P}_\mu} b_j P_j , \quad \mathcal{P}_\mu = \{ P \in \mathcal{P}_n : P = \Motimes_{j\not\in \mu}\mathbb{1}_j \Motimes_{j\in \mu} \sigma_j \},
\end{aligned}
\label{eq:generaldynamics}
\end{equation}
where $\mathcal{P}_n$ is the set of all Pauli strings with size $n$, which is the total amount of qubits (see subsection detailing the notation for resources further ahead). Note that the Pauli strings of size $n$ form an orthogonal basis for the vector space of complex $2^n\times2^n$ matrices. Although this description can be quite general, what is usually the subject of sensing scenarios are channels which are local and separable for each qubit \cite{Proctor2017,Rubio2020a}(see Fig.~\ref{fig:circuit}b):

\begin{equation}
\begin{aligned}
	\Lambda_{\theta_\mu}(\rho) = e^{-i \theta_\mu \boldsymbol{G}_\mu} , \quad \boldsymbol{G}_\mu = \sum_{j\in\mu} G_j ,  G_j \equiv \mathbb{1} \otimes \cdots \otimes G_j \otimes \cdots \otimes \mathbb{1} \implies \Lambda_{\theta_\mu}(\rho) = \Motimes_{j\in\mu} e^{-i \theta_\mu G_j} ,
\end{aligned}
\label{eq:dynamics}
\end{equation}
where we did not fixed a given generator of translation for any qubit, $i.e.$ the operators $G_j$ are described by $\vec{x}\cdot\vec{\sigma}$, where $\vec{x}\in\mathbb{R}^3$ is normalized, and $\vec{\sigma}$ is the Pauli vector. As it is known, the Pauli group on 1 qubit exponentiates to the unitary group on 1 qubit $U(2)$ \cite{Reilly2023}. This consists on the most general local and separable encoding dynamics, where each qubit acquires the phase once. One can also refer to this case as parallel, in the sense that if every qubit encoding is separable from one another, then each local encoding can be done at the same time, in parallel, and consequently so can the global dynamics. If one were to consider more general encoding dynamics, such as the case in Eq.~\ref{eq:generaldynamics} and Fig.~\ref{fig:circuit}c, one may arrive at a problem that quickly becomes infeasible, and the definition of sampling the system starts to drift. This is due to the fact that, for separable dynamics, one might argue that each qubit goes through the parameter channel only once, or is sampled only one time \cite{Zwierz2010}. With more complex Hamiltonians, this is not always the case. Some Hamiltonians cannot be directly rewrote into only one access to a channel for each qubit. We will discuss this case further, and provide a statement to building private states for a given arbitrary Hamiltonian.

Nonetheless, for separable encoding dynamics, even though we do not fix the generators $G_j$, we can reduce the problem complexity by fixing the dynamics and resorting to the following lemma:

\begin{figure}[t]
\begin{center}
\includegraphics[width=\columnwidth]{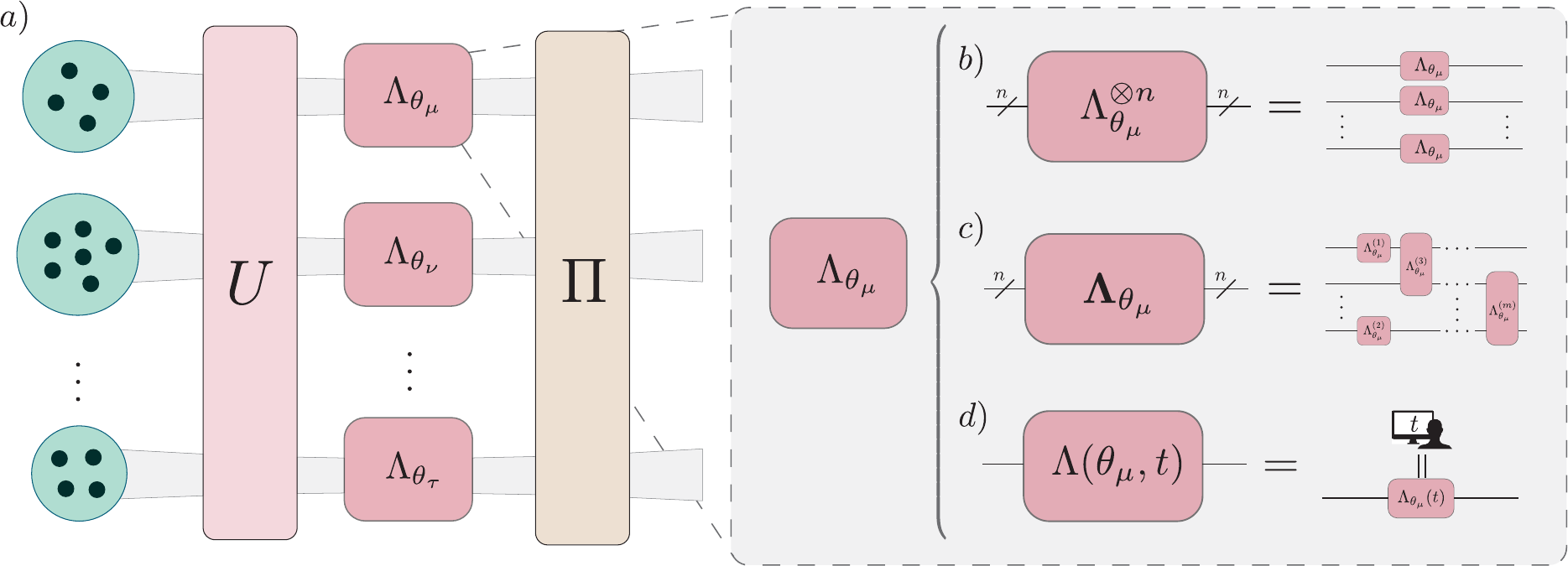} \caption{Quantum circuit description of the estimation scenario. \emph{a} Each network node qubits' have access only to their encoding dynamics, which imprint each of the parameters $\theta_\mu$ locally. Nonetheless, the initial state is not bounded to local states, they can be arbitrarily entangled between all the sets of qubits, illustrated above as a global entanglement preperation unitary $U$, and so can the final POVM $\Pi$. We divide the local encoding dynamics into three types: \emph{b)} parallel single-qubit dynamics (or separable); \emph{c)} most general non-controllable dynamics and \emph{d)} one-qubit controllable dynamics, where a extra parameter $t$ is provided by the user.}
\label{fig:circuit}
\end{center}
\end{figure}

\begin{restatable}{lemma}{unitaryequiv}
\label{lemma:unitaryequiv}
Consider two different generators for single qubit unitary operators $G$ and $G'$, and a unitary operator $W$.
\[	\forall G, G'; \exists W : W^\dagger G W = G'\]
\end{restatable}

\begin{restatable}[Interplay Dynamics $\leftrightarrow$ Initial State]{corollary}{unitaryequivalence}\label{thm:unitaryequivalence}
Consider two different unitary encodings, in terms of two different generators $U_\theta = e^{-i\theta G}$ and $U_\theta'= e^{-i\theta G'}$, and a unitary operator $W$.
\[	\forall U_\theta, U_\theta'; \exists W : W^\dagger U_\theta W = U_\theta'\]
\end{restatable}

Both Lemma \ref{lemma:unitaryequiv} and Corollary \ref{thm:unitaryequivalence} can be intuitively thought as changing the local definition of each qubit's basis from the regular basis $\{\ket{0}, \ket{1}\}$ to the eigenbasis of the generators $G$ given by $\{\ket{g}, \ket{-g}\}$, where $G\ket{\pm g} = \pm \ket{\pm g}$.
We provide proofs for the above lemma and corollary in Appendix~\ref{appendix:unitaryequivalence}. From Corol.~\ref{thm:unitaryequivalence} we get that changing the dynamics is equivalent to changing the initial state locally, and vice-versa. This in turn means that, for every separable dynamics we are always able to find a state with access to the same information, by applying a local unitary gate. This allows us to fix either one part or the other when trying to optimize the estimation. Because of this, we will often fix the local dynamics to a unitary of the type $e^{i\theta\sigma_j}$ where $\sigma_j$ is a Pauli operator (we will choose $Z$ whenever we need to), and say that the initial state attaining a certain estimation outcome is equivalent up to a certain local-unitary (LU) operation.

\subsection{Resources}\label{sec:resources}

Unlike the single parameter scenario of quantum sensing, in distributed sensing the resources are spatially distributed over a structure resembling a network, hence the name quantum sensor networks \cite{Proctor2017,Proctor2018d,Eldredge2018,Qian2020}. To convey their distribution over the nodes (or parties), it makes sense to have a notation that attributes a corresponding node to each qubit. The ensemble of qubits constitutes the set of resources for the estimation scenario, even though one could, equivalently, consider the amount of samples or duration of access to the encoding dynamics. 
The case where the encoding dynamics are not separable is more involved. 
We will comment further on this over Section \ref{section:generalhamiltonian}, but for now let us define the number of resources in terms of amount of the qubits available locally.

 Introducing the notation for the resources sets we will use throughout the paper: we denote by $\mathcal{N}_\mu$ as the set of qubits (or local resources) belonging to node $\mu$ (see Fig.~\ref{fig:network}), which amounts to the total set of qubits given by $\mathcal{N}= \cup_{\mu=1}^k \mathcal{N}_\mu$. For this reason, we denote $ \boldsymbol{\mathcal{N}} = \{\mathcal{N}_\mu\}_{\mu \in V}$ as a partition of the set, corresponding to the distributed resources. Moreover, the number of qubits in each set is given by $|\mathcal{N}_\mu| = n_\mu$, which means the vector $\vec{n} = (n_1,n_2,\cdots,n_k)$ is the vector of resources and $|\mathcal{N}|=\normone{\vec{n}} = n$ is the total amount of resources. Moreover, we will assume that the states can be picked from the complete Hilbert space of all qubits, which we denote by $\mathcal{H} (\mathcal{N})$.

Another important aspect that will be used throughout the paper, is to be able to order different vectors, such as the vector of resources. Since we are usually working with a vector-space defined over some ordered ring, we resort to the regular product order, which has appeared in quantum information under the name of majorization as well (e.g. \cite{Nielsen}), defined as follows:

\begin{definition}
Let $\mathbb{K}$ be a ordered-ring and $H \subseteq \mathbb{K}^k$. Let $\vec{a}, \vec{b}$ be two vectors belonging to $H$. Let $\leq$ be a total order of the field $\mathbb{K}$. The product order is defined as:
\begin{equation*}
	\vec{a} \preceq \vec{b} \qquad \text{iff. } \qquad a_j \leq b_j \ \forall j \in \{1,2,\cdots,k\}.
\end{equation*}
Moreover, we say $\vec{a} \prec \vec{b} $ if $\vec{a} \preceq \vec{b} $ and $\exists j: (\vec{a})_j < (\vec{b})_j$.
\end{definition}

Note that unlike $\leq$, $\preceq$ is not a total order, $i.e.$ there are vectors such that neither $\vec{a} \preceq \vec{b}$ nor $\vec{b} \preceq \vec{a}$. In this case we use the notation $\vec{a} \not \preceq \vec{b} $. 

\begin{example}
\[	 (2 , 4, 3) \preceq (2,4,3)  \quad, \quad (2 , 4, 3) \prec (2,6,4)  \quad , \quad  (1 , 3, 2) \not \preceq (3,2,4)
\]
\end{example}

\subsection{Functions of Parameters}

Before discussing the bounds on the precision and information metrics, let us first clear up the idea of what is the ``good'' definition for functions of parameters. The first thing to notice is the following: if one is able to estimate $f(\thetab)$, then one can also estimate $\alpha f(\thetab)$ for any $\alpha$. For this reason, let us establish an equivalence between functions: we say $f(\thetab) \sim g(\thetab)$ if $f(\thetab) = \alpha g(\thetab)$ for some $\alpha$. Since we are working with linear functions of parameters, the functions can be described as an inner-product between a vector of coefficients and a vector of the parameters $\thetab$: $f(\thetab) = \vec{a}\cdot \vec{\thetab}$, similarly to previous works \cite{Proctor2017,Shettell2022a,Rubio2020a,Eldredge2018}. 

As we will see in this paper, the linear coefficients available are somehow related with the eigenvalues of the Hamiltonian. For this reason, we start by finding a way to describe them, that is able to convey the distributed scenario we are working with. Taking the most general local Hamiltonians description in Eq.~\ref{eq:generaldynamics}, one can find an orthonormal basis for the local states by choosing the eigenvectors of the local Hamiltonians. Associated to each of these eigenvectors we then have their corresponding eigenvalues, which we can arrange as follows:
\begin{equation}
	\boldsymbol{H}_\mu \ket{\lambda_j^{\mu}} = \lambda_j^{\mu} \ket{\lambda_j^{\mu}}, \quad \mathcal{B}^\mu = \{ \ket{\lambda_j^\mu} \}_{i \in \mu} , \quad \mathcal{O}^\mu = \{\lambda_j^\mu \}_{j \in \mu },
\end{equation}
where $\mathcal{B}^\mu$ is a local basis for the qubits belonging to node $\mu$, and $\mathcal{O}^\mu $ is the list of non-identical eigenvalues for the local Hamiltonian of node $\mu$. Note that this list does not have the same size as the basis if the Hamiltonian is degenerate, although this will not pose any problem.

Going from the local Hamiltonians to the full network states can be easily done by choosing an orthonormal basis from the tensor product of the local basis: 
\begin{equation}
    \ket{\lambda_{\vec{j}}} = \Motimes_{\mu} \ket{\lambda_{j_\mu}^\mu} , \quad 	\boldsymbol{H}_\mu \ket{\lambda_{\vec{j}}} = \lambda_{j_\mu}^\mu \ket{\lambda_{\vec{j}}}, \quad \mathcal{B} = \Motimes_\mu \mathcal{B}^\mu , \quad \mathcal{O} = \prod_\mu\mathcal{O}^\mu, 
    \label{eq:hamiltonianbasis}
\end{equation}
where the vector $\vec{j}$ is a labeling vector such that each component $(\vec{j})_\mu \equiv j_\mu$ corresponds to one of the local basis states of node $\mu$. Moreover, the set $\mathcal{O}$ can be seen as a discrete subset of the $k$-orthotope $\overline{\mathcal{O}}$ defined by $[\lambda_{\min}^1,\lambda_{\max}^1] \times \cdots \times [\lambda_{\min}^k,\lambda_{\max}^k]$.

As we will verify further ahead, the natural functions one can estimate correspond to vectors living inside of the $k$-orthotope $\overline{\mathcal{O}}$. We will start by analyzing the case of separable dynamics and later generalize the most general encoding dynamics.
For  separable dynamics we have that:
\begin{enumerate}
    \item \textit{Controlled scenario}: $\vec{a} \in \mathbb{R}^k$ -- we are able to control the local encoding dynamics by an additional parameter $t$, $\Lambda = \Lambda(\theta,t) = \Lambda (\theta t) \sim e^{i\theta t \sigma_j} $ (see Fig.~\ref{fig:circuit} d);
    \item \textit{Non-controlled scenario}: $\vec{a} \in \mathbb{Z}^k $ -- only the local number of qubits per node or party is controllable ($\Lambda = \Lambda(\theta) \sim e^{i\theta \sigma_j} $).     
\end{enumerate}
Note that from the equivalence of functions introduced above, any $f(\thetab) = \vec{a}\cdot \vec{\thetab}$ such that $\vec{a} \in \mathbb{Q}^k $ is equivalent to some $g(\thetab) = \vec{a}'\cdot \vec{\thetab}$ such that $\vec{a}' \in \mathbb{Z}^k $. This extends the \textit{non-controlled scenario} to not only integers, but also fractionals. However, there will exist an overhead in the amount of resources. This overhead can be calculated from finding the smallest number $\alpha$ one has to multiply $\vec{a}$ such that $\alpha \vec{a} = \vec{a}' \in \mathbb{Z}^k$, which is given by the least common multiplier of all denominators. This depends entirely on the vector $\vec{a}$.

For the \textit{controlled scenario} one can establish an equivalence with the \textit{non-controlled}, by taking one qubit in every node ($\vec{n}=\vec{1}$) and control each local parameter $t_\mu$ in order to get $\vec{a}$. One can additionally change $\vec{n}$ and $\vec{t}$ such that $n_\mu t_\mu = a_\mu$, and get to the same result. This should provide that any statement made for the \textit{non-controlled scenario} should also verify for the \textit{controlled scenario}. On the other hand, in the \textit{non-controlled scenario}, we are limited by the number of qubits in each node, which only takes values over the integers. Taking the quotient of the equivalence relation, we get that, in this case, it suffices to analyse functions such that every element of $\vec{a}$ has no greatest common divisor than 1. For this we introduce the following definition:

\begin{definition}
Let $\vec{a}$ be a integer vector of size $k$, $\vec{a} \in \mathbb{Z}^k, \vec{a} = (a_1,a_2,\cdots,a_k)$. 
\begin{equation*}
	gcd(\vec{a}) = gcd(a_1,a_2,\cdots,a_k) = gcd(a_1,gcd(a_2,\dots,gcd(a_{k-1},a_k)))
\end{equation*}
Where $gcd$ is the greatest common divisor of a set of integer numbers.
\end{definition}
This means that the representative of each equivalence class of functions has $\vec{a}$ such that $gcd(\vec{a})=1$. Note the results for arbitrary $\vec{a}$ verifying $gcd(\vec{a})=1$ work backwards for $\vec{n}=\vec{1}$ in the \textit{controlled scenario}.

In the \textit{non-controlled scenario}, let $\vec{a}$ and $\vec{n} \in \mathbb{Z}^k$ be the vector of the target function and the vector of resources, respectively. From here on, when we refer to the vector of the target function $\vec{a}$, we will assume that $gcd(\vec{a})=1$. Using this ordering, let us define four different regions of amount of resources:

\begin{enumerate}[label=(\Roman*)]
	\item $\vec{n} \prec \vec{a}$ or $\vec{n} \not\preceq \vec{a} $ (No-Privacy Zone)
	\item $\vec{n} = \vec{a}$ (Minimal Privacy Zone)
	\item $\vec{a} \prec \vec{n}$ and ($ \vec{n}\prec 2\vec{a}$ or  $ \vec{n} \not\preceq 2\vec{a}$) (Minimal plus Ancilla Privacy Zone)
	\item $2\vec{a} \preceq \vec{n}$ (Multiple Privacy Zone)
\end{enumerate}
The labeling will become obvious as we will be able to prove different statements for each zone.

\subsection{Privacy Measure}

The privacy of a quantum estimation scenario, given a target linear function was introduced in \cite{Shettell2022a}. In a distributed estimation scenario, a private state allows each party to access the information given only by the target function. This means that while the parties can jointly estimate the target function, they cannot infer each other’s individual parameters. 

More generally, we consider a set of malicious parties, which we label  $D$ for dishonest, and the remaining set of parties is labeled $H$ for honest. The honest parties will follow the protocol as instructed, and the dishonest parties can work together, and along with the source, in any way they like (bounded only by quantum mechanics) to try to find out as much information as possible. Privacy for networked parameter estimation means that the dishonest set can only access the information of the agreed target function, and their local parameters (and any combinations thereof).
More concretely, this leads to the definition of privacy as:
\begin{enumerate}[label=(\roman*)]
	\item If all parties are honest, they can estimate $f(\thetab)$.
	\item For any set of dishonest parties $D$, the set $D$ can only know the function $f(\thetab)$ and all $\theta_\mu$ for $\mu \in D$ (and functions thereof).
\end{enumerate}

This definition is inspired from similar definitions for privacy in secure multiparty computation \cite{Shettell2022a}. Note that, as in multiparty computation, dishonest sets can always conspire to give a distorted value of the target function $f(\thetab)$. Though this possibly feels like a failure, there is little one could hope to do against that, as they are valid members of the network. For example, one could simply follow the protocol but add an extra phase on their system to distort the value of $f(\thetab)$. There would be no way of detecting this as different from the phase that was actually desired (see e.g. \cite{Shettell2022c}). What is important is that the dishonest parties do not get information about the honest parties parameters any more than is agreed on by the target global function.

To motivate the construction of the a privacy measure, let us first describe one property of the QFI matrix. Given that a linear function can be described as $f(\vec{\boldsymbol{\theta}}) = \vec{a} \cdot \vec{\boldsymbol{\theta}}$, one can always construct a non-unique orthonormal basis that includes the vector $\vec{a}/\norm{\vec{a}}$ and defines a unitary transformation. Call this matrix $A$, where every vector belongs to one orthonormal basis. Then:
\begin{equation}
\begin{aligned}
	\vec{\boldsymbol{\theta}}' &= A \  \vec{\boldsymbol{\theta}} \qquad \text{ with } \theta'_1 = \sum_i A_{1i} \theta_i = \frac{\vec{a}}{\norm{\vec{a}}} \cdot \vec{\boldsymbol{\theta}} \sim f(\vec{\boldsymbol{\theta}}).
\end{aligned}
\end{equation}

Since the QFI depends on the derivatives, which change through linear transformations, finding the new derivatives is straightforward:

\begin{equation}
\begin{aligned}
	\frac{\partial }{\partial \theta_i'} \rightarrow \sum_j \frac{\partial \theta_j}{\partial \theta_i'} \frac{\partial }{\partial \theta_j} \implies \nabla_{\vec{\boldsymbol{\theta}'}} \rightarrow A^T \cdot \nabla_{\vec{\boldsymbol{\theta}'}}.
\end{aligned}
\end{equation}
This in turn implies that the QFI changes in the following way when changing variables \cite{Paris2009,Liu2019}:
\begin{equation}
\begin{aligned}
\mathcal{Q'}_{ij} (\rho_{\thetab}) &= \int \frac{1}{\Tr [\Pi_\x \rho_\theta]} \partial_{\theta'_i} \left( \Tr [\Pi_\x \rho_\theta] \right) \partial_{\theta'_j} \left( \Tr [\Pi_\x \rho_\theta] \right) d\x \\
&  \downarrow\\
\mathcal{Q'}_{ij} (\rho_{\thetab}) &=\sum_{k,l} \int \frac{1}{\Tr [\Pi_\x \rho_\theta]} A_{ik}^T\partial_{\theta_k} \left( \Tr [\Pi_\x \rho_\theta] \right) A_{jl}^T \partial_{\theta_l} \left( \Tr [\Pi_\x \rho_\theta] \right) d\x  \\
&  \downarrow\\
\mathcal{Q}'(\rho_{\thetab})&= A^T \mathcal{Q} A.
\end{aligned}
\end{equation}

This has an interesting consequence: by calculating the QFI for the canonical $\boldsymbol{\theta}$ set of parameters, and then diagonalizing the matrix we are finding the linear transformation that retrieves the set of independent orthonormal parameters that the estimation scheme is sensitive to. This means the eigenvectors of the QFI matrix correspond to the natural set of functions available at hand given the estimation scenario, and their corresponding eigenvalues are the precision at which we can estimate them. Note however, that this does not mean that every function precision is always attainable, given the not-guaranteed commutation of the measurements which would allow one to saturate the Cramér-Rao bound \cite{Ragy2016}. Nonetheless, this allows to introduce a notion of privacy - having a rank-1 QFI matrix, with the only positive eigenvector spanning the space of the target function. Note that a rank-1 matrix is naturally non-invertible, making the multiparameter Crámer-Rao bound in Eq.~\ref{eq:cramermulti} ill-defined, as all the functions belonging to the null-space of the matrix would have an associated infinite variance. This aligns with the fact that these functions cannot be estimated, as they are not accessible. One could make use of a pseudo-inverse to address this fact, or transform the problem into a single-parameter scenario by applying a weight matrix in the vector direction \cite{Goldberg2021}. 
It is also important and useful to be able to quantify such a property - no state will perfectly satisfy such a condition due to inevitable noise, for example. It is then important to be able to say how close we are to being private. Although in \cite{Shettell2022a} privacy is quantified for the special case considered there, it is not general, not easily calculable, and its precise meaning in terms of the information leaked is not clear. 

With this in mind, we list some properties that are naturally desirable for a measure if one is to quantify how close one is to being perfectly private:
\begin{enumerate}
	\item $\mathcal{P} = \mathcal{P} (\mathcal{Q},\vec{a})$ where: \textit{(i)} $\mathcal{Q}$ is the QFI matrix with respect to parameters $\thetab = \{ \theta_1, \theta_2, ... ,\theta_k \} $, and provides the available information of the estimator and \textit{(ii)} $\vec{a}$ is the vector that defines the target function $f(\thetab) = \vec{a}\cdot\vec{\thetab}$.
	\item $\mathcal{P} (\mathcal{Q},\vec{a}) \in [0,1]$ to allow for normalization. The measure should depend not on the amount of information, nor the target vector norm, but on how the information is aligned with the target function.
	\item $\mathcal{P} (\mathcal{Q},\vec{a}) = 1$ iff. $\mathcal{Q} = \alpha \vec{a}\vec{a}^T$, meaning if and only if the information is aligned with the target function, then a state is completely private, with respect to that same function.
	\item $\mathcal{P} (B\mathcal{Q}B^{T},B\vec{a}) = \mathcal{P} (\mathcal{Q},\vec{a})$ as changing the basis should be equivalent to changing the target function. This is motivated by linear transformations of the parameters not changing the information available, in the rotated corresponding linear function.
	\item $\forall \epsilon > 0, A \geq 0, \  \exists \delta > 0 : | \mathcal{P} (\mathcal{Q}+\epsilon A,\vec{a}) - \mathcal{P} (\mathcal{Q},\vec{a})| \leq \delta$ to ensure that continuity of the QFI implies the continuity of the privacy.
\end{enumerate}
Property 1 and 3 ensure that the information available of a quantum estimation procedure (provided by $\mathcal{Q}$) align solely with the target function. Property 2 is a matter of defining the values this measure can take. Property 4 asserts that any linear reassignment of the parameters does not change the privacy statement for the new parameters. Finally, property 5 allows for a statement of the continuity of quantum states implying continuity of the privacy measure, as we know that if a quantum state is close to another, so is the QFI matrix \cite{Rezakhani2019}. This is a set of necessary properties for a measure able to characterize and distinguish different quantum estimation procedures, with respect to the linear subspace of the obtained information. We now present our privacy measure, which satisfies these properties:
\begin{definition}\label{def:privacymeasure}
The privacy measure of a multi-parameter estimation problem, which results in a quantum Fisher information matrix $\mathcal{Q}$, with respect to a target linear function $f(\vec{\thetab})= \vec{a}\cdot\vec{\thetab}$ with $||\vec{a}||=1$ is given by:
\begin{equation*}
	\mathcal{P}(\mathcal{Q},\vec{a}) = \frac{ \vec{a}^T \mathcal{Q} \vec{a}}{\Tr \mathcal{Q}} \equiv \frac{\Tr \left[ \mathcal{Q} \vec{a}\vec{a}^T \right]}{\Tr \mathcal{Q}} = \frac{\Tr \left[ \mathcal{Q} W_{\vec{a}} \right]}{\Tr \mathcal{Q}}.
\end{equation*}
In case $||\vec{a}|| \neq 1$, simply redefine $\vec{a} = \vec{a}/||\vec{a}||$.
\end{definition}

In Appendix~\ref{appendix:privacymeasure} we prove that the privacy measure in Def.~\ref{def:privacymeasure} verifies all the proposed properties, and address as well the meaning behind the numerical value of $\mathcal{P}$. Nonetheless, we provide some intuition below: 
\begin{enumerate}
	\item $\mathcal{P}=1$, means complete privacy, and in particular a state that verifies this is guaranteed to satisfy all conditions. 
	\item For $\mathcal{P}=0$, this implies there is a function of parameters, orthogonal to the target function $f(\thetab)$, that the dishonest parties can access optimally, which by the privacy condition (iii), the dishonest parties should have no information of. On the other hand, if everybody is honest, no information about $f(\thetab)=\vec{a}\cdot\vec{\thetab}$ could possibly be obtained.
	\item If $\mathcal{P}=\epsilon$, then one can do the same as above, but less efficiently, $i.e.$ the best strategy can only gather less information about the same function than when $\mathcal{P}=0$. In particular, the amount of information accessible decreases with the increase of $\epsilon$, which is equivalent to needing more repetitions to get the same variance on the information available.  
\end{enumerate}

We note here that in our setting the dynamics are part of the problem setting, so that privacy above depends only on the state, hence the notion of private state. 
The idea of the private sensing protocols of \cite{Shettell2022a} is to certify that one possesses such a private state, then use it for sensing, in this way ensuring privacy. 
Our definition of privacy contains several implications for a distributed secure sensing protocol. As a consequence of the definition of privacy, if one certifies a private state, then one is able to guarantee that a private estimation protocol is also completely secure \cite{Shettell2022a}. All the private states we find here can be certified using the techniques of \cite{Shettell2022a,Unnikrishnan2022}, and so can be included into protocols, guaranteeing private sensing.

\section{Building Private States \label{section:buildingprivatestates}}

To build private states, we start by developing an efficient way to characterize the QFI matrix as a function of the input state and fixed separable and parallel dynamics, first for stabilizer states using the stabilizer formalism, and next taking an arbitrary state as input, by choosing an appropriate basis. Using Corol.~\ref{thm:unitaryequivalence} this generalizes to any local single-qubit (or separable) dynamics, up to some local single-qubit operation. Later we will extend into the case of the most general dynamics, by finding another basis with respect to the more general encoding dynamics. Using this, and applying the privacy definition, we are able to prove which states and superpositions of states will be private ($\mathcal{P} = 1$), allowing us to create families of private states. Moreover, we do this in function of the amount of resources, as one will see is key to being able to find private states.

\subsection{Initial Concepts}
Before diving into building private states, let us introduce some well-known concepts and adaptations of them, that allow us to build the tools for our proofs of privacy. Starting with the Hamming-weight, a function that, given a bit string, outputs the number of ones of the string:

\begin{definition}
Let $s$ be a string of bits of size $n$, $s \in \mathbb{F}_2^n$. The Hamming weight of $s$, $h(s)$ is given by:
\begin{equation*}
	h(s) = \sum_{j=1}^n s_j  .
\end{equation*}
\end{definition}

Note, one key property of the Hamming-weight is its invariance under permutations of bits within the string $s$:
\begin{proposition}\label{prop:invarianceperm}
Let $X = \{1,2,...,n\}$ be the set of positions of a bit string $s$ of size $n$. Let $S_X$ be the permutation group over $X$.
\begin{equation*}
\begin{aligned}
	h(s) = h(\sigma(s)) \quad	\forall \sigma \in S_X, \forall s
\end{aligned}
\end{equation*}
\end{proposition}
This in fact allows one to create equivalence classes on binary strings for the Hamming-weight, as permutations generate the group that results in the invariance of the Hamming weight:
\begin{center}
\begin{tikzcd}[column sep=small]
 s  \in  \hspace{-0.5cm} &\mathbb{F}_2^n \arrow[rrd, "h"] \arrow[d, "\sigma"'] &  & \\
 \text{[}s \text{]} \in  \hspace{-0.5cm} &\mathbb{F}_2^n/S_X \arrow[rr, "\tilde{h}"'] & &  \mathbb{Z}
\end{tikzcd}
\end{center}
Where $\tilde{h}$ is an isomorphism. This means that doing the quotient $\mathbb{F}_2^n$ over $S_X$ will create the equivalent classes under the Hamming-weight. We denote them $[s]$, and they constitute the objects in $\mathbb{F}_2^n / S_{X}$ which are unequivocally mapped to a weight under the Hamming-weight map. For us it is also useful to define a slight variant of the Hamming weight, which we call the symmetrized Hamming-weight:

\begin{definition}
Let $s$ be a string of bits of size $n$, $s \in \mathbb{F}_2^n$. The symmetrized Hamming weight of $s$, $h^*(s)$ is given by:
\begin{equation*}
	h^*(s) = \sum_{j=1}^n (-1)^{s_j}  .
\end{equation*}
\end{definition}

Moreover, these two weight-functions are related via a linear transformation:
\begin{equation}
	h^* (s) = n - 2\cdot h(s),
\end{equation}
where $n$ is the integer size of the set of indices (or the length of the bit string). Because we are in a multiparameter scenario, where the parties have access to different amounts of resources, let us define a new vectorial Hamming-weight. We call it vectorial since in this case, it outputs a vector of Hamming-weights. Using the notation introduced before for the partition of the set of resources:

\begin{definition}\label{def:vectorialhammingweight}
Let $s$ be a string of bits of size $n$, $s \in \mathbb{F}_2^n$. Let $\boldsymbol{\mathcal{N}}= \mathcal{N}_1 \cup \cdots \mathcal{N}_k$ be a size $k$ partition over the set of indices of the bit string $s$, such that $n_j = |\mathcal{N}_j|$, and $\sum_{j}n_j  = n$. Define the $\boldsymbol{\mathcal{N}}$-Hamming weight of $s$, $\vec{h}_{\boldsymbol{\mathcal{N}}}(s)$ by:
\begin{equation*}
	\left( \vec{h}_{\boldsymbol{\mathcal{N}}}(s) \right)_j =  \sum_{i \in \mathcal{N}_j} s_{i},
\end{equation*}
where $j$ runs from $1$ to $k$, meaning $\vec{h}_{\boldsymbol{\mathcal{N}}}(s) \in \mathbb{Z}^k$.
\end{definition}
In the same way one can define its symmetrized counterpart and rewrite the linear relation between them:
\begin{equation}
	\vec{h}^*_{\boldsymbol{\mathcal{N}}} (s) = \vec{n} - 2\cdot \vec{h}_{\boldsymbol{\mathcal{N}}} (s).
\end{equation}

What before was invariance under any permutation of the bit string, now becomes only under a restricted set of permutations. Let us define this subset and prove that the vectorial Hamming weight is invariant under this subset of permutations.

\begin{definition}\label{def:permutationsvector}
Let $S_X$ be the permutations group over a set $X$. Moreover, let $\boldsymbol{X} = \{X_j\}_{j=1,2,\dots,k}$ be a partition of $X$. Define $S_{\boldsymbol{X}} \subseteq S_X$ as:
\begin{equation*}
\begin{aligned}
	S_{\boldsymbol{X}} = S_{X_1} \times S_{X_2} \times \cdots \times S_{X_k}.
\end{aligned}
\end{equation*}
\end{definition}

\begin{proposition}\label{prop:invariancepermS}
Let $X$ be a set and $\boldsymbol{X}= \{X_j\}_{j=1,2,\dots,k}$ be a partition of $X$. Let $S_X$ be the permutation group over $X$ and $S_{\boldsymbol{X}}$ the subset of it defined according to Def.~\ref{def:permutationsvector}.
\begin{equation*}
\begin{aligned}
	\vec{h}_{\boldsymbol{X}} (s) = \vec{h}_{\boldsymbol{X}} (\sigma(s)) \quad	\forall \sigma \in S_{\boldsymbol{X}}
\end{aligned}
\end{equation*}
\end{proposition}

Prop.~\ref{prop:invarianceperm} is simply a consequence of the commutativity of the sum. The proof of Prop.~\ref{prop:invariancepermS} is a direct consequence of Prop.~\ref{prop:invarianceperm} into each of the subsets of the partition.  Similarly to before, one can create the equivalence classes of the vectorial Hamming weight. Let us denote them $[s]_{\boldsymbol{X}}$, which are again the objects in $\mathbb{F}_2^n / S_{\boldsymbol{X}}$ which are unequivocally mapped to a weight under the vectorial Hamming weight map. 
\begin{center}
\begin{tikzcd}[column sep=small]
s \in \hspace{-0.5cm}&\mathbb{F}_2^n \arrow[rrd, "\vec{h}_{\boldsymbol{X}}"] \arrow[d, "\sigma"'] &  & \\
 \text{[}s \text{]}_{\boldsymbol{X}} \in  \hspace{-0.5cm} &\mathbb{F}_2^n/S_{\boldsymbol{X}} \arrow[rr, "\tilde{h}_{\boldsymbol{X}}"'] & &  \mathbb{Z}^k
\end{tikzcd}
\end{center}

As an example of using the vectorial Hamming weight and its equivalence classes, consider the following bit strings:

\begin{example}
Let $\vec{n} = \{ 1,3,2 \}$ be a vector in $\mathbb{N}^3$ and $s_1,s_2,s_3 \in \mathbb{F}^{6}$ be given by $s_1 =1\ 110 \ 10,s_2 = 1 \ 101 \ 10, s_3 = 1 \  011 \ 01$.
\begin{equation*}
\begin{aligned}
	&\vec{h}_{\boldsymbol{\mathcal{N}}}(s_1) = \vec{h}_{\boldsymbol{\mathcal{N}}}(\underbrace{1}_{\mathcal{N}_1} \underbrace{110}_{\mathcal{N}_2} \underbrace{10}_{\mathcal{N}_3})  =  \{ h(1),h(110),h(10) \} = \{ 1, 2, 1 \} \\
	&\vec{h}_{\boldsymbol{\mathcal{N}}}(s_1) =\vec{h}_{\boldsymbol{\mathcal{N}}}(s_2) =\vec{h}_{\boldsymbol{\mathcal{N}}}(s_3) \implies [s_1]_{\boldsymbol{\mathcal{N}}} = [s_2]_{\boldsymbol{\mathcal{N}}} = [s_3]_{\boldsymbol{\mathcal{N}}}
\end{aligned}
\end{equation*}
This means that all bit strings $s_1,s_2,s_3$ belong to the same equivalence class. Moreover, as a direct consequence of this, there exists necessarily a permutation in $S_{\boldsymbol{\mathcal{N}}}$ that transforms each one into each other one.
\end{example}

One can then show that any bit string $s$, indexed by a set $\mathcal{N}$, with partition $\boldsymbol{\mathcal{N}}$ with correspondent resource vector $\vec{n}$ has a $\boldsymbol{\mathcal{N}}$-Hamming weight bounded by:
\begin{equation}
\begin{aligned}
	\vec{0} \preceq \vec{h}_{\boldsymbol{\mathcal{N}}} (s) \preceq \vec{n} \qquad \implies \qquad -\vec{n} \preceq \vec{h}^*_{\boldsymbol{\mathcal{N}}} (s) \preceq \vec{n}.
\end{aligned}
\end{equation}

We also have the following consequences for vectors of integers, which we prove in Appendix~\ref{appendix:integers} and that will later be crucial:
\begin{restatable}{proposition}{gcdmultiple}\label{prop:gcdmultiple}
$\forall \ \vec{a},\vec{b} \in \mathbb{Z}^k $ such that $gcd(\vec{a})=1$ then, if $ \vec{b} = \alpha \vec{a} $ implies that $ \alpha \in \mathbb{Z}$.
\end{restatable}
\begin{restatable}{proposition}{gcdorder}\label{thm:gcdorder}
Let $\vec{m}$ be a integer vector in $\mathbb{Z}_+^k$ such that $gcd(\vec{m})=1$. Then:
\begin{equation*}
	\nexists \ \vec{a} \neq \vec{b} \in \mathbb{Z}_+^k,  \vec{a} , \vec{b}\prec \vec{m}: (\vec{m} -2\vec{a}) \pm (\vec{m} -2\vec{b}) = \alpha \vec{m} \quad , \quad \alpha \neq 0.
\end{equation*}
\end{restatable}

\begin{restatable}{proposition}{gcdorderm}\label{thm:gcdorder2}
Let $\vec{m}$ be a target vector in $\mathbb{Z}_+^k$ such that $gcd(\vec{m})=1$. Let $\vec{n}$ be a resource vector in the same space, such that $\vec{n} \not \preceq \vec{m}$. Then:
\begin{equation*}
	\nexists \ \vec{a} \neq \vec{b} \in \mathbb{Z}_+^k,  \vec{a} , \vec{b} \prec \vec{n}: (\vec{n} -2\vec{a}) \pm (\vec{n} -2\vec{b}) = \alpha \vec{m} \quad , \quad \alpha \neq 0.
\end{equation*}
\end{restatable}

All of these concepts we have introduced focus on describing the functions available by employing integer resources in an estimation scenario. The manipulation of the Hamming weight into its symmetrized counterpart is a direct consequence (together with Corol.~\ref{thm:unitaryequivalence}) of being able to choose any local basis for each qubit. If one is able to choose any basis for one qubit, then each basis will have two possible qubit eigenstates, which can be characterized by a bit (0 or 1), and will have eigenvalues $\pm 1$. Then, each eigenstate of the entire state in such basis can be labeled by a bit string of the size of the number of qubits. Since we choose the basis in the same "direction" as the individual Hamiltonians, then the eigenvectors of the local Hamiltonian will be given by the eigenstates created. Moreover, and as a consequence of the $\pm 1$ eigenvalues of each qubit Hamiltonian, each eigenstate characterized by the bit string $s$ will have a correspondent eigenvalue given by the symmetrized Hamming weight of that bit string. We will use this to prove then how to find the private states, and build the entire set of private states.

\subsection{Privacy for Stabilizer States with Separable Hamiltonians}

Stabilizer formalism is a very useful formalism to describe quantum states in terms of their symmetries \cite{Hein2004b,Hein2006,Aaronson2004}. As the name suggests, it involves the notion of a stabilizer, an operator that stabilizes a state ($S \ket{\psi} = \ket{\psi}$). A $n$-qubit state is called a stabilizer state if it is stabilized by a set of $n$ stabilizers $s \in \mathcal{S}$ such that:

\begin{enumerate}[label=(\roman*)]
	\item $[s_i,s_j]=0, \forall s_i,s_j \in \mathcal{S}$ (Abelian),
	\item $\forall s_i,s_j, s_k \in \mathcal{S}: s_i \circ s_j \neq s_k$ (Linear Independence),
	\item $s_i \in \{\pm 1,\pm i\} \mathcal{P}^n$ such that $s_i \neq \id$ (Pauli strings).
\end{enumerate}
We will denote the ensemble of stabilizer states in the Hilbert space $\mathcal{H}$ by $\mathsf{Stab}(\mathcal{H})$. Given these conditions, the group generated by $\mathcal{S}$ is the complete set of stabilizers of the state. We call this group $\boldsymbol{S} = \langle \mathcal{S} \rangle$. It is also known that:

\begin{align}
	\ketbra{\psi}{\psi} = \frac{1}{2^n} \sum_{s \in \boldsymbol{S}} s = \frac{1}{2^n} \prod_{s_i \in \mathcal{S}} (\mathbb{1} + s_i).
	\label{eq:stabdescription}
\end{align}

Moreover, there are multiple representations for the generators of a state. We present an example of GHZ state with three qubits, in different representations, namely a stabilizer table representation and the binary sympletic representation:
\begin{equation}
\begin{tabu}{ c | c | c}
\boldsymbol{q}_1 & \boldsymbol{q}_2 & \boldsymbol{q}_3 \\ \hline \hline
X_1 & X_2 & X_3 \\
Z_1 & Z_2 & \mathbb{1} \\
Z_1 & \mathbb{1} & Z_3 
\end{tabu} \qquad \text{ or } \qquad (\boldsymbol{X}|\boldsymbol{Z}) = 
\left(
\begin{tabular}{ c c c | c c c}
    1 & 1 & 1 & 0 & 0 & 0 \\
    0 & 0 & 0 & 1 & 1 & 0 \\
    0 & 0 & 0 & 1 & 0 & 1
\end{tabular}
\right).
\label{eq:sympletic}
\end{equation}

One can easily check that a star-graph centered on the first node is related with the GHZ state of three qubits via Hadamard gates being applied to both leafs of the star-graph, $i.e.$ nodes 2 and 3. Doing so is equivalent to applying Hadamard gates to $\boldsymbol{q}_2$ and $\boldsymbol{q}_3$:

\begin{align*}
	&\ket{GHZ_3} \xrightarrow{H_2 H_3} \ket{S_3} \\ \\
	\begin{tabu}{ c | c | c}
	\boldsymbol{q}_1 & \boldsymbol{q}_2 & \boldsymbol{q}_3 \\ \hline \hline
	X_1 & X_2 & X_3 \\
	Z_1 & Z_2 & \mathbb{1} \\
	Z_1 & \mathbb{1} & Z_3 
	\end{tabu}
	\xrightarrow{H_2 H_3}
	&
	= \begin{tabu}{ c | c | c}
	\boldsymbol{q}_1 & \boldsymbol{q}_2 & \boldsymbol{q}_3 \\ \hline \hline
	X_1 & Z_2 & Z_2 \\
	Z_1 & X_2 & \mathbb{1} \\
	Z_1 & \mathbb{1} & X_3 
	\end{tabu} 
    \equiv \left(
    \begin{tabular}{ c c c | c c c}
        1 & 0 & 0 & 0 & 1 & 1 \\
        0 & 1 & 0 & 1 & 0 & 0 \\
        0 & 0 & 1 & 1 & 0 & 0
    \end{tabular}
    \right),
\end{align*}
where the binary sympletic representation becomes $(\boldsymbol{X}|\boldsymbol{Z}) = (\mathbb{1}|\boldsymbol{\Gamma})$, with $\boldsymbol{\Gamma}$ being the adjacency matrix of the star graph. In general, one can always do this for a graph state. 

Using the stabilizer formalism, one can then find an expression for the QFI in terms of the number of stabilizers \cite{Shettell2022,Shettell2020c,Tao2023}. Under the assumptions in Eq.~\ref{eq:dynamics}:

\begin{equation}
\begin{aligned}
	\ket{\psi_{\thetab}} &= U_{\thetab} \ket{\psi} = \Motimes_{\mu \in V} e^{-i \theta_\mu \boldsymbol{G}_\mu}  \ket{\psi} , \quad\ket{\partial_{\theta_\mu} \psi_{\thetab}} = -i \boldsymbol{G}_\mu U_{\thetab} \ket{\psi}, \\
	\mathcal{Q}_{\mu \nu} (\rho_{\thetab}) &=4 \re{ \braket{\partial_{\theta_\mu} \psi_{\thetab} | \partial_{\theta_\nu} \psi_{\thetab} }  - \braket{ \partial_{\theta_\mu} \psi_{\thetab} | \psi_{\thetab}}\braket{ \psi_{\thetab} | \partial_{\theta_\nu} \psi_{\thetab}}  } \\
	&= 4 \re{ \braket{ \psi | \boldsymbol{G}_\mu \boldsymbol{G}_\nu | \psi } - \braket{ \psi | \boldsymbol{G}_\mu | \psi }\braket{ \psi | \boldsymbol{G}_\nu | \psi }  } \\
	&= 4 \sum_{j \in \mu} \sum_{k \in \nu} \re{ \braket{ \psi | G_j G_k | \psi } - \braket{ \psi | G_j | \psi }\braket{ \psi | G_k | \psi }  } \\
    &= 4 \sum_{j \in \mu} \sum_{k \in \nu} \id_{\mathcal{S}_{\psi}} (G_j G_k) - \id_{\mathcal{S}_{\psi}} (G_j)\id_{\mathcal{S}_{\psi}} (G_k),
\end{aligned}
\label{eq:qfistab}
\end{equation}
where $ \id_{\mathcal{S}_{\psi}} (G)$ is the indicator function over the stabilizer set of $\ket{\psi}$, $\mathcal{S}_\psi$, evaluating to 1 if $G \in \mathcal{S}_\psi$ and 0 otherwise.
From here we can use the stabilizer formalism to understand which terms come out different from zero. There are two cases one should analyze:
\begin{enumerate} 
	\item If both $a G_j$ and $b G_k$, where $a,b \in \{ \pm 1, \pm i\}$, are stabilizers of the state $\psi$, then so is $ab G_j G_k$, making its contribution to the QFI null;
	\item The only scenario for which the QFI is different from zero is the case where $\pm G_j G_k$ is a stabilizer (note the real part in Eq.~\ref{eq:qfistab}), but neither $a G_j$ nor $b G_k$ are stabilizers of the state.
\end{enumerate}
The case where either $a G_j$ or $b G_k$ is a stabilizer and $\pm G_j G_k$ is also a stabilizer is not possible, as it would imply all of them are stabilizers.
The QFI then becomes a combinatorial problem. The generators $G_j$ can always be made up such that $G_j^{p} = G_j^{p\mod 2}$, which is the same type of structure when working with the vector space of $\mathbb{F}_2^n$, where the sum is already made modulus 2. Suppose the following set of generators:
\begin{equation}
	g(\vec{l}) = G_1^{l_1} \otimes G_2^{l_2} \otimes \dots \otimes G_n^{l_n},
	\label{eq:set1}
\end{equation}
where $\vec{l} = (l_1, l_2, ... ,l_n)$ is a vector in $\mathbb{F}_2^n$, made up of zeros and ones. Using this, we can define the trivial base of $n$ vectors that generate all vectors of $\mathbb{F}_2^n$ by simply:
\begin{equation}
\mathfrak{B}_e
\begin{cases}
	e_1 &= (1, 0, 0, ...,0,  0) \\
	e_2 &= (0, 1, 0, ..., 0, 0) \\
	& \ \  \vdots \\
	e_n &= (0, 0, 0, ...,0 , 1) 
\end{cases}.
\label{eq:basisfield}
\end{equation}
From the vector space characteristics, the choice of basis is not unique. This means that taking $\{ g(\vec{e}_j) \}_{\vec{e}_j \in \mathfrak{B}_e}$ for our stabilizer generators is equivalent to taking any other basis $g(\vec{\tilde{e}}_j)$, resulting in all terms $G_j$ and $G_j G_k$ being part of the stabilizers of the state. Note that, from its construction using only $G_j$ and $\mathbb{1}$, the commutation rule for the set of generators is already verified. This means that:
\begin{proposition}
Any generator set of stabilizers in the form of Eq. \ref{eq:set1} with an associated complete basis for the binary vectors defines a state with no information retrievable.
\label{prop:privatebasis}
\end{proposition}

\begin{proof}
For a generator set to correspond to a state with information, it should contain terms of the form $G_j G_k$ but not the terms of the form $G_j$ nor $G_k$. Since there is an isomorphism between the group of stabilizers $\mathcal{S}$ and the $\mathbb{F}_2^n$, any linear independent set of vectors in $\mathbb{F}_2^n$ defines a linear independent set of generators for $\mathcal{S}$. Moreover, taking whichever complete linearly independent basis for $\mathbb{F}_2^n$ we verify that its span includes all bit strings with length $n$. This in turn implies that whichever basis we choose for Eq. \ref{eq:set1}, it will include both terms of the form $G_j G_k$ and $G_j$ or $G_k$, resulting in a state with no information.
\end{proof}

This allows us to prove which states are able to remain private for estimating a family of functions. The following two theorems essentially tell us that the only private stabiliser state resource is locally equivalent to the GHZ state, and the distributed resources (number of qubits per network node) is fixed by the function. We start by proving there is a zone of resources where no private state can be built:

\begin{theorem}[No-Privacy Zone for Stabilizer States]
Let $\vec{a}, \vec{n} \in \mathbb{Z}^k$, $n = \norm{\vec{n}}_1$, $gcd(\vec{a})=1$ and $f(\thetab)$ be the linear function given by $f(\thetab) = \vec{a}\cdot \vec{\thetab}$. Then, there is no stabilizer state capable of estimating $f$ privately with $\vec{n}$ distributed resources and access to local separable dynamics $U_{\thetab}$, such that $\vec{n} \prec \vec{a}$ or $\vec{n} \not\preceq \vec{a}$.
\begin{equation*}
	\nexists \ket{\psi} \in \mathsf{Stab}(\mathcal{H}(n)): \mathcal{P} (\mathcal{Q} (U_{\thetab}\ket{\psi}),\vec{a}) = 1
\end{equation*}
\label{thm:noprivacystab}
\end{theorem}
Note we have only considered the amount of qubits per node, and not the more general qubits distribution using the partition notation $\boldsymbol{\mathcal{N}}$. We follow with the proof:
\begin{proof}
Privacy, $\mathcal{P}=1$, is equivalent to having a QFI matrix given by $\boldsymbol{\mathcal{Q}} = \lambda_1 \cdot \vec{a}\vec{a}^T$
where we can observe a rank one matrix in the direction of $\vec{a}$.
Given that the number of operators $G_j$ for each node $\mu$ is an integer number bounded by the amount of qubits in that same node $n_\mu$, the maximum of $\boldsymbol{\mathcal{Q}}_{\mu\nu}$ is given by $n_\mu n_\nu$. From the conditions stated one can say that $\exists \mu: n_\mu < a_\mu$. However, by construction $a_\mu$ has no common divisors but 1 with all the other $a_\nu$. Since Prop.~\ref{prop:gcdmultiple} only allows $\lambda_1$ to be integer, one could never achieve $\vec{a}\vec{a}^T$. For this reason, whichever basis ends up maximizing the privacy of parameter $\vec{a}\cdot \vec{\thetab}$ has a projection onto $\vec{a}$ plus a non-null projection onto at least another linearly independent vector, making it not completely private.
\end{proof}

This already states we require a minimum amount of resources per node to achieve privacy. Next we prove that the first private state requires at least $\vec{n}=\vec{a}$ amount of resources:

\begin{theorem}[Minimal Privacy Zone for Stabilizer States]
Let $\vec{a}, \vec{n} \in \mathbb{Z}^k$, $gcd(\vec{a})=1$ and $f(\thetab)$ be the linear function given by $f(\thetab) = \vec{a}\cdot \vec{\thetab}$. Then, there is only one stabilizer state, up to LU operations, capable of estimating $f$ privately with $\vec{n} =  \vec{a}$ distributed resources and access to local separable dynamics $U_{\thetab}$.
\label{thm:oneprivacystab}
\end{theorem}

\begin{proof}
Take Corol.~\ref{thm:unitaryequivalence} to deal with LU and use $G_j,G_j^\perp$ as the non-commuting dual of generators for each qubit. Given that $\vec{n}=\vec{a}$, then the only possibility to be private is if every $G_j G_k \in \boldsymbol{\mathcal{S}}$, as one can only find integer multiples of $\vec{a}$ when $gcd(\vec{a})=1$ (Prop.~\ref{prop:gcdmultiple}). 
Consider the table of stabilizers given by:
\begin{equation}
	\mathcal{S} =
    \left(\begin{tabu}{c}
        \boldsymbol{G} \\ 
        \boldsymbol{G}^\perp
    \end{tabu}\right) = \left(
    \begin{tabu}{ccccc}
    G_1 & \mathbb{1} & \mathbb{1} & \cdots & G_n \\
    \mathbb{1} & G_2 & \mathbb{1}  & \cdots & G_n \\
    \vdots & \vdots & & \ddots & \vdots \\
    \mathbb{1} & \mathbb{1} & \mathbb{1} & G_{n-1} & G_n \\ \hline
    G_1^\perp & G_2^\perp & G_3^\perp & \cdots & G_n^\perp \\
    \end{tabu} \right).
\label{eq:tablegenerators1}
\end{equation}
Note the choice of $\boldsymbol{G}$ is not unique, but if generated via Eq.~\ref{eq:set1} we cannot use a complete basis (see Prop.~\ref{prop:privatebasis}). By enforcing condition 2, one may arrive at $\boldsymbol{G}$, given that all $G_j G_k \in \langle\boldsymbol{G}\rangle$ and no $G_j \in \langle\boldsymbol{G}\rangle$. If one would add $G_j \notin \boldsymbol{G}$, then $\exists G_j G_k \in \langle \boldsymbol{G} \rangle: G_j G_k \circ G_j = G_k \ \forall \ k$ which means a state with zero information. If one would remove a stabilizer from $\boldsymbol{G}$, then $\exists G_p: G_p G_q \notin \boldsymbol{\mathcal{S}} \  \forall \ q$, meaning a non private state. The only option is then to add a new stabilizer that commutes with every other in $\boldsymbol{G}$ and is linearly independent of them. This can only happen by choosing a different orthogonal generator and every element of the generator must be the same, hence: $\boldsymbol{G}^\perp = a G_1^\perp G_2^\perp \cdots G_n^\perp$, $a \in \{\pm 1,\pm i\}$. Up to unitary equivalence (see Corol.~\ref{thm:unitaryequivalence}), this defines one state (up to LU) which is able to remain private.
\end{proof}

Note that this state is LU equivalent to the GHZ state. Additionally, our statements are only valid if our stabilizers and generators $G_j,G^\perp_j$ belong to the Pauli strings, as it is required by the stabilizer definition. With the help of Corol.~\ref{thm:unitaryequivalence} one can always work with $G_j = Z, G_j^\perp = X$ and then up to LU, not necessarily still a stabilizer state, the statements above are true.

\subsection{Privacy for Arbitrary States with Separable Hamiltonians}

To build an arbitrary state, one needs an orthonormal basis of quantum states. The choice of basis is not fixed, which allows to choose whichever preferred basis. If one chooses the basis that contains a state that is LU equivalent to a GHZ state, the calculations become rather straightforward. For this, one can pick the generator set in Eq.~\ref{eq:tablegenerators1} as a representation of our state which is LU equivalent to a GHZ state. Since we are no longer working with just stabilizer states, let $G_j, G^\perp_j$ have arbitrary definitions $G_j = \vec{a}\cdot\vec{\sigma}, G^\perp_j = \vec{b}\cdot\vec{\sigma}$ such that $\norm{\vec{a}} = \norm{\vec{b}} = 1$ and $\vec{a}\cdot \vec{b} = 0$. In particular, one can and should choose the generators of the encoding dynamics, meaning that the $G_j$s should be local generators in Eq.~\ref{eq:dynamics}. Call this state, represented by the stabilizer table in Eq.~\ref{eq:tablegenerators1}, $ \ket{\mathcal{G}_1}$.

Then, one can think of a generator set of a state as a $n$ subspace of $\mathbb{F}_2^{2n}$ with the help of the symplectic notation introduced in Eq.~\ref{eq:sympletic}. This means there still exists another subspace of dimension $n$, orthogonal to the first. One can use this orthogonal subspace to generate the set of orthogonal states, and the orthogonality conditions will appear naturally. In a more concrete way, take the following generator set:

\begin{equation}
    \mathcal{S}^\perp = \left(
    \begin{tabu}{ccccc}
    G_1^\perp & \mathbb{1} & \cdots & \mathbb{1} & \mathbb{1} \\
    \mathbb{1} & G_2^\perp & \cdots & \mathbb{1} & \mathbb{1} \\
    \vdots & \vdots & \ddots & \vdots & \vdots \\
    \mathbb{1} & \mathbb{1} & \cdots & G_{n-1}^\perp & \mathbb{1} \\
    \mathbb{1} & \mathbb{1} & \cdots & \mathbb{1} & G_{n} 
    \end{tabu}
    \right).
\end{equation}

This means one can describe an orthornormal basis as $\{\ket{\mathcal{G}_j} = \tilde{s}_j \ket{\mathcal{G}_1}\}_{\tilde{s}_j \in \langle \mathcal{S}^\perp \rangle} $ with $2^n$ elements, as $|\langle \mathcal{S}^\perp \rangle| = 2^n$. To prove orthonormality take:

\begin{equation}
\begin{aligned}
	 \braket{\mathcal{G}_i |\mathcal{G}_j } &=  \bra{\mathcal{G}_1} \tilde{s}_i  \tilde{s}_j \ket{\mathcal{G}_1} = \Tr [ \tilde{s}_i  \tilde{s}_j \ket{\mathcal{G}_1} \bra{\mathcal{G}_1} ] = \frac{1}{2^N} \sum_{s_k \in \langle \mathcal{S} \rangle} \Tr [ \tilde{s}_i  \tilde{s}_j s_k ] = \delta_{ij} .
\end{aligned}
\end{equation}
We now decompose an arbitrary state $|\psi\rangle$ over our resources in this basis 
\begin{equation}
	\ket{\psi} = \sum_j \alpha_j \ket{\mathcal{G}_j}.   
\label{eq:arbitrarystate}
\end{equation}

Using this, we can find a simple expression for the QFI matrix (see derivation over Appendix~\ref{appendix:derivationQFI}) using the notation introduced for our distributed qubits, namely using the partition $\boldsymbol{\mathcal{N}}$ with correspondent vector of resources given by $\vec{n}$: 
\begin{equation}
\begin{aligned}
	\boldsymbol{\mathcal{Q}} (\rho_{\thetab}) &= C \boldsymbol{\mathfrak{Q}} C^T , \quad \boldsymbol{\mathfrak{Q}} &= \Lambda - \vec{v}\vec{v}^T, 
\end{aligned}
\label{eq:qfiarbitrarymain}
\end{equation}
where $C$ is a matrix that can be made such that each line $i$ ranges from $i=0,\dots,2^{n-1}-1$ is given by the corresponding vectorial Hamming-weight of the index $i$, $\vec{h}^*_{\boldsymbol{\mathcal{N}}} (i) = \vec{n}-2\vec{c}_i$, with $\vec{0} \preceq \vec{c}_i \prec \vec{n}$. The matrix $\boldsymbol{\mathfrak{Q}}$ is a rank-1 correction of a diagonal matrix, namely $\Lambda = diag(\lambda_0,...,\lambda_j,...,\lambda_{n})$ and $\vec{v} = (v_0, ..., v_j, ..., v_{n})$, such that $\lambda_j \geq |v_j|$ for all $j$. The structure of this decomposition is exactly what allows us to prove the privacy statements we will make. We will be able to establish bounds on the rank of the QFI matrix using this fact, together with a few additional theorems and propositions we detail in Appendix~\ref{appendix:matrixrank}. 

As we have seen before, the functions at disposition are limited by the amount of resources used, so to provide the limits, we have defined privacy in terms of resources. Let us make a statement for each possible zone introduced in the Resources section (Sec.~\ref{sec:resources}). Note that in the previous section we had limited ourselves to states described by stabilizers. In this section we perform equivalent statements on privacy, but for arbitrary pure states, not necessarily limited to stabilizers. The first result is that we still cannot find a private state in the \textit{No privacy zone}:

\begin{theorem}[No privacy zone]
Let $\vec{a}, \vec{n} \in \mathbb{Z}^k$, $n = \norm{\vec{n}}_1$, $gcd(\vec{a})=1$ and $f(\thetab)$ be the linear function given by $f(\thetab) = \vec{a}\cdot \vec{\thetab}$. Then, there is no pure state capable of estimating $f$ privately with $\vec{n}$ distributed resources and access to local separable dynamics $U_{\thetab}$, such that $\vec{n} \prec \vec{a}$ or $\vec{n} \not\preceq \vec{a}$.
\begin{equation*}
	\nexists \ket{\psi} \in \mathcal{H}(n): \mathcal{P} (\mathcal{Q} (U_{\thetab}\ket{\psi}),\vec{a}) = 1
\end{equation*}
\label{thm:noprivacyarb}
\end{theorem}
\begin{proof}
The privacy is only zero if $\mathcal{Q} (U_{\thetab}\ket{\psi}) = \alpha \vec{a}\vec{a}^T \equiv \mathcal{Q}$. From Eq.~\ref{eq:qfiarbitrarymain}, take the support of $\mathfrak{Q}$, $\supp\mathfrak{Q}$ as the sub-matrix with the lines and columns such that $\lambda_j\neq 0$. Moreover, let $C_{\supp\mathfrak{Q}}$ be the correspondent sub-matrix of $C$ in the $\supp\mathfrak{Q}$, $i.e.$ the lines multiplying by the non-zero values of $\mathfrak{Q}$. Note that $\vec{a} \notin C$, and, in particular, $\vec{a} \notin C_{\supp\mathfrak{Q}}$. This is a consequence of $\vec{n} \prec \vec{a}$ or $\vec{n} \not\preceq \vec{a}$. Look at the following table covering all possible combinations of outcomes:

\begin{table}[H]
\centering
\begin{tabular}{c|c|c|c|c}
	rank $\supp\mathfrak{Q}$ & rank $C_{\supp\mathfrak{Q}}$ & rank $C\mathfrak{Q}C^T$ & Reasoning & Privacy  \\ \hline
	any & 1 & $\leq 1$ & $\vec{a} \notin C_{\supp(\mathfrak{Q})}$ & Never \\
	dim $\supp\mathfrak{Q}$ & $\geq 2$ & $\geq 2$ & Prop. \ref{prop:fullrank} & Never  \\
	dim $\supp\mathfrak{Q}$-1 & 2 & 1 or 2 & Prop. \ref{prop:notfullrank} & Check \\
	dim $\supp\mathfrak{Q}$-1 & $>2$ & $>1$ & Prop. \ref{prop:notfullrank} & Never  
\end{tabular}
\end{table}

The only case to check is that of rank $\supp\mathfrak{Q}$ = dim $\supp\mathfrak{Q}-1$ with rank $C_{\supp\mathfrak{Q}}$ = 2. According to Thm. \ref{thm:rank}, $\mathfrak{Q} = diag(\vec{\lambda}) - \vec{\lambda}\vec{\lambda}^T$, up to a matrix made up of $\pm 1$ in the diagonal, call it $D$. That means the vector $D\vec{1}$ spans the null space, therefore the positive eigenvectors of $\mathfrak{Q}$ are spanned by the orthogonal subspace of $D\vec{1}$, call it $P^\perp$. So, take the set of vectors $W = \{\vec{w} | \vec{w} = 1/\sqrt{2}(\vec{e}_i - \vec{e}_j), i\neq j \}$. Notice that the $\mathsf{span} (D W) = \mathsf{span} (P^T)$ as $\average{D\vec{w};D\vec{1}}=0 \ \forall w\in W$. Using Prop.~\ref{prop:vectorspan} we can analyse then the vectors in $DW$, which are simpler to check:

\begin{equation}
\begin{aligned}
	D\vec{w} C_{\supp \mathfrak{Q}} &= 1/\sqrt{2} \left( \vec{h}_{\boldsymbol{\mathcal{N}}}^* (i) \pm \vec{h}_{\boldsymbol{\mathcal{N}}}^* (j) \right) 
    =\begin{cases}
	2/\sqrt{2} \left[ \vec{n} - \left(\vec{h}_{\boldsymbol{\mathcal{N}}} (i) + \vec{h}_{\boldsymbol{\mathcal{N}}} (j) \right) \right] \\
	2/\sqrt{2} \left(\vec{h}_{\boldsymbol{\mathcal{N}}} (i) - \vec{h}_{\boldsymbol{\mathcal{N}}} (j) \right) 
	\end{cases},
\end{aligned}
\end{equation}
which by Props. \ref{thm:gcdorder} and \ref{thm:gcdorder2} we see it is never proportional to $\vec{a}$, given $gcd(\vec{a})=1$ and $\vec{n} \prec \vec{a}$ or $\vec{n} \not\preceq \vec{a}$. If it happens that it is zero, then no information can be obtained. This concludes the proof by noticing there is no way to construct any state achieving privacy for the function presented.
\end{proof}

In this proof, the tools and approach used provide a good insight on how to deal with the other cases. We present another statement for the \textit{Minimal privacy zone}, for which a similar proof can be found in Appendix \ref{appendix:privatestateproofs}.

\begin{restatable}[Minimal privacy zone]{theorem}{onecopy}
Let $\vec{a}, \vec{n} \in \mathbb{Z}^k$, $gcd(\vec{a})=1$ and $f(\thetab)$ be the linear function given by $f(\thetab) = \vec{a}\cdot \vec{\thetab}$. Then, there is only one family of pure states that is capable of estimating privately the function $f$ using $\vec{n} =  \vec{a}$ distributed resources and access to local separable dynamics $U_{\thetab}$.
\begin{equation*}
\begin{aligned}
	F_{GHZ} = \{ \alpha \ket{0}^{\otimes n} + \beta \ket{1}^{\otimes n}: \alpha,\beta \in \mathbb{C}\setminus \{0\}, |\alpha|^2+|\beta|^2=1 \} \\
	\\
	\mathcal{P} (\mathcal{Q} (U_{\thetab}\ket{\psi}),\vec{a}) = 1 \quad \Longleftrightarrow \quad \exists U\in LU: U\ket{\psi} \in F_{GHZ} 
\end{aligned}
\end{equation*}
\label{thm:oneprivacyarb}
\end{restatable}

Thms.~\ref{thm:noprivacyarb} and \ref{thm:oneprivacyarb} allow us to define the minimum amount of distributed resources so we can have a private state for a target linear function $f(\vec{\thetab}) = \vec{a}\cdot\vec{\thetab}$. This minimum is provided by the vector $\vec{a}$, under the separable Hamiltonians and integer target functions assumptions. 

Finally, one can also wonder what happens if we add more qubits than we need, but not sufficient enough, for example, to build two copies of a GHZ state with $\vec{a}$ resources. In order to do so, let us first provide an example of a simple transformation that preserves privacy:

\begin{example}
Suppose we have a partition for the distributed qubits given by $\boldsymbol{\mathcal{N}}$ with corresponding vector of resources $\vec{n}=(2,1,2)$ and we are trying to estimate a function $\vec{a}= (1,1,1) \prec \vec{n}$. If one uses dynamics generated by $Z$ operators, one can simply take a GHZ state and add ancillas in a separable way:
\begin{equation}
\begin{aligned}
	\ket{\psi} &= \frac{1}{\sqrt{2}} \left( 
	\ket{ \overbrace{\underbrace{0}_{a_1} \underbrace{0}_{b_1}}^{n_1} \overbrace{\underbrace{0}_{a_2}}^{n_2} \overbrace{\underbrace{0}_{a_3} \underbrace{0}_{b_3}}^{n_3} } + 
	\ket{ \overbrace{\underbrace{1}_{a_1} \underbrace{0}_{b_1}}^{n_1} \overbrace{\underbrace{1}_{a_2}}^{n_2} \overbrace{\underbrace{1}_{a_3} \underbrace{0}_{b_3}}^{n_3} }
	 \right), \\
	 &=  \frac{1}{\sqrt{2}} \left( 
	\ket{0_{a_1} 0_{a_2} 0_{a_3}} + \ket{1_{a_1} 1_{a_2} 1_{a_3}} \right) \ket{0_{b_1} 0_{b_3}}.
\end{aligned}
\end{equation}
This state is still private. And more so, it is still private if one chooses any state for the ancillas over the possible bit strings for the qubits $\vec{b} = \vec{n}-\vec{a}$, meaning $(0_{b_1}0_{b_3}), (0_{b_1}1_{b_3}), (1_{b_1}0_{b_3}),$ $ (1_{b_1}1_{b_3})$. Additionally, if one does any superposition of local permutations in $S_{\boldsymbol{\mathcal{N}}}$ over the qubits, this action preserves the vectorial Hamming-weight, meaning it will also preserve the privacy.
\end{example}

Taking this example, let us formalize everything by introducing some families of pure states that emerge as a natural consequence of the equivalence of the vectorial Hamming weight under local permutations of qubits:

\begin{definition}[Distributed $s$-States]\label{def:distdickestates}
Let $\boldsymbol{\mathcal{N}} = \{ \mathcal{N}_\mu \}_{\mu = 1,\dots,k}$ be a partition of the set of qubits and $s \in \mathbb{F}^n_2$ be a bit string. We call a quantum state $\ket{\psi}$ a distributed $s$-state if it can be represented by:
\begin{equation*}
\begin{aligned}
	\ket{\psi} = \sum_{r \in [s]_{\boldsymbol{\mathcal{N}}} } \alpha_r \ket{r} ,
\end{aligned}
\end{equation*}
where $\alpha_r \in \mathbb{C}$, such that $\sum_r |\alpha_r|^2 =1$ for normalization purposes. Let $\vec{s} = \vec{h}_{\boldsymbol{\mathcal{N}}} (s)$, we denote this family of states by $\mathcal{D}(\boldsymbol{\mathcal{N}},\vec{s})$.
\end{definition}

Note that the bit strings that belong to the equivalence class $[s]_{\boldsymbol{\mathcal{N}}}$ are exactly the ones for which the vectorial Hamming weight $h_{\boldsymbol{\mathcal{N}}}$ is the same, as they all differ by a permutation in $S_{\boldsymbol{\mathcal{N}}}$ from each other. In particular, if it was the case that the partitions have size 1 ($k=1$) and $\alpha_r$ is the same for every $r$, one would recover the regular Dicke state with $|\mathcal{N}_1| = n$ qubits and $h(s) = m$ excitations ($\ket{D_{n}^{m}}$).

This allows us to define a countable number of families that achieves privacy at estimating a linear function:

\begin{definition}[$\vec{a}$-Private-Ancilla States]\label{def:privatefamily}
Let $\boldsymbol{\mathcal{N}} = \{ \mathcal{N}_\mu \}_{\mu = 1,\dots,k}$ correspond to the partition of a set of distributed qubits, with associated resource vector $\vec{n} \succeq \vec{a}$ necessarily. 
\begin{multline}
	\mathcal{F}(\boldsymbol{\mathcal{N}},\vec{a},\vec{d}) = \Big\{ \alpha \ket{\psi} + \beta \ket{\tilde{\psi}},\ket{\psi} \in \mathcal{D}(\boldsymbol{\mathcal{N}},\vec{d}), \ket{\tilde{\psi}} \in \mathcal{D}(\boldsymbol{\mathcal{N}}, \vec{a}+\vec{d}),
	\\
    \alpha,\beta \in \mathbb{C}\setminus \{0\}, |\alpha|^2+|\beta|^2=1  \Big\} ,
\end{multline}
where $\vec{d}$ can be such that $\vec{d} \in \mathbb{N}_0^k, \vec{0} \preceq \vec{d} \preceq \vec{n}-\vec{a}$. Note that for a fixed $\vec{a}$ there will be multiple $\vec{d}$ available, depending on the number of qubits. Each of them defines one family of private states.
\end{definition}
 
To provide some intuition, let us provide all the private families for a given set of resources and target function:
\begin{example}
Let $\vec{n}= (2,2,5)$ and $\vec{a}=(1,2,3)$. This implies we have additionally in each node $\vec{b}=(1,0,2)$ distributed qubits. However, we can attribute to the ancilla qubits different strings $s$ such that $\vec{h}_{\boldsymbol{\mathcal{N}}} (s) = \vec{b}_1$ defining different families of private states. One way to see this is considering that $\vec{b}_1$ represents the amount of 1's we put in our ancilla qubits. We can additionally define a vector $\vec{b}_0$ as the amount of 0's in the ancilla qubits. Naturally $\vec{b}_0 + \vec{b}_1 = \vec{b}$. The total amount of 6 families are given by:
\begin{gather*}
	 s_1 = 0_{a_1}  0_{b_1} \ 00_{a_2} \ 000_{a_3} 00_{b_3}, \quad \tilde{s}_1 = 1_{a_1}  0_{b_1} \ 11_{a_2} \ 111_{a_3} 00_{b_3},\quad  \vec{d}_1 = (0, 0,0) \\
	 s_2 = 0_{a_1}  0_{b_1} \ 00_{a_2} \ 000_{a_3} 10_{b_3}, \quad \tilde{s}_2 = 1_{a_1}  0_{b_1} \ 11_{a_2} \ 111_{a_3} 10_{b_3},\quad \vec{d}_2 = (0, 0,1)\\
	 s_3 = 0_{a_1}  0_{b_1} \ 00_{a_2} \ 000_{a_3} 11_{b_3}, \quad \tilde{s}_3 = 1_{a_1}  0_{b_1} \ 11_{a_2} \ 111_{a_3} 11_{b_3},\quad \vec{d}_3 = (0, 0,2)\\
	 s_4 = 0_{a_1}  1_{b_1} \ 00_{a_2} \ 000_{a_3} 00_{b_3}, \quad \tilde{s}_4 = 1_{a_1}  1_{b_1} \ 11_{a_2} \ 111_{a_3} 00_{b_3},\quad \vec{d}_4 = (1, 0,0) \\
	 s_5 = 0_{a_1}  1_{b_1} \ 00_{a_2} \ 000_{a_3} 10_{b_3}, \quad \tilde{s}_5 = 1_{a_1}  1_{b_1} \ 11_{a_2} \ 111_{a_3} 10_{b_3}, \quad\vec{d}_5 = (1, 0,1)\\
	 s_6 = 0_{a_1}  1_{b_1} \ 00_{a_2} \ 000_{a_3} 11_{b_3}, \quad \tilde{s}_6 = 1_{a_1}  1_{b_1} \ 11_{a_2} \ 111_{a_3} 11_{b_3}, \quad\vec{d}_6 = (1, 0,2)\\
 \end{gather*}
\begin{gather*}
	 \mathcal{F}(\boldsymbol{\mathcal{N}},\vec{a},\vec{d}_i) = \Big\{  \alpha \sum_{r \in [s_i]_{\boldsymbol{\mathcal{N}}} } \alpha_r \ket{r} + \beta \sum_{\tilde{r} \in [\tilde{s}_i]_{\boldsymbol{\mathcal{N}}} } \alpha_{\tilde{r}} \ket{\tilde{r}} , \alpha,\beta \in \mathbb{C}\setminus \{0\}, |\alpha|^2+|\beta|^2=1  \Big\} .
\end{gather*}
Which is equivalent to saying:
\begin{multline*}
	 \mathcal{F}(\boldsymbol{\mathcal{N}},\vec{a},\vec{d}_i) = \Big\{  \alpha \ket{\psi} + \beta \ket{\tilde{\psi}} ,  \ket{\psi}  \in \mathcal{D} ( \boldsymbol{\mathcal{N}},s_i), \ket{\tilde{\psi}}  \in \mathcal{D} ( \boldsymbol{\mathcal{N}},\tilde{s}_i), \\
	 \alpha,\beta \in \mathbb{C}\setminus \{0\}, |\alpha|^2+|\beta|^2=1  \Big\} .
\end{multline*}
\end{example}

The number of existent families of states depends on the number of ancilla qubits in each node and the different ways one can prepare these ancillas. Considering that a node $\mu$ has $n_\mu$ qubits and is trying to estimate the linear function with correspondent $a_\mu$, then one has $b_\mu = n_\mu - a_\mu$ ancilla qubits. This means we have $b_\mu+1$ ways to find bitstrings with different Hamming-weights. The overall number of families is given by $\prod_{\mu \in V} (b_\mu+1)$. By construction, all the states that fit within one of these families are private for the estimation of $f(\thetab) = \vec{a}\cdot \vec{\thetab}$:

\begin{restatable}{proposition}{privatefamily}\label{prop:privatefamily}
All families of pure states $\mathcal{F}(\boldsymbol{\mathcal{N}},\vec{a},\vec{d})$, are capable of estimating privately the function $f(\thetab)=\vec{a} \cdot \vec{\thetab}$ using $\vec{n}=\vec{a}+\vec{b}$ distributed resources.
\end{restatable}
This allows us to finally get to the statement on privacy for when the amount of resources is larger than the minimum needed, but not as large so one could have multiples of the number of resources.

\begin{restatable}[Minimal Plus Ancilla Zone]{theorem}{onecopyplusancillaarb}\label{thm:onecopyplusancillaarb}
Let $\vec{a},\vec{b}, \vec{n} \in \mathbb{Z}^k$, $gcd(\vec{a})=1$ and $f(\thetab)$ be the linear function given by $f(\thetab) = \vec{a}\cdot \vec{\thetab}$. Then, all pure states, up to LU, capable of estimating privately the function $f$ using $\vec{n} =  \vec{a} + \vec{b}$ distributed resources and access to local separable dynamics $U_{\thetab}$, belong a family of the set of $\vec{a}$-Private-Ancilla States introduced in Def.~\ref{def:privatefamily}.
\begin{equation*}
\begin{aligned}
	\mathcal{P} (\mathcal{Q} (U_{\thetab}\ket{\psi}),\vec{a}) = 1 \quad \Longleftrightarrow \quad \exists U\in LU, \vec{d} \in \mathbb{N}_0^k, \vec{0} \preceq \vec{d} \preceq \vec{n}-\vec{a} : U\ket{\psi} \in \mathcal{F}({\boldsymbol{\mathcal{N}}},\vec{a},\vec{d})
\end{aligned}
\end{equation*}
\end{restatable}

The proofs for the above two statements, Prop.~\ref{prop:privatefamily} and Thm.~\ref{thm:onecopyplusancillaarb} are provided over Appendix~\ref{appendix:privatestateproofs}, following similar arguments as the proof for Thm.~\ref{thm:noprivacyarb}. Moreover, one can find an expression for QFI of the private state very similar to the one for the GHZ state, as detailed in the proof of Prop.~\ref{prop:privatefamily}.

Finally, we can discuss the case where we have multiples of the number of resources. For this case, we will not provide a complete set of private states, in the sense we will not prove that the states built in such fashion are the only private states. However, we are able to prove that every state built as so is indeed private. The way to think of this case is to consider that we have logical qubits such that the encoding dynamics by default already encode the desired function. This means exactly choosing our logical qubits to be the ones such that $\alpha \ket{0_L} + \beta \ket{1_L}$ generates a private state:

\begin{equation}
\begin{aligned}
	&\ket{0_L} \in \mathcal{D}(\boldsymbol{\mathcal{N}},\vec{d}), \\
	&\ket{1_L} \in \mathcal{D} (\boldsymbol{\mathcal{N}},\vec{a}+\vec{d}),
\end{aligned}
\end{equation}
for some $\vec{d}$, using the notation introduced in Defs.~\ref{def:distdickestates},\ref{def:privatefamily}. Using for each logical qubit $\vec{a}+\vec{b}$ resources, encompasses the case of $\vec{b}=\vec{0}$ which is the single GHZ state distributed among the $\vec{a}$ resources, and the states in Minimal plus ancilla region. Using these states, and our regular estimation scenario of $(\rho,U_{\thetab},\vec{a}\cdot\vec{\thetab})$ it is easy to understand that now our dynamics when applied to a superposition of $\ket{0_L}$ and $\ket{1_L}$ will always be sensitive only to the desired function $f(\thetab) = \vec{a}\cdot\vec{\thetab}$. Moreover, not all logical qubits need to be defined in the same way, as long as each of them are private. Let us define the states which are built like this:

\begin{definition}[Private Logical States]\label{def:privatelogicalstates}
Let $\mathfrak{N} = \{ \boldsymbol{\mathcal{N}}_1 , \boldsymbol{\mathcal{N}}_2, \cdots , \boldsymbol{\mathcal{N}}_d \} $ be a set of partitions such that each $\boldsymbol{\mathcal{N}}_l$ defines the partition set of a private state, with associated vector of resources $\vec{n}_l = \vec{a} + \vec{b}_l \succeq \vec{a} \  \forall l$. For each partition, one can define two logical states, such that:
\begin{equation*}
\begin{aligned}
    &\ket{0_L}_l \in \mathcal{D}(\boldsymbol{\mathcal{N}}_l,\vec{d}_l), \\
    &\ket{1_L}_l \in  \mathcal{D}(\boldsymbol{\mathcal{N}}_l,\vec{a} + \vec{d}_l) .
\end{aligned}
\end{equation*}
The set of logical private states can be defined as follows:
\begin{equation*}
\begin{aligned}
    \mathfrak{F} = \left\{ \ket{\psi} = \sum_{j=0}^{2^d-1} \alpha_j \ket{{j_1} {j_2} \cdots {j_d}} : \alpha_j \in \mathbb{C}, \ket{j_l} \in \left\{ \ket{0_L}_l, \ket{1_L}_l \right\},\braket{\psi | \psi} = 1 \right\}.
\end{aligned}
\end{equation*}
\end{definition}

As one can see, the problem is now combinatorial. One can arrange the qubits in multiple partitions, and pick multiple private logical states inside of these partitions. Let us make a statement about the privacy of states built like this, which we prove in Appendix~\ref{appendix:privatestateproofs}:

\begin{restatable}[Multiple Privacy]{theorem}{multiplecopiesprivacy}\label{thm:multiplecopiesprivacy}
All states which are described as a private logical state (Def.~\ref{def:privatelogicalstates}) define private states for quantum distributed sensing, up to LU operations.
\end{restatable}

This has an important consequence: we can treat this case as the single-parameter scenario. In particular, one could use the robustness results of the single-parameter case and construct equivalent states, but instead using the logical qubits. Moreover, we can use results on how logical qubits are affected by error to predict how a state made of those logical qubits evolves under error. 

This section finishes with a set of private states for distributed quantum sensing with separable dynamics, and proof that, for limited resources, they constitute the only set of private states. If we extend our resources, we find a construction of sets of private states, by logically encoding multiple qubits into logical qubits that are only sensitive to our target function.

\subsection{Privacy for Arbitrary States with General Hamiltonians \label{section:generalhamiltonian}}

In a similar manner as for separable hamiltonians, the problem complexity can be reduced by choosing an appropriate basis. The basis we will use in this case will be the basis introduced in Eq.~\ref{eq:hamiltonianbasis}, constituted by the eigenvectors of the Hamiltonian. This already implies that finding a private state requires finding eigenstates of the Hamiltonian, which is not necessarily trivial as it was for the separable Hamiltonian scenario. By choosing this basis one can build a very similar system to Eq.~\ref{eq:qfiarbitrary2} in the following way:

\begin{equation}
\begin{aligned}
		\mathcal{Q}_{\mu \nu} (\rho_{\thetab}) &= 4\re{ \sum_{\vec{l},\vec{m}} \alpha_{\vec{l}}^* \alpha_{\vec{m}} \braket{ \lambda_{\vec{l}} |   \boldsymbol{H}_\mu \boldsymbol{H}_\nu | \lambda_{\vec{m}} } - \sum_{\vec{l},\vec{m}} \alpha_{\vec{l}}^* \alpha_{\vec{m}} \braket{ \lambda_{\vec{l}} |   \boldsymbol{H}_\mu | \lambda_{\vec{m}} } \sum_{\vec{p},\vec{q}} \alpha_{\vec{p}}^* \alpha_{\vec{q}} \braket{ \lambda_{\vec{p}} |   \boldsymbol{H}_\nu | \lambda_{\vec{q}} } } \\
		&= 4\re{ \sum_{\vec{l},\vec{m}} \alpha_{\vec{l}}^* \alpha_{\vec{m}} \lambda^\mu_{l_\mu} \lambda^\nu_{m_\nu} \braket{ \lambda_{\vec{l}}  | \lambda_{\vec{m}} } - \sum_{\vec{l},\vec{m}} \alpha_{\vec{l}}^* \alpha_{\vec{m}}  \lambda^\mu_{m_\mu}\braket{ \lambda_{\vec{l}} | \lambda_{\vec{m}} } \sum_{\vec{p},\vec{q}}\alpha_{\vec{p}}^* \alpha_{\vec{q}}   \lambda^\nu_{q_\nu} \braket{ \lambda_{\vec{p}} |   \lambda_{\vec{q}} } } \\
		&=4 \sum_{\vec{m}} | \alpha_{\vec{m}}|^2 \lambda^\mu_{m_\mu} \lambda^\nu_{m_\nu} - \sum_{\vec{m}} |\alpha_{\vec{m}} |^2 \lambda^\mu_{m_\mu} \sum_{\vec{q}} | \alpha_{\vec{q}} |^2  \lambda^\nu_{q_\nu}   \\
		\boldsymbol{\mathcal{Q}} (\rho_{\thetab})&= C\mathfrak{Q} C^T,
\end{aligned}
\label{eq:generalQFI}
\end{equation}
where we used the fact that the basis states are $\ket{\lambda_{\vec{j}}}$ with $\vec{j}$ being given by some labeling function in a space with dimension of the product of the local Hilbert spaces, $e.g.$ $\vec{j} \in \{0, \cdots,  |\mathcal{H}_1| \} \times \cdots \times \{0, \cdots,  |\mathcal{H}_k| \}$. In this case, for each of the vectors $\vec{j}$ with associated $\alpha_{\vec{j}}$ in $\mathfrak{Q}$, there will exist an associated $c_{\vec{j}} $ in $C$, given by  $c_{\vec{j}} = (\lambda^1_{j_1},\lambda^2_{j_2},\cdots,\lambda^k_{j_k}) \in \mathcal{O}$. We provide a simple example:

\begin{example}
Let $k=3$ meaning 3 nodes, each holding one qubit and having access to a $Z$ hamiltonian (the local dynamics are given by $U_{\theta_\mu} = e^{-i\theta_\mu Z}$). Each local set of eigenvalues is the same, $\{-1,+1\}$, meaning $\mathcal{O} = \{-1,+1\} \times \{-1,+1\} \times \{-1,+1\} $. The set of possible eigenstates of the distributed Hamiltonian are then $\ket{j_1 j_2 j_3}$ where $\ket{j_i} \in \{\ket{0}, \ket{1}\}$. The possible set of functions can be mapped to $c_{\vec{j}} = ((-1)^{j_1}, (-1)^{j_2}, (-1)^{j_3})$. They form the vertices of a cube centered in 0 with side equal to 2.
\end{example}

Note that the matrix $\mathfrak{Q}$ can be written again in the same form as before, with a slight change: $\mathfrak{Q} = \text{diag}(\vec{\alpha}) - \vec{\alpha}\vec{\alpha}^T = \Lambda - \vec{\alpha}\vec{\alpha}^T$, where we used that each entry of $\vec{\alpha}$ corresponds to the probability amplitude $|\alpha_{\vec{j}}|^2$ for each and every $\vec{j}$. This vector will have the length of the amount of vectors $\vec{j}$ there are, $i.e.$ the total Hilbert space dimension. Given Thm.~\ref{thm:rank}, we know this matrix has rank equal to rank $\Lambda-1$.

From here we can already make some statements about finding private states, although we cannot go very far as it highly depends on the structure of $\mathcal{O}$, which naturally depends on the Hamiltonian. In order to better analyze this let us introduce the set of vectors going from each point of $\mathcal{O}$ to another, as it will become apparent later that these are the vectors defining linear functions for which we can find private states.

\begin{definition}
Let $\mathcal{O}$ be a discrete subset inside an orthotope in $k$ dimensions over some ordered ring $\mathbb{K}$. Define the higher order discrete subset $\mathcal{O}^2_\pm$ as:
\begin{equation}
    \mathcal{O}^2_\pm = \{ \vec{z} \in \mathbb{K}^k : \vec{z} = \vec{x} \pm \vec{y} , \forall \vec{x},\vec{y} \in \mathcal{O}\}.
\end{equation}
\end{definition}

One can further extend the class of equivalence of linear functions to a class of equivalence of the discrete subset of the orthotopes. As before $f(\thetab) \sim g(\thetab)$ if $\exists \alpha \in \mathbb{R}: f(\thetab)=\alpha g(\thetab)$, one can say two sets of vectors are equivalent $\mathcal{O}_1 \sim \mathcal{O}_2$ if $\forall \vec{v} \in \mathcal{O}_1 \exists \alpha \in \mathbb{R}, \vec{w} \in \mathcal{O}_2: \vec{v} = \alpha \vec{w} $ and the correspondence is bijective. We say one subset contains another up to equivalence $\mathcal{O}_1 \subsetsim \mathcal{O}_2$ if $\mathcal{O}_1 \sim \tilde{\mathcal{O}}_1 \subseteq \mathcal{O}_2$. Using this, we can perform a statement on a necessary condition in order to create private state in this scenario:

\begin{restatable}[Existence of Private States]{theorem}{existenceprivate}\label{thm:existenceprivate}
Let a distributed sensing scenario be governed by local encoding dynamics in the form of Eq.~\ref{eq:generaldynamics}, with associated discrete subset $\mathcal{O}$. For each vector $\vec{a}$ belonging to the higher order discrete subset $\mathcal{O}^2_-$, one can find a private state with respect to the linear function $f(\thetab)=\vec{a}\cdot\vec{\boldsymbol{\theta}}$. Moreover, no vector outside the equivalence class of $\mathcal{O}^2_-$ has a private state.
\end{restatable}

Let us provide again another more complex example to help visualize what we are doing, by plotting an illustration of the functions such that a private state exists.

\begin{example}\label{ex:cube} \textbf{(Cube in 3-D)}
Let the three local hamiltonians $\{H_1,H_2,H_3\}$ have possible eigenvalues $\{\lambda_j^\mu\}_{j=1,...,|H_\mu|}^{\mu = 1, 2, 3}$. All of the vectors in blue with accompanying points correspond to the discrete subset $\mathcal{O}$ and correspond to all the possible combinations of the eigenvalues in each local hamiltonian $(\lambda^1_i,\lambda^2_j,\lambda^3_k)$. These are the private linear functions one has access to, since for each of them we prove the existence of a private state. The 3-cube in red (given that the dimension $k=3$) corresponds to the $3-$orthotope $\overline{\mathcal{O}}$. The blue and black dots on the right side correspond to the set of $\mathcal{O}\cap \overline{\mathcal{O}}$. The black dots on the right side also correspond to the vertices of the orthotope  $V(\overline{\mathcal{O}})$, which can be seen as the intersection with the $k$-sphere. These correspond to the maximum of the 2-norm of the vectors inside the orthotope.
\begin{figure}[ht]
    \centering
\tdplotsetmaincoords{60}{120}
\begin{tikzpicture}
	[tdplot_main_coords,
        grid/.style={very thin,gray},
	axis/.style={->,black,thick},
	cube/.style={opacity=.2,dashed,fill=red},
        axis hidden/.style={opacity=.4,black,thick},
        vector/.style={->,blue,thin,opacity=.4},
        scale=0.9]
	
	\draw[axis] (2,0,0) -- (6,0,0) node[anchor=west]{$x$};
	\draw[axis] (0,2,0) -- (0,4,0) node[anchor=west]{$y$};
	\draw[axis] (0,0,2) -- (0,0,4) node[anchor=west]{$z$};
    \draw[axis hidden] (-2,0,0) -- (2,0,0);
	\draw[axis hidden] (0,-2,0) -- (0,2,0);
	\draw[axis hidden] (0,0,-2) -- (0,0,2);

	\draw[cube] (-2,-2,-2) -- (-2,2,-2) -- (2,2,-2) -- (2,-2,-2) -- cycle;
	
	\draw[cube] (-2,-2,-2) -- (-2,2,-2) -- (-2,2,2) -- (-2,-2,2) -- cycle;

	\draw[cube] (-2,-2,-2) -- (2,-2,-2) -- (2,-2,2) -- (-2,-2,2) -- cycle;

	\draw[cube] (2,-2,-2) -- (2,2,-2) -- (2,2,2) -- (2,-2,2) -- cycle;

	\draw[cube] (-2,2,-2) -- (2,2,-2) -- (2,2,2) -- (-2,2,2) -- cycle;

	\draw[cube] (-2,-2,2) -- (-2,2,2) -- (2,2,2) -- (2,-2,2) -- cycle;

    \node[draw=none,shape=circle,fill, inner sep=1pt] (d1) at (0,2,0){};
    \node[draw=none,shape=circle,fill, inner sep=1pt] (d1) at (0,-2,0){};
    \node[draw=none,shape=circle,fill, inner sep=1pt] (d1) at (2,0,0){};
    \node[draw=none,shape=circle,fill, inner sep=1pt] (d1) at (-2,0,0){};
    \node[draw=none,shape=circle,fill, inner sep=1pt] (d1) at (0,0,2){};
    \node[draw=none,shape=circle,fill, inner sep=1pt] (d1) at (0,0,-2){};

    \foreach \x in {-2,-0.5,0.5,2}
	\foreach \y in {-2,-0.5,0.5,2}
            \foreach \z in {-2,-1,1,2}
		{
			\draw[vector] (0,0,0) -- (\x,\y,\z);
            		\node[draw=none,shape=circle,fill, inner sep=1pt,blue] (d1) at (\x,\y,\z){};
		} 
    
\end{tikzpicture}
\hfil
\begin{tikzpicture}
	[tdplot_main_coords,
        grid/.style={very thin,gray},
	axis/.style={->,black,thick},
	cube/.style={opacity=.5,dashed,fill=red},
        axis hidden/.style={opacity=.4,black,thick},
        vector/.style={->,blue,thin},
        scale=0.9]
    \tdplotsetrotatedcoords{0}{0}{0};
    \draw[dashed,tdplot_rotated_coords,gray] (0,0,0) circle (3.4641);
 
    \tdplotsetrotatedcoords{90}{90}{90};
    \draw[dashed,tdplot_rotated_coords, gray] (3.4641,0,0) arc (0:360:3.4641);
 
    \tdplotsetrotatedcoords{0}{90}{90};
    \draw[dashed,tdplot_rotated_coords,gray] (3.4641,0,0) arc (0:360:3.4641);
 
	\draw[axis] (2,0,0) -- (6,0,0) node[anchor=west]{$x$};
	\draw[axis] (0,2,0) -- (0,4,0) node[anchor=west]{$y$};
	\draw[axis] (0,0,2) -- (0,0,4) node[anchor=west]{$z$};

	\draw[cube] (2,-2,-2) -- (2,2,-2) -- (2,2,2) -- (2,-2,2) -- cycle;

	\draw[cube] (-2,2,-2) -- (2,2,-2) -- (2,2,2) -- (-2,2,2) -- cycle;

	\draw[cube] (-2,-2,2) -- (-2,2,2) -- (2,2,2) -- (2,-2,2) -- cycle;

    \node[draw=none,shape=circle,fill, inner sep=1pt] (d1) at (0,2,0){};
    \node[draw=none,shape=circle,fill, inner sep=1pt] (d1) at (2,0,0){};
    \node[draw=none,shape=circle,fill, inner sep=1pt] (d1) at (0,0,2){};

    \foreach \x in {-2,-0.5,0.5,2}
	\foreach \y in {-2,-0.5,0.5,2}{
		\node[draw=none,shape=circle,fill, inner sep=1pt,blue] (d1) at (\x,\y,2){};
	}

    \foreach \x in {-2,-0.5,0.5,2}
	\foreach \y in {-2,-1,1,2}{
      		 \node[draw=none,shape=circle,fill, inner sep=1pt,blue] (d1) at (2,\x,\y){};
	} 
	
    \foreach \x in {-2,-1,1,2}
	\foreach \y in {-2,-0.5,0.5,2}{
		\node[draw=none,shape=circle,fill, inner sep=1pt,blue] (d1) at (\y,2,\x){};
	} 

    \node[draw=none,shape=circle,fill, inner sep=2pt,black] (d1) at (2,2,2){};
    \node[draw=none,shape=circle,fill, inner sep=2pt,black] (d1) at (2,-2,2){};
    \node[draw=none,shape=circle,fill, inner sep=2pt,black] (d1) at (2,2,-2){};
    \node[draw=none,shape=circle,fill, inner sep=2pt,black] (d1) at (2,-2,-2){};
    \node[draw=none,shape=circle,fill, inner sep=2pt,black] (d1) at (-2,2,2){};
    \node[draw=none,shape=circle,fill, inner sep=2pt,black] (d1) at (-2,-2,2){};
    \node[draw=none,shape=circle,fill, inner sep=2pt,black] (d1) at (-2,2,-2){};

    \shade[ball color = lightgray,opacity = 0.1] (0,0,0) circle (3.4641cm);

\end{tikzpicture}
\end{figure}
\end{example}

Moreover, one may wonder about the possibility of finding private states for functions within $\mathcal{O}$, not only in $\mathcal{O}_-^2$. This is again highly dependent on the Hamiltonian. For this reason we will give a sufficient condition for one to find these Hamiltonians, which relates with the inclusion of $\mathcal{O}$ in the higher order subset $\mathcal{O}^2_-$:

\begin{restatable}{proposition}{symmetrichamiltonians}\label{prop:symmetrichamiltonians}
Let $\{\boldsymbol{H}_\mu\}_{\mu \in V}$ be the Hamiltonians describing some local general encoding dynamics. 
\begin{equation}
    \forall \ \mu \in V \ \exists \ X_\mu \in \mathbb{1}\otimes \mathcal{P}_\mu: \{X_\mu, \boldsymbol{H}_\mu\} = 0
    \implies \mathcal{O} \subsetsim \mathcal{O}^2_-
\end{equation}
Where $\mathbb{1}\otimes \mathcal{P}_\mu$ are the Pauli strings where only the elements in $\mu$ can be non-identity. This means the functions in $\mathcal{O}$ are available privately by choosing an appropriate state.
\end{restatable}

We prove Thm.~\ref{thm:existenceprivate} and Prop.~\ref{prop:symmetrichamiltonians} in Appendix~\ref{appendix:general}. The latter is supposed to give some symmetry conditions for the Hamiltonian, under which the functions available are directly extractable from the original set $\mathcal{O}$. Note that the form of $\mathcal{O}$ is not always a cube centered in the origin, even though $\mathcal{O}^2_-$ is. Prop.~\ref{prop:symmetrichamiltonians} also implies that $\mathcal{O}$ will be symmetric with respect to the origin. From here, one can deduce a simple, yet important consequence that relates with the idea of concentration of information. 

\begin{restatable}[Information Concentration]{corollary}{concentration}\label{corol:concentration}
Given a distributed sensing scenario governed by local encoding dynamics in the form of Eq.~\ref{eq:generaldynamics}, the maximum information extractable, given by functions of parameters in the vertices of $\overline{\mathcal{O}}^2_-$, is accessible only by a private state.
\end{restatable}

The proof for Corol.~\ref{corol:concentration} can also be found over Appendix \ref{appendix:general}. This has important consequences, and provides an intuition that has been in the background the whole time: one can either try and estimate one function with maximum precision, or a set of functions with less precision each. Moreover, one can retrieve the solutions for separable Hamiltonians from Thm.~\ref{thm:existenceprivate}, by looking at the imposed condition on $gcd(\vec{a})=1$. Taking into consideration that the orthotope $\mathcal{O}\subseteq \mathbb{Z}^k$, the existence of private states is a consequence of integers' properties.

\section{Robustness \label{section:robustness}}

So far, we have been able to calculate the QFI for special types of quantum states. We started with the case of stabilizer states with separable Hamiltonians. We then generalized for any arbitrary pure state, still under the assumption of separable Hamiltonians. We then lifted this last assumption to general local Hamiltonians, at the expense of not being able to identify the state, and instead proving that private states still exist and for which functions they do exist. The next logical step is to consider mixed states as well. These types of states are often associated with the presence of noise in the state, hence the name of this section. Robustness should measure the ability of a quantum state to retain some of its properties upon being affected by noise \cite{Demkowicz-Dobrzanski2017,Ouyang2022a,Shettell2022,Hayashi2022a}. In our case, our goal is to maximize precision, which is connected with the maximization of the QFI. Robustness is then linked to how the QFI varies as noise is applied onto the system. In this section, we will start by finding a way to calculate the QFI of mixed states and verify how the QFI changes under different noises types, keeping in mind the privacy as well. Moreover, we assume that the noise is applied to the state prior to any encoding dynamics, as this noise should be the one that could possibly break privacy.  

\subsection{QFI for Mixed States}

Starting with the equation for the QFI in the multiparameter scenario:

\begin{equation}
\begin{aligned}
	\mathcal{Q}_{\mu\nu} (\rho_{\thetab}) =  \Tr \mathcal{R}_{\rho_{\thetab}}^{-1} (\partial_{\theta_\mu}\rho_{\thetab}) \rho_{\thetab} \mathcal{R}_{\rho_{\thetab}}^{-1} (\partial_{\theta_\nu}\rho_{\thetab}),
\end{aligned}
\label{eq:qfimultimixed1}
\end{equation}
where the super operator $ \mathcal{R}_{\rho_{\thetab}}^{-1} (\hat{O})$ is the inverse of the super operator $ \mathcal{R}_{\rho_{\thetab}} (\hat{O})$, defined by the following set of equations:
\begin{equation}
\begin{aligned}
	\mathcal{R}_\rho (O)&= \frac{\rho O + O \rho}{2}, \\
	\mathcal{R}_\rho^{-1} (O)&= \sum_{j,k} \frac{2}{\lambda_j + \lambda_k} O_{jk} \ketbra{\mathcal{G}_j}{\mathcal{G}_k}.
\end{aligned}
\end{equation}
Given a decomposition of the density matrix in its eigenvectors $\rho = \sum_k \lambda_k \ket{\mathcal{G}_k}\bra{\mathcal{G}_k}$ and the matrix elements $O_{jk} = \bra{\mathcal{G}_j} \hat{O} \ket{\mathcal{G}_k}$. Developing Eq.~\ref{eq:qfimultimixed1}, one can arrive at the following expression, similar to \cite{Rezakhani2019}:
\begin{equation}
\begin{aligned}
	\mathcal{Q}_{\mu\nu} (\rho_\thetab) = \sum_{\substack{n \in \supp \\ k\in\supp}} 2\frac{(\lambda_n - \lambda_k)^2}{\lambda_n + \lambda_k}  \braket{\mathcal{G}_n | \boldsymbol{G}_\mu | \mathcal{G}_k} \braket{ \mathcal{G}_k | \boldsymbol{G}_\nu | \mathcal{G}_n} + \sum_{\substack{n \in \nul \\ k\in\supp}}  4  \lambda_k  \braket{\mathcal{G}_n | \boldsymbol{G}_\mu | \mathcal{G}_k} \braket{ \mathcal{G}_k | \boldsymbol{G}_\nu | \mathcal{G}_n} ,
	\label{eq:noisyqfi}
\end{aligned}
\end{equation}
where we used $\supp$ and $\null$ to denote the indices $k$ such that $\lambda_k$ is either positive or zero. This can be further simplified when $\rho = \sum_k \lambda_k \ket{\mathcal{G}_k}\bra{\mathcal{G}_k} $ is given in an orthonormal basis, $i.e.$ such that every $\braket{\mathcal{G}_j | \mathcal{G}_k} = \delta_{jk}$. Even if this is not the case, it is a well established property of quantum information that one can always find a basis for $\ket{\mathcal{G}_j}$ that makes a certain density operator have such form. Nonetheless, one arrives at:
\begin{equation}
\begin{aligned}
	\mathcal{Q}_{\mu\nu} (\rho_\thetab) = \sum_{\substack{n>k \in \supp \\ n\neq k}} 4\left[ \frac{(\lambda_n - \lambda_k)^2}{\lambda_n + \lambda_k} - (\lambda_n + \lambda_k) \right] \re{ a_{nk}^\mu a_{kn}^\nu} + \sum_{\substack{n\in\supp}}  \lambda_n \mathcal{Q}_{\nu\mu} (\ket{\mathcal{G}_n}) ,
\end{aligned}
\label{eq:qfinoisysimple}
\end{equation}
where $a_{nk}^\mu = \braket{\mathcal{G}_n | \boldsymbol{G}_\mu | \mathcal{G}_k}$. The entire derivation of these equations can be found in Appendix~\ref{appendix:qfimixedstates}. In the following sections we will assume that the encoding dynamics are generated by $Z$ operators, without loosing generality. If they are generated by other generator $G$, assume that dephasing noise is given by an equivalent quantum channel that has a mixture of error in the direction of $G$. For the bit-flip noise, assume the error is in the direction of some $G^\perp$. Again, we also assume that the noise is applied to the state, prior to evolution under the dynamics that encode the parameters $\thetab$.

\subsection{Dephasing Noise}

Dephasing noise is characterized by applying a quantum channel with a certain probability of a $Z$ error. Taking a GHZ state with $\vec{n}$ qubits distributed over the network, the usual encoding dynamics in the $Z$ direction and the usual target function $f(\thetab) = \vec{a}\cdot\vec{\thetab}$, such that $\vec{a}=\vec{n}$, $i.e.$ we are working with a perfectly private state. Since the GHZ state is made up of eigenvectors of the $Z$ operator, dephasing noise acts identically on whichever qubit of the GHZ state. Namely consider the two orthogonal states, differeing from applying one $Z$ operation on whichever qubit, $\ket{GHZ^\pm} = 1/\sqrt{2}(\ket{0}^{\otimes n} \pm \ket{1}^{\otimes n}) \equiv \ket{\mathcal{G}_0^\pm}$, where $n=\normone{\vec{n}}$. Applying a dephasing noise channel to a $\ket{GHZ^{\pm}}$ results in introducing a mixture between these two states:

\begin{equation}
\begin{aligned}
	\mathcal{D}_i (\rho, p) &= (1-p) \rho + p Z_i \rho Z_i, \qquad Z \ket{\mathcal{G}_0^\pm} = \ket{\mathcal{G}_0^\mp}.
\end{aligned}
\end{equation}

Applying this recursively for each qubit, we can find the following expression for an initial GHZ state :
\begin{equation}
\begin{aligned}
	\mathfrak{D} (\rho, \vec{p}) &= \mathcal{D}_{1}^{p_1} \circ \mathcal{D}_{2}^{p_2} \circ \dots \circ \mathcal{D}_{n}^{p_n} (\rho) , \\
	\mathfrak{D} (\ketbra{\mathcal{G}_0^\pm}{\mathcal{G}_0^\pm}, \vec{p}) &=   E(\vec{p} ) \ketbra{\mathcal{G}_0^\pm}{\mathcal{G}_0^\pm} + O(\vec{p}) \ketbra{\mathcal{G}_0^\mp}{\mathcal{G}_0^\mp},
\end{aligned}
\end{equation}
where the even, $E(\cdot)$, and ddd functions, $O(\cdot)$, correspond to the Kraus operators of the total channel with even and odd number of $Z$ operators, and were introduced in \cite{Bugalho2023} to calculate the fidelity of a GHZ state.
One can use Eq. \ref{eq:noisyqfi} to calculate the QFI of a mixed state where we have a mixture between $\ket{\mathcal{G}_0^\pm}$, given by $\rho = \lambda_+ \ketbra{\mathcal{G}_0^+}{\mathcal{G}_0^+} + \lambda_- \ketbra{\mathcal{G}_0^-}{\mathcal{G}_0^-}$. Moreover, using the $Z$ dynamics as $U_\thetab = \Motimes_{\mu} e^{i\theta_\mu\sum_{j\in\mu}Z_j} = \Motimes_{\mu} e^{i \theta_\mu \boldsymbol{Z}_\mu}$:
\begin{equation}
\begin{aligned}
	\mathcal{Q}_{\mu\nu} (\rho_\thetab) &= 4\frac{(\lambda_+ - \lambda_-)^2}{\lambda_+ + \lambda_-}  \braket{\mathcal{G}_0^+ | \boldsymbol{Z}_\mu | \mathcal{G}_0^-} \braket{ \mathcal{G}_0^- | \boldsymbol{Z}_\nu | \mathcal{G}_0^+} \\
	 &= 4(\lambda_+-\lambda_-)^2 a_\mu a_\nu ,
\end{aligned}
\end{equation}
which preserves the privacy. This is a consequence of the dynamics being generated by the same operation as the error, namely by $Z$ operators on each qubit. Moreover, if $\vec{p}$ is uniform, as in every qubit dephases at the same rate, we get that $\lambda_+ - \lambda_- = (1-2p)^n$, where $n$ is the number of qubits suffering dephasing.

\subsection{Bit-flip Noise}

Bit-flip noise on the other hand is characterized by applying a quantum channel with a certain probability of $X$ error. Under the same estimation scenario as before, we will arrive at the following density matrix:

\begin{equation}
\begin{aligned}	
	\rho &= \lambda_0 \ketbra{\mathcal{G}_0^+}{\mathcal{G}_0^+} + \sum_{i=1}^{2^{n-1}-1} \lambda_i \ketbra{\mathcal{G}_i^+}{\mathcal{G}_i^+}, 
\end{aligned}
\end{equation}
where $\ket{\mathcal{G}_m} =1/\sqrt{2} ( \ket{m} +  \ket{\overline{m}} ) $ and $0\leq m \leq 2^{(n-1)}-1$ to prevent repetition of states, given the invariance of $m\rightarrow \overline{m}$. Taking the derivation of the QFI in Eq.~\ref{eq:qfinoisysimple} we get that all terms $a_{nk}^\mu$ are going to be zero cause the dynamics are generated by $Z$ operators, and the errors are generated by $X$ operators. This way, no operator containing only $\{\mathbb{1},Z\}$ terms will ever take one $\ket{\mathcal{G}_j} $ to another $\ket{\mathcal{G}_k} $, since this map is generated by some operator containing only $\{\mathbb{1},X\}$ terms. This results in the QFI being given by:

\begin{equation}
\begin{aligned}	
    \boldsymbol{\mathcal{Q}} (\rho_\thetab) &= 4 \sum_{j=0}^{2^{n-1}-1} \lambda_j \vec{h}_{\boldsymbol{\mathcal{N}}}^* (j) \vec{h}_{\boldsymbol{\mathcal{N}}}^* (j)^T  ,\\
    \lambda_j &= g(\vec{p},j) + g(\vec{p},\overline{j}),
\end{aligned}
\label{eq:bitflipnoise}
\end{equation}
where we introduce the function that maps a vector of reals and a bitstring into a real number, namely:
\begin{equation}
\begin{aligned}	
    g: \mathbb{R}^k \times \mathbb{F}_2^k &\longrightarrow \mathbb{R} \\
    g(\vec{p},j) &\longmapsto \prod_{i=1}^k p_i^{j_i} (1-p_i)^{1-j_i}.
\end{aligned}
\end{equation}

One can observe in Eq.~\ref{eq:bitflipnoise}, complete privacy is lost as soon as one qubit suffers a bit-flip. This is due to the fact that now the generators of the dynamics are $Z$ operators and the errors only contain $X$ operators. However, one can still make a statement using the fifth property of the privacy measure: if the noise can be maintained under a certain threshold, then one can bound the loss of privacy. Equivalently, the amount of error will bound the amount of information leaked.

\subsection{Depolarizing Noise}

Depolarizing noise can be seen as a very symmetric noise, as it decoheres in every direction with identical probabilities. In particular, taking a GHZ state and applying a depolarizing channel in each qubit, one arrives at a state which is very similar to a generalized Werner state \cite{Werner1989}, apart from the identity not really being an identity matrix, but a diagonal matrix. Nonetheless, using the description of the depolarizing channel as in \cite{Bugalho2023}:

\begin{equation}
\begin{aligned}	
    \mathcal{D}_i^p(\rho) & =p \rho+\frac{1-p}{3}\left(X_i \rho X_i^{\dagger}+Y_i \rho Y_i^{\dagger}+ Z_i \rho Z_i^{\dagger}\right) \\
    & =\frac{1+2 p}{3} \rho+\frac{2(1-p)}{3} \Lambda_i\left(Y_i \rho Y_i^{\dagger}\right).
\end{aligned}
\end{equation}

The density matrix can be expressed as:
\begin{equation}
\begin{aligned}	
	\rho &= \lambda_0^+ \ketbra{\mathcal{G}_0^+}{\mathcal{G}_0^+} + \lambda_0^- \ketbra{\mathcal{G}_0^-}{\mathcal{G}_0^-} + \sum_{i=1}^{2^{n}-2} \lambda_i \ketbra{i}{i} ,
\end{aligned}
\end{equation}
since $\braket{i|j} = \braket{i|\mathcal{G}_0^\pm}=0$ for $i\neq 0,2^n-1$, and, similar to before, the dynamics are in $Z$, all the terms $a_{0^\pm j}^\mu = 0$. Moreover, the QFI of each of the $\ketbra{i}{i}$ with $i>0,$ is zero, as no information can be obtained from states which are just a binary string. This allows us to write the QFI matrix as:
\begin{equation}
\begin{aligned}	
    \boldsymbol{\mathcal{Q}} (\rho_\thetab) &= 4\frac{(\lambda_0^+ - \lambda_0^-)^2}{\lambda_0^+ + \lambda_0^-} \vec{a} \vec{a}^T + \sum_{j=1}^{2^{n}-2} \lambda_j \boldsymbol{\mathcal{Q}}(\ket{j}) \\
    &= 4\frac{(\lambda_0^+ - \lambda_0^-)^2}{\lambda_0^+ + \lambda_0^-} \vec{a} \vec{a}^T ,
\end{aligned}
\end{equation}
where $\lambda_0^\pm$ can be found be solving the following set of equations:
\begin{equation}
\begin{cases}	
    \lambda_0^+ + \lambda_0^- = g((1+2\vec{p})/3,0) + g(2(1-\vec{p})/3,0) \\
    \lambda_0^+ - \lambda_0^- = g((4\vec{p}-1)/3,0).
\end{cases}
\end{equation}

Importantly, this type of noise again preserves privacy. Unlike previous errors, this error contains both the same operators as in the encoding dynamics and orthogonal to the ones generating the dynamics. The privacy preservation in this case comes from the symmetry of this channel acting on the GHZ state, which is also highly symmetric.

\subsection{Amplitude-Damping Noise}

Amplitude-damping noise can be seen as a consequence of spontaneous emission of excited photons, in a way that the excited states (in our case $\ket{1}$), with some probability, $p_i$ per qubit, go to the ground state (in our case $\ket{0}$). Taking a GHZ state and applying an amplitude-damping channel in each qubit, one arrives at something very similar to the depolarizing channel, with some slight changes:

\begin{equation}
\begin{aligned}	
	\rho &= \lambda_0^+ \ketbra{\mathcal{G}_0^+}{\mathcal{G}_0^+} + \lambda_0^- \ketbra{\mathcal{G}_0^-}{\mathcal{G}_0^-} + \sum_{i=0}^{2^{n}-1} \lambda_i \ketbra{i}{i} \\
    &= \lambda_0 \ketbra{0}{0} + \lambda_0^+ \ketbra{\mathcal{G}_0^+}{\mathcal{G}_0^+} + \lambda_0^- \ketbra{\mathcal{G}_0^-}{\mathcal{G}_0^-} + \lambda_1 \ketbra{1}{1} + \sum_{i=1}^{2^{n}-2} \lambda_i \ketbra{i}{i}.
\end{aligned}
\label{eq:amplitudedamping}
\end{equation}
One can diagonalize the system into the following:
\begin{equation}
\begin{aligned}	
	\rho &= \tilde{\lambda}_0^+ \ketbra{\tilde{\mathcal{G}}_0^+}{\tilde{\mathcal{G}}_0^+} + \tilde{\lambda}_0^- \ketbra{\tilde{\mathcal{G}}_0^-}{\tilde{\mathcal{G}}_0^-} + \sum_{i=1}^{2^{n}-2} \lambda_i \ketbra{i}{i} .
\end{aligned}
\end{equation}
The QFI can then be obtained by:
\begin{equation}
\begin{aligned}	
    \boldsymbol{\mathcal{Q}} (\rho_\thetab) &= 4\frac{(\tilde{\lambda}_0^+ - \tilde{\lambda}_0^-)^2}{\tilde{\lambda}_0^+ + \tilde{\lambda}_0^-} \left|\braket{ \tilde{\mathcal{G}}_0^+|Z| \tilde{\mathcal{G}}_0^- } \right|^2 \vec{a} \vec{a}^T ,
\end{aligned}
\end{equation}
where $\tilde{\lambda}_0^\pm$ and $\tilde{\mathcal{G}}_0^\pm$ can be found from the solution of the following eigensystem, written in the basis $\{ \ket{0}^{\otimes n} \equiv \ket{0},\ket{1}^{\otimes n}\equiv \ket{1}\}$:
\begin{equation}
\begin{pmatrix}	
    a & b \\
    b & c
\end{pmatrix} =
\begin{pmatrix}	
    1 + \prod_i p_i & \prod_i \sqrt{1-p_i} \\
    \prod_i \sqrt{1-p_i} & \prod_i (1-p_i)
\end{pmatrix} .
\end{equation}

Yet again, this noise still preserves privacy. The reason behind this comes again from the symmetric form of the GHZ state with respect to the Kraus operators of the amplitude damping channel, observed from Eq.~\ref{eq:amplitudedamping}.

\subsection{Particle Loss}

To check what happens to the state under particle loss, let us start with the simple case where we have $\vec{n}=\vec{a}$ resources. This means the only private state is given by the GHZ state, with the distribution of resources given by $\vec{a}$, in the regular setting. As it is known, the GHZ state has no resilience against particle loss, and all the information is lost upon particle loss. A GHZ state, $\ket{\mathcal{G}_0^+}$, after loosing whichever particle is transformed into:

\begin{equation}
\begin{aligned}
	\mathcal{E}_i (GHZ) &= \Tr_i \ketbra{\mathcal{G}_0^+}{\mathcal{G}_0^+} \\
	&=\braket{0_i \ketbra{\mathcal{G}_0^+}{\mathcal{G}_0^+} 0_i} + \braket{1_i \ketbra{\mathcal{G}_0^+}{\mathcal{G}_0^+} 1_i} \\
	&= \frac{1}{2} \left[ \ketbra{0_{\mathcal{N}\setminus i}}{0_{\mathcal{N}\setminus i}}  + \ketbra{1_{\mathcal{N}\setminus i}}{1_{\mathcal{N}\setminus i}} \right].
\end{aligned}
\end{equation}
Using again the equation for the noisy QFI, we get the QFI becomes $0$ for whichever parameter, as this state is invariant under the encoding dynamics.

Throughout the paper, we have defined families of private states when we have more ancilla qubits at our disposal, besides the minimum amount to create a private state. One can observe that there are states in these private families that show robustness against particle loss. To do this, suppose we have a state belonging to one of these families, $\mathcal{F}(\boldsymbol{\mathcal{N}},\vec{a},\vec{d})$:
\begin{equation}
\begin{aligned}
	\ket{\psi} &= \alpha \ket{\psi_0} + \beta \ket{\psi_1}, \qquad \ket{\psi_0} \in \mathcal{D}(\boldsymbol{\mathcal{N}},\vec{d}), \ket{\psi_1} \in \mathcal{D}(\boldsymbol{\mathcal{N}},\vec{a}+\vec{d})  .
 \label{eq:particlelossinitial}
\end{aligned}
\end{equation}

Let us first analyze what happens to the states in the private family $\mathcal{F}(\boldsymbol{\mathcal{N}},\vec{a},\vec{d})$ when a qubit is traced. Since we can choose whichever basis for the partial trace, let us measure the qubit $j$ in the $Z$ basis:
\begin{equation}
\begin{aligned}
	&\ket{\psi_0} \in \mathcal{D}(\boldsymbol{\mathcal{N}},\vec{d}) \implies \braket{0_j | \psi_0} \in \mathcal{D}(\boldsymbol{\mathcal{N}}\setminus j,\vec{d}) , \\
    &\ket{\psi_1} \in \mathcal{D}(\boldsymbol{\mathcal{N}},\vec{a} + \vec{d})  \implies \braket{0_j | \psi_1} \in  \mathcal{D}(\boldsymbol{\mathcal{N}}\setminus j,\vec{a}+\vec{d}) , \\
    &\hspace{2cm}\Downarrow\\[-.5ex]
    &\alpha' \braket{0_j | \psi_0} + \beta' \braket{0_j | \psi_1} \in \mathcal{F}(\boldsymbol{\mathcal{N}}\setminus j,\vec{a},\vec{d}).
\end{aligned}
\label{eq:loss0}
\end{equation}
This holds true if both $\braket{0_j | \psi_0}$ and $\braket{0_j | \psi_1}$ are larger then zero. If not, it will lead to a state with no information available. The same can be said if one measures the state $\ket{\psi}$ in $\ket{1}_j$:
\begin{equation}
\begin{aligned}
	&\ket{\psi_0} \in \mathcal{D}(\boldsymbol{\mathcal{N}},\vec{d}) \implies \braket{1_j | \psi_0} \in \mathcal{D}(\boldsymbol{\mathcal{N}}\setminus j,\vec{d}-\vec{e}_j) , \\
    &\ket{\psi_1} \in \mathcal{D}(\boldsymbol{\mathcal{N}},\vec{a} + \vec{d})  \implies \braket{1_j | \psi_1} \in  \mathcal{D}(\boldsymbol{\mathcal{N}}\setminus j,\vec{a}+\vec{d}-\vec{e}_j) , \\
    &\hspace{2cm}\Downarrow\\[-.5ex]
    &\alpha' \braket{1_j | \psi_0} + \beta' \braket{1_j | \psi_1} \in \mathcal{F}(\boldsymbol{\mathcal{N}}\setminus j,\vec{a},\vec{d} -\vec{e}_j),
\end{aligned}
\label{eq:loss1}
\end{equation}
where we used $\vec{e}_j$ to denote the vector that is one for the node $\mu$ such that qubit $j\in \mu$ and zero everywhere else.
It is intuitive to understand that there is a generalization for this at each qubit that is lost. The resilience against any particle loss can be made larger by increasing the amount of ancillas. The privacy is guaranteed as, after loosing each qubit, we obtain a mixture of orthogonal states, which either have information about $\vec{a}$, as they belong to a family of private states, or they do not have information at all. To prove this consider:

\begin{proposition}\label{prop:orthogonalfamilies}
Let $\ket{\psi} \in  \mathcal{F}(\boldsymbol{\mathcal{N}} ,\vec{a},\vec{d})$, $\ket{\phi} \in  \mathcal{F}(\boldsymbol{\mathcal{N}} ,\vec{a},\vec{d}')$.
\[
	\vec{d} \neq \vec{d}' \implies \braket{\psi|\phi} = 0
\]
\end{proposition}
\begin{proof}
The proof comes straightforward from realizing that distributed $s$-states are orthogonal if their generating strings $s$ have different vectorial Hamming-weights.
\end{proof}
\begin{proposition}\label{prop:spanfamily}
Let $\Psi = \{ \ket{\psi_1}, \dots, \ket{\psi_n} \}$ be a collection of states of a private family, $\ket{\psi_j} \in  \mathcal{F}(\boldsymbol{\mathcal{N}} ,\vec{a},\vec{d}) \ \forall \ j = 1,\dots,n$.
\[
	\ket{\Psi} = \sum_{j=1}^n \alpha_j \ket{\psi_j} \in \mathcal{F}(\boldsymbol{\mathcal{N}} ,\vec{a},\vec{d}) \ \forall \  \alpha_j \in \mathbb{C}: \sum_{j=1}^n |\alpha_j|^2
\]
\end{proposition}
\begin{proof}
The proof comes again straightforward as $\ket{\Psi}$ will only contain vectors belonging to $\mathcal{D}(\boldsymbol{\mathcal{N}},\vec{d})$ and $\mathcal{D}(\boldsymbol{\mathcal{N}},\vec{a}+\vec{d})$.
\end{proof}

Using this we are now able to prove the following theorem:
\begin{theorem}
All the private states in the region of minimal plus ancilla amount of resources are able to either remain completely private under $\vec{e}$-qubit loss, if the amount of ancillas are at least $\vec{b} \succeq \vec{e}$, or provide no information at all.
\end{theorem}

\begin{proof}
Any state belonging to the family of private states on $\mathcal{F}(\boldsymbol{\mathcal{N}},\vec{a},\vec{d})$, with associated vector of resources $\vec{n}=\vec{a}+\vec{b}$ already verifies privacy, by construction (Prop.~\ref{prop:privatefamily}). The evolution of these states under loss can be seen from Eqs.~\ref{eq:loss0},\ref{eq:loss1}. Note that after the partial trace of each qubit we will have a mixture of states that belong to private families. We can group them according to the families they belong to, ensuring the states in the mixture are orthogonal and then using Prop.~\ref{prop:orthogonalfamilies} and Eq.~\ref{eq:qfinoisysimple} to verify privacy. In the case they belong to the same family, diagonalizing them makes, by Prop.~\ref{prop:spanfamily}, new orthogonal states belonging to the same family. Nonetheless in this last case by analyzing Eq.~\ref{eq:qfinoisysimple} we also verify that this remains private. Applying the same procedure as in these equations to all qubits $j\in\vec{e}$, we get that the final eigenstates of the density matrix all belong to either private families, or are states with zero information at hand. 
\end{proof}

\subsection{Generalizing Results}

When addressing the most general case, we use the private logical states introduced in Def.~\ref{def:privatelogicalstates}. For this, there is a condition to be able to use noise resiliency strategies already reproduced for single-parameter quantum metrology, and achieve the same outcome for multiparameter private metrology. The property is the following:

\begin{condition}
We say a logical qubit basis $\{\ket{0_L},\ket{1_L}\}$ is private against an error generated by a quantum channel $\mathcal{E}(\cdot)$ if:
\begin{equation}
    \forall \alpha, \beta \in \mathcal{C}: \ket{\psi} = \alpha \ket{0_L} + \beta \ket{1_L}, \mathcal{E}(\ketbra{\psi}{\psi}) = \sum_i \lambda_i \ketbra{\psi_i}{\psi_i},
\end{equation}
such that $\braket{\psi_i | \psi_j} = \delta_{ij}$ and $\exists \gamma : \braket{\psi_i | \boldsymbol{G}_\mu | \psi_j} = \gamma a_\mu$, for all $i,j$. 
\end{condition}

This allows us to say that, if a logical qubit basis is private against an error, then any state generated in that basis will also be private, under that same error. The consequence follows trivially from Eq.~\ref{eq:qfinoisysimple}. This is already the case for the states in the minimal privacy zone under dephasing, depolarising and amplitude-damping noises and for minimal plus ancillas private states under particle loss.

\section{Discussion of the Results}

We have condensed the results of this work for the existence of private families of states with the corresponding assumptions and theorems in Table~\ref{tab:results}.

\begin{table*}
\begin{tabular}{c|c|c|c|c}
Target & Hamiltonian & Resources & Private Family & Theorem \\ \hline \hline
$\vec{a} \in \mathbb{R}^k$ & Controlable (Fig.~\ref{fig:circuit} d)) & $\vec{n}=\vec{1}$ & GHZ State & Thm.~\ref{thm:oneprivacyarb} \\ \hline
\multirow{4}{*}{ $\vec{a} \in \mathbb{N}^k$ } & \multirow{4}{*}{ Separable (Fig.~\ref{fig:circuit} a)) } & Zone (I) & None & Thms.~\ref{thm:noprivacystab}, \ref{thm:noprivacyarb} \\
& & Zone (II) & GHZ State & Thms.~\ref{thm:oneprivacystab}, \ref{thm:oneprivacyarb} \\
& & Zone (III) & Ancilla-Private States & Thm.~\ref{thm:onecopyplusancillaarb} \\
& & Zone (IV) & Private Logical States & Thm.~\ref{thm:multiplecopiesprivacy} \\ \hline
$\vec{a} \in \mathcal{O}_-^2$ & General (Fig.~\ref{fig:circuit} b)) & Depends on $H$ & $\exists$ Private State & Thm.~\ref{thm:existenceprivate} \\
\end{tabular}
\caption{Main results regarding the existence of private states under the assumptions on the encoding dynamics, the target functions and the resources utilized.}
\label{tab:results}
\end{table*}

Throughout this paper we presented several families of states, which hold an important property for distributed sensing scenarios. Privacy comes as a natural consequence of choosing a state which makes available information about one single function of parameters. In all the results for pure states, the decomposition of the QFI in Eqs.~\ref{eq:qfiarbitrarymain} and ~\ref{eq:generalQFI} was crucial, and simplified the problem. Given its generality, one could potentially use this as means to efficiently calculate the QFI matrix, by following a similar procedure, and grouping states according to their vectorial Hamming weight.

The first no-go result is that, if one is not able to control the encoding dynamics, there is a minimum amount of resources such that one can estimate privately the target linear function, and this naturally depends on the Hamiltonian of the encoding dynamics. In particular, for separable Hamiltonians, with integer resources, one is only able to measure a linear function that is proportional to an integer combination of the parameters. The second result still for this class of Hamiltonians is, if we have exactly the amount of required resources, there is only one family of states that is able to estimate the target function privately, and this family is SLOCC equivalent to a GHZ state. Moreover, the GHZ state is the one for which the QFI is maximized, and therefore, the precision is also maximized. This result also affects the controllable encoding dynamics case, implying the privacy of a GHZ state with arbitrary number of qubit in each node. If one has $n_\mu$ on mode $\mu$, and control over a linear parameter, $t_\mu$, locally, then as long as one can make $n_\mu t_\mu = a_\mu$ for all $\mu$, then one finds a private state.

Another important result is that, by adding extra qubits and superpositions of states that differ by local permutations, we can gain some redundancy, allowing us to create private states that maintain privacy and information even after qubit loss. This has important consequences for practical applications in near-term experiments. 

Finally, we provide methods to analyze when multiples of the vector of resources are available. In this scenario, the scaling of the precision with the amount of resources is a pertinent question. In most cases, this is identical to dealing with the single parameter scenario by choosing logical qubits that are constituted by the basis of private states. 

Moving to more general Hamiltonians, not necessarily separable, but still local in the sense that each node can only encode their own parameter, we found a structure on the space of available private functions and a way to generate the private states. However, in this scenario, it requires knowledge about the eigenstates of the Hamiltonian which is not a trivial problem, although there are some strategies to find them \cite{Santagati2018,Meister2022}. Most importantly, this approach for general Hamiltonians provides an interesting insight to the structure of finding states for sensing. In particular it relates the maximum precision at which one can estimate a function employing a quantum sensing scheme and the privacy of such estimation, as they walk hand in hand. This was expected as the best local strategy is often to use locally entangled states, where the amount of information is the same, but scattered across each local parameter. When analyzing the best strategy for a given target function, privacy can also be perceived as an intuitive idea of concentrating all the information along a target function. 

Overall, these results obtained clarify exactly which states one should build in order to perform secure and private from construction protocols for distributed quantum sensing. The continuity the privacy measure with respect to the QFI information, allows one to infer and create protocols that provide bounds on the QFI to establish bounds on the privacy of a state. These bounds in the privacy can then be converted into security concerns by an appropriate protocol, as we have discussed, that uses state certification rounds intercalated with sensing rounds.

This work opens up a different way to look at sensing protocols, relating resources and hamiltonian dynamics to information directly. While the framework has been consistently that of quantum sensing, where measurements are repeated and from them estimators are created, one could potentially apply our results to distributed quantum computation. In particular, we dealt with only linear functions, which seems to be a consequence of the Hamiltonian dynamics not being able to encode arbitrary functions of parameters. A possible direction of future work could be coming up with protocols to access non-linear functions.

Another possible direction would be to find all of the private states when multiples of the vector of resources are available. Since technically from the definition of a Private-Ancilla state, these are not limited to $\vec{n} \preceq 2\vec{a}$, it is possible that finding all states involves constructing a hierarchy of private states by the logical encoding formalism taking the logical qubits inside each of the $(m\vec{a})$-Private-Ancilla states, $m=1,2,\dots$, and verify possible repetitions of families by playing with different encodings.

Finally, the structure of the private orthotope, and the fact that it is an orthotope, might hint at some more fundamental results. It would be interesting to understand exactly are the implications of this, outside of the privacy in the sensing scenario, and into an information perspective by using the intuition of concentration of information.

\section*{Acknowledgements}

The authors acknowledge the support from the EU Quantum Flagship project QIA (101102140) and France 2030 under the French National Research Agency projects HQI ANR-22-PNCQ-0002 and the PEPR integrated project EPiQ ANR-22-PETQ-0007. L.B. and Y.O. thank the support from Funda\c c\~ao para a Ci\^encia e a Tecnologia (FCT, Portugal), namely through project UIDB/04540/2020. L.B. acknowledges the support of FCT through scholarship BD/05268/2021.

\clearpage

\appendix
\title

\section{Unitary Equivalence of Encoding Dynamics \label{appendix:unitaryequivalence}}

In this appendix we go over the proofs of the interplay between the dynamics and the initial state. Starting by restating the first theorem:

\unitaryequiv*

\begin{proof}
Since $G$ and $G'$ are generators for single qubit unitary operators, then they can be decomposed into a linear combination of Pauli operators. Let them be $G = \vec{a} \cdot \vec{\sigma}$ and $G' = \vec{b} \cdot \vec{\sigma}$, such that $|\vec{a}|=|\vec{b}|=1$. Using the expression for the adjoint action on the Pauli vector
\begin{equation}
	R_{\vec{n}} (-\alpha) \ \vec{\sigma} \  R_{\vec{n}} (\alpha) = \vec{\sigma} \cos(\alpha) + \vec{n} \times \vec{\sigma} \sin(\alpha) + \vec{n} \ \vec{n}\cdot \vec{\sigma} (1-\cos(\alpha))
\end{equation} 
Taking the inner product with $\vec{a}$ we get that:
\begin{align*}
	R_{\vec{n}} (-\alpha) \ \vec{a} \cdot \vec{\sigma} \  R_{\vec{n}} (\alpha) &= \vec{a}\cdot\vec{\sigma} \cos(\alpha) + \vec{a} \cdot (\vec{n} \times \vec{\sigma}) \sin(\alpha)  + \vec{a}\cdot \vec{n} \ \vec{n}\cdot \vec{\sigma} (1-\cos(\alpha)) \\
	&= \left[ \vec{a} \cos(\alpha) + \vec{a}\times \vec{n} \sin(\alpha) + \vec{n} \ \vec{a}\cdot\vec{n} (1-\cos(\alpha))\right] \cdot \vec{\sigma} \\
	&= \left[ \mathbb{1} \cos(\alpha) + \vec{n} \vec{n}^T (1-\cos(\alpha)) + [\vec{n}]_\times \sin(\alpha) \right] \vec{a} \cdot \vec{\sigma} \\
	&= (R_{\vec{n}}^{xyz} (\alpha) \vec{a}) \cdot \vec{\sigma} \\
	&= \vec{b} \cdot \vec{\sigma},
\end{align*}
where $R_{\vec{n}}^{xyz} (\alpha)$ is the rotation in $\mathbb{R}^3$ around the axis $\vec{n}$, by an angle of $\alpha$. This in turn means that we can always find values for $\vec{n}$ and $\alpha$ such that one can transform the vector $\vec{a}$ into vector $\vec{b}$. Setting $W = R_{\vec{n}} (\alpha)$ we can always find a unitary transformation that transforms one generator into another.
\end{proof}

Which consequently leads to:

\unitaryequivalence*

\begin{proof}
Follows trivially from previous theorem:
\begin{align*}
	W^\dagger U_\theta W &=  W^\dagger e^{-i\theta G} W \\
	&= W^\dagger \left[ \mathbb{1} + \frac{-i\theta}{1!} G + \frac{(-i\theta)^2}{2!} G^2 +  ... \right] W \\
	&=W^\dagger W + \frac{-i\theta}{1!} W^\dagger G W + \frac{(-i\theta)^2}{2!} W^\dagger G^2 W  + ... \\
	&= \mathbb{1} + \frac{-i\theta}{1!} W^\dagger G W + \frac{(-i\theta)^2}{2!} W^\dagger G W W^\dagger G  W  + ... \\
	&= \mathbb{1} + \frac{-i\theta}{1!} G' + \frac{(-i\theta)^2}{2!} G'^2  + ... \\
	&= e^{-i\theta G'}.
\end{align*}
\end{proof}

\section{Privacy Measure \label{appendix:privacymeasure}} 
In this appendix we look at the measure of privacy introduced in Def.~\ref{def:privacymeasure}, and prove the properties it should follow, introduced in the main text. Looking at the second and third properties, they should come by construction, as we verify below. For the last two, it is not obvious, but we still get them:

\begin{enumerate}
\setcounter{enumi}{1}
	\item From the fact that $\mathcal{Q}$ is a positive semi-definite matrix, we get that $\mathcal{P} (\mathcal{Q},\vec{a}) \leq 1$. The minimum value follows from:
		\begin{equation}
		\begin{aligned}
			\frac{ \vec{a}^T \mathcal{Q} \vec{a}}{\Tr \mathcal{Q}} & =   \frac{ \vec{a}^T \sum_i \beta_i \vec{b}_i \vec{b}_i^T \vec{a}}{\Tr \mathcal{Q}}  \\
			&=  \frac{ \sum_i \beta_i ( \vec{a} \cdot \vec{b}_i )^2 }{\Tr \left[ \sum_i \beta_i \vec{b}_i \vec{b}_i^T \right] } \\
			&=  \frac{ \sum_i \beta_i ( \vec{a} \cdot \vec{b}_i )^2 }{ \sum_i \beta_i  } \\
			&\leq 1,
		\end{aligned}
		\end{equation}
		as $\vec{a} \cdot \vec{b}_i \leq 1$ always, choosing normalized vectors $\vec{b}_i$.
	\item The fact that only a QFI of the type $\mathcal{Q} = \alpha \vec{a}\vec{a}^T$ saturates the bounds comes trivially from before, as $\vec{a}\cdot \vec{b}_i = 1$ iff $\vec{b}_i = \vec{a}$. Moreover, even if the QFI does not come in a diagonal form, we can still show that:
		\begin{equation}
		\begin{aligned}
			\frac{ \vec{a}^T \mathcal{Q} \vec{a}}{\Tr \mathcal{Q}} & =   \frac{ \vec{a}^T \sum_{ij} \beta_{ij} \vec{b}_i \vec{b}_j^T \vec{a}}{\Tr \mathcal{Q}}  \\
			&=  \frac{ \vec{a}^T \sum_{i} \tilde{\beta}_{i} \vec{\tilde{b}}_i \vec{\tilde{b}}_i^T \vec{a}}{\Tr \mathcal{Q}} \\
			&=  \frac{ \sum_i \tilde{\beta}_i ( \vec{a} \cdot \vec{\tilde{b}}_i )^2 }{ \sum_i \tilde{\beta}_i  } \\
			&\leq 1,
		\end{aligned}
		\end{equation}
		as the matrix of the QFI $\mathcal{Q}$ is always real and symmetric, so there is always a diagonalization for it. Moreover, since it is positive semi-definite matrix all $\tilde{\beta}_i$ are positive. The only way the bound is saturated is if $\exists i: \vec{\tilde{b}}_i = \vec{a}$ and $\tilde{\beta}_j = 0 \forall j\neq i$.
	\item 
		\begin{equation}
		\begin{aligned}
			\mathcal{P} (B\mathcal{Q}B^{T},B\vec{a}) &=  \frac{\Tr B\mathcal{Q}B^T - \vec{a}^T B^T B\mathcal{Q}B^T B \vec{a}}{\Tr B\mathcal{Q}B^T} \\
			&=\frac{\Tr \mathcal{Q}B^TB - \vec{a}^T \mathcal{Q} \vec{a}}{\Tr \mathcal{Q}B^TB} \\
			&=\frac{\Tr \mathcal{Q} - \vec{a}^T \mathcal{Q} \vec{a}}{\Tr \mathcal{Q}} \\
			&= \mathcal{P} (\mathcal{Q},\vec{a}) ,
		\end{aligned}
		\end{equation}
		where we made the assumption that $B$ is a orthonormal change of basis matrix, $i.e.$ $BB^T=\mathbb{1}$. In case it is not, the two measures are not equivalent, but they still present a valid skewed measure.
	\item 
		\begin{equation}
		\begin{aligned}
			\mathcal{P} (\mathcal{Q}+\epsilon A,\vec{a}) - \mathcal{P} (\mathcal{Q},\vec{a}) &=  \frac{\Tr \left[\mathcal{Q}+\epsilon A\right] - \vec{a}^T (\mathcal{Q}+\epsilon A) \vec{a}}{\Tr \left[\mathcal{Q}+\epsilon A\right]}  -\frac{\Tr \mathcal{Q} - \vec{a}^T \mathcal{Q} \vec{a}}{\Tr \mathcal{Q}} \\
			&= 1-\frac{\vec{a}(\mathcal{Q}+\epsilon A)\vec{a}^T}{\Tr \mathcal{Q}+\epsilon \Tr A} - \left(1-\frac{\vec{a}\mathcal{Q}\vec{a}^T}{\Tr \mathcal{Q}}  \right) \\
			&=  \frac{\vec{a}\mathcal{Q}\vec{a}^T}{\Tr \mathcal{Q}} -  \frac{\vec{a}\mathcal{Q}\vec{a}^T}{\Tr \mathcal{Q}+\epsilon \Tr A} -\epsilon \frac{\vec{a}A\vec{a}^T}{\Tr \mathcal{Q}+\epsilon \Tr A} \\
			|\mathcal{P} (\mathcal{Q}+\epsilon A,\vec{a}) - \mathcal{P} (\mathcal{Q},\vec{a}) | &\leq \left|  \frac{\vec{a}\mathcal{Q}\vec{a}^T}{\Tr \mathcal{Q}+\epsilon \Tr A}  + \epsilon \frac{\vec{a}A\vec{a}^T}{\Tr \mathcal{Q}+\epsilon \Tr A} - \frac{\vec{a}\mathcal{Q}\vec{a}^T}{\Tr \mathcal{Q}}\right| \\
			&\leq  \left|  \frac{\vec{a}\mathcal{Q}\vec{a}^T}{\Tr \mathcal{Q}}  + \epsilon \frac{\vec{a}A\vec{a}^T}{\Tr \mathcal{Q}} - \frac{\vec{a}\mathcal{Q}\vec{a}^T}{\Tr \mathcal{Q}}\right| \\
			&= 	 \left| \epsilon \frac{\vec{a}A\vec{a}^T}{\Tr \mathcal{Q}} \right| \equiv \delta		,
		\end{aligned}
		\end{equation}
	providing the continuity of the privacy upon change of the initial state.
\end{enumerate}

Using this measure of private one can try and understand what the different values of $\mathcal{P}$ mean. The complete private case $\mathcal{P}=1$ can be verified to naturally satisfy all conditions for the privacy. If the QFI matrix is of the form $\mathcal{Q}= \alpha \vec{a}\vec{a}^T$ then is only possible to find an estimator for the function $f(\thetab) = \vec{a}\cdot \vec{\thetab}$. This also means that whatever estimators for other functions would be ill-defined and not accessible (infinite variance). Consequently all parties can not estimate but the target function, and the dishonest parties having access to all the values in the dishonest subset, could at most gather information about the honest parties function $f(\thetab_H)$ but not their local values. Here and henceforth, taking a vector of distributed parameters $\vec{a}$, we will denote $\vec{a}_H$ has the subset of the same vector belonging to the honest parties.

For the case of $\mathcal{P}=\epsilon$ let us provide two examples to understand how $\epsilon$ would affect the information that might be leaked in the case of dishonest parties:

\begin{example}
Suppose the QFI matrix is given by:
\[
	\mathcal{Q} = \epsilon \vec{a}\vec{a}^T + (1-\epsilon)\vec{b}\vec{b}^T,
\]
such that $\vec{a}_H \cdot \vec{b}_H = 0$, $\vec{a}\cdot\vec{b} = 0$, denoting by $\vec{x}_H$ the vector restricted to the honest parties. First verify that $\mathcal{P}=\epsilon$ for the estimation of the function $f(\thetab)=\vec{a}\cdot \vec{\thetab}$. Then, since the dishonest parties can share their own values and therefore have access to all $\theta_\mu \in D$, they also have access to $\vec{b}_H \cdot \vec{\thetab}$ which is a different from $f(\thetab_H)$, which is the private function. This is not private, since it violates the third condition of privacy. If $\epsilon=0$, then their information about this function is the maximum possible for this estimation $(\Tr \mathcal{Q})$.
\end{example}

\begin{example}
Suppose we have $k$ parties and $d\leq k/2$ dishonest parties among them. Le the QFI matrix be given by:
\[
	\mathcal{Q} = \epsilon \vec{a}\vec{a}^T + \frac{1-\epsilon}{d}\sum_{j=1}^d \vec{b}_j\vec{b}_j^T,
\]
Such that $\vec{a} \cdot \vec{b}_j = 0$ for all vectors $\vec{b}_j$. In particular, one can choose these vectors $\vec{b}_j$ such that $\vec{b}_{j_H}$, which is the correspondent to slicing the vectors $\vec{b}_j$ for the indices correspondent to honest parties, form a $d$-dimensional orthogonal subspace. Again, we first verify that $\mathcal{P} = \epsilon$ as every $\vec{b}_j$ is orthogonal to $\vec{a}$. Second, we verify that the set of information available by considering every dishonest party knows the dishonest subspace, is given by an orthogonal basis of $d$ vectors in the honest subspace. Namely they have information about all $\vec{b}_{j_H}\cdot\vec{\thetab}$ after using the dishonest information. In particular, if $d=k/2$, then they have $d$ linearly independent vectors to estimate the honest subspace with $k-d = d$ functions. This means they also have access to each individual parameter of the honest parties. Again, if $\epsilon \rightarrow 0$ then they are more efficient, in the sense the variance at which they can do this the same amount of resources is smaller.
\end{example}

Note however that in both of these cases, for the parties to implement a non-private strategy and acquire information about the honest parties, they would have to distribute a non-private state. If one utilizes a state certification protocol in conjunction with this \cite{Shettell2022a}, then one can guarantee privacy and security.

\section{Integer Vectors Properties \label{appendix:integers}}

\gcdmultiple*
\begin{proof}
    If $\alpha \in \mathbb{Z}$, then $\vec{a} \in \mathbb{Z}^k \implies \vec{b} \in \mathbb{Z}^k$. To prove it is the only option choose a $b_i$ such that $b_i = \alpha a_i$, where $\alpha \notin \mathbb{Z} = c/d$. For $b_j$ to be integer for every other $j$, then $gcd(\vec{a},d)>1$, as in, every $a_i$. However, $gcd$ is associative, meaning $gcd(\vec{a},d)=gcd(gcd(\vec{a}),d)=1$. This means $\alpha$ has to be integer.
\end{proof}

\gcdorder*
\begin{proof}
Since $\vec{a},\vec{b} \prec \vec{m}, \vec{0} \prec \vec{a} + \vec{b} \prec 2\vec{m}, -\vec{m} \prec \vec{a} - \vec{b} \prec \vec{m}$. $gcd(\vec{m})=1 \implies $ if $ \alpha \vec{m} \propto \vec{m}$ then $\alpha \in \mathbb{Z}$. 
\begin{equation}
\begin{aligned}
	\Leftrightarrow \begin{cases}
	2 \left[ \vec{m} - \left(\vec{a}+ \vec{b} \right) \right]  \propto \vec{m}\\
	2 \left(\vec{a}- \vec{b} \right) \propto \vec{m}
	\end{cases} \Leftrightarrow 
	\begin{cases}
	\left[ \vec{m} - \left(\vec{a}+ \vec{b} \right) \right]  \propto \vec{m}\\
	\left(\vec{a}- \vec{b} \right) \propto \vec{m}
	\end{cases} \text{ but }
	\begin{cases}
	-\vec{m} \prec \vec{m} - \left(\vec{a}+ \vec{b} \right) \prec \vec{m}\\
	-\vec{m} \prec \vec{a} - \vec{b} \prec  \vec{m}
	\end{cases} .
\end{aligned}
\end{equation}
\end{proof}

\gcdorderm*

\begin{proof}
The proof follows similarly. Since $\vec{a},\vec{b} \prec \vec{n}, \vec{0} \prec \vec{a} + \vec{b} \prec 2\vec{n}, -\vec{n} \prec \vec{a} - \vec{b} \prec \vec{n}$. $gcd(\vec{m})=1 \implies $ if $ \alpha \vec{m} \propto \vec{m}$ then $\alpha \in \mathbb{Z}$. 
\begin{equation}
\begin{aligned}
	\Leftrightarrow \begin{cases}
	2 \left[ \vec{n} - \left(\vec{a}+ \vec{b} \right) \right]  \propto \vec{m}\\
	2 \left(\vec{a}- \vec{b} \right) \propto \vec{m}
	\end{cases} \Leftrightarrow 
	\begin{cases}
	\left[ \vec{n} - \left(\vec{a}+ \vec{b} \right) \right]  \propto \vec{m}\\
	\left(\vec{a}- \vec{b} \right) \propto \vec{m}
	\end{cases} \text{ but }
	\begin{cases}
	-\vec{n} \prec \vec{n} - \left(\vec{a}+ \vec{b} \right) \prec \vec{n}\\
	-\vec{n} \prec \vec{a} - \vec{b} \prec  \vec{n}
	\end{cases} .
\end{aligned}
\end{equation}
From here we verify that as $\vec{n}\not\preceq \vec{m}, \exists j: n_j < m_j$ and $gcd(\vec{m})=1$ and therefore no proportionality can be found.
\end{proof}

\section{QFI Matrix for Arbitrary States \label{appendix:derivationQFI}}

Starting from the definition of an arbitrary pure state of Eq.~\ref{eq:arbitrarystate} can write the elements of the QFI matrix (see Eq.~\ref{eq:qfipuremulti}) as:
\begin{equation}
\begin{aligned}
	\mathcal{Q}_{\mu \nu} (\rho_{\thetab}) &= 4\re{ \sum_{l,m=1}^{2^n} \alpha_l^* \alpha_m \braket{ \mathcal{G}_l |   \boldsymbol{G}_\mu \boldsymbol{G}_\nu | \mathcal{G}_m } - \sum_{l,m=1}^{2^n} \alpha_l^* \alpha_m  \braket{ \mathcal{G}_l |\boldsymbol{G}_\mu |  \mathcal{G}_m } \sum_{p,q=1}^{2^n} \alpha_p^* \alpha_q\braket{ \mathcal{G}_p | \boldsymbol{G}_\nu |  \mathcal{G}_q }  },
\end{aligned}
\label{eq:qfiarb1}
\end{equation}
with $\boldsymbol{G}_\mu = \sum_{j\in\mu} G_j$. To calculate the different elements above consider the following labeling for the basis states: let $\tilde{s}_l^\pm = (G_1^\perp)^{l_1} \otimes (G_2^\perp)^{l_2} \otimes \cdots \otimes (G_{n-1}^\perp)^{l_{n-1}} \otimes (G_n)^{\pm}$, where $G_n^+ = \mathbb{1}$ and $G_n^- = G_n$. Let $l \in \{0,1,\cdots,2^{n-1}-1\}$ and $l_1l_2\cdots l_{n-1}$ be the binary string associated to $l$. From our definition of $G,G^\perp$, we also have that:
\begin{equation}
\begin{aligned}
	 G_j \circ G_j^\perp \circ G_j &= -G_j^\perp ,\qquad  G_j \circ G_j \circ G_j =  G_j ,\\
    G_j^\perp \circ G_j \circ G_j^\perp &= -G_j ,\qquad  G_j^\perp \circ G_j^\perp \circ G_j^\perp =  G_j^\perp.
\end{aligned}
\end{equation}
Using this when applying $\tilde{s}_l^\pm$ to the stabilizer table in Eq.~\ref{eq:tablegenerators1}, we get a minus sign for the lines $j$ such that $l_j=1$. This yields:

\begin{equation}
\begin{aligned}
    \braket{ \mathcal{G}_l^\pm | \sum_{j \in \mu} \sum_{k \in \nu} G_j G_k | \mathcal{G}_m^\pm } &= h^*_{\mathcal{N}_\mu} (l) h^*_{\mathcal{N}_\nu} (m) \braket{ \mathcal{G}_l^\pm | \mathcal{G}_m^\pm }\\
    \braket{ \mathcal{G}_l^\pm | \sum_{j \in \mu} G_j |  \mathcal{G}_m^\pm } &= h^*_{\mathcal{N}_\mu} (l)   \braket{ \mathcal{G}_l^\pm | \mathcal{G}_m^\mp },
\end{aligned}
\end{equation}
where $h^*_{\mathcal{N}_\mu}$ is the $\mu$ element of the vector-Hamming weight in Def.~\ref{def:vectorialhammingweight}. Note that states differing on $\pm$ have the same Hamming-weight, as expected since $G_n$ commutes with all $G_jG_k$ terms. We can rewrite then Eq.~\ref{eq:qfiarb1} the QFI in the following way:

\begin{equation}
\begin{aligned}
	\mathcal{Q}_{\mu \nu} (\rho_{\thetab}) &= \sum_{l=0}^{2^{n-1}-1}  \left[ |\alpha_l^+|^2 + |\alpha_l^-|^2 \right] h^*_{\mathcal{N}_\mu} (l) h^*_{\mathcal{N}_\nu} (l) - \\
	&- \left[ \sum_{l=0}^{2^{n-1}-1} 2 \re{\alpha_l^{+^*} \alpha_l^-}  h^*_{\mathcal{N}_\mu} (l) \right] \left[ \sum_{q=0}^{2^{n-1}-1} 2\re{\alpha_{q}^{+^*} \alpha_q^-} h^*_{\mathcal{N}_\nu} (q)  \right] \\
	&= \sum_l \lambda_l h^*_{\mathcal{N}_\mu} (l) h^*_{\mathcal{N}_\nu} (l) - \left[ \sum_{l} v_l h^*_{\mathcal{N}_\mu} (l) \right] \left[ \sum_{q} v_q h^*_{\mathcal{N}_\nu} (q)  \right] .
\end{aligned}
\label{eq:qfiarbitrary1}
\end{equation}

Let $\boldsymbol{\mathcal{N}} = \{ \mathcal{N}_\mu\}_{\mu \in V}$ be the partition of the set of qubits into each of the nodes. Assembling the last equation for all entries $\mu\nu$ we get exactly:

\begin{equation}
\begin{aligned}
	\boldsymbol{\mathcal{Q}} (\rho_{\thetab}) &= \sum_m \lambda_m \vec{h}^*_{\boldsymbol{\mathcal{N}}} (m) \vec{h}^{*^T}_{\boldsymbol{\mathcal{N}}} (m)  - \sum_{m} v_m \vec{h}^*_{\boldsymbol{\mathcal{N}}} (m)  \sum_{q} v_q \vec{h}^{*^T}_{\boldsymbol{\mathcal{N}}} (q) \\
	&= \sum_{m,q} \mathfrak{Q}_{mq} \vec{h}^*_{\boldsymbol{\mathcal{N}}}  (m) \vec{h}^{*}_{\boldsymbol{\mathcal{N}}}  (q)^T \\
	&= C \boldsymbol{\mathfrak{Q}} C^T,
\end{aligned}
\label{eq:qfiarbitrary2}
\end{equation}
where $C$ is a matrix where each line is given by a vector $\vec{n}-2\vec{c}_i$, where $\vec{0} \preceq \vec{c}_i \prec \vec{n}$.
This is equivalent, under the choice of $G = Z, G^\perp = X$ to the GHZ state on $n$ qubits, and the basis formed by $\ket{\mathcal{G}_m^\pm} \propto \ket{m} \pm \ket{\overline{m}}$ where $m$ ranges between 0 and $2^{n-1}-1\equiv N$, is represented by the binary string associated to the number $m$ and $\overline{m}$ is the binary negation of $m$ ($e.g.$ $m=00101, \overline{m}=11010$). Furthermore, we also get two additional important properties: from Eq.~\ref{eq:qfiarbitrary2} we can see that the QFI matrix can be decomposed in $\mathfrak{Q} = \Lambda - \vec{v}\vec{v}^T$, where $\Lambda = diag(\lambda_0,...,\lambda_m,...,\lambda_{N})$ and $\vec{v} = (v_0, ..., v_m, ..., v_{N})$. On top of this, we can verify that:

\begin{equation}
\begin{aligned}
	|\alpha_j^+ \pm \alpha_j^-|^2 \geq 0 \implies |\alpha_j^+|^2 + |\alpha_j^-|^2 \geq \pm 2\re{\alpha_j^+  \alpha_j^{-^*}} \implies \lambda_j \geq |v_j|.
\end{aligned}
\end{equation}

\section{Matrix Rank Inequalities \label{appendix:matrixrank}}

Let us introduce some additional theorems regarding symmetric real matrices, which are the ones resultant from our problem. We start by the following theorem, which allows us to pick one and only one subspace in the set of states we can choose to construct a private state:

\begin{restatable}{theorem}{rankk}\label{thm:rank}
Let $A \in \mathbb{R}^{n\times n}$ be a symmetric matrix given by $A = \Lambda - \vec{v}\vec{v}^T$, where $\Lambda$ is a diagonal matrix with all entries real and non-negative summing to 1 and $\vec{v}$ is a vector in $\mathbb{R}^{n}$ such that $\norm{\vec{v}}\leq 1$ and in particular $\forall i \in \{1,\dots,n\}: \lambda_i \geq |v_i |$. $A$ is positive semi-definite and the rank of $A$ is at least
\begin{equation*}
	\rank \ \Lambda \geq \rank \ A \geq \rank \ \Lambda -1,
\end{equation*}
and $\rank \ A = \rank \ \Lambda -1$ iff. $\lambda_j = |v_j|$, $ \forall j$.
\end{restatable} 

\begin{proof}~
We start by reducing the problem into considering only the non-zero entries of $\Lambda$ and $\vec{v}$, as all the zero entries of $\Lambda$ do not contribute for the rank and the zero entries of $\vec{v}$ have corresponding values in $\Lambda$ which are already eigenvalues of $A$. Also, we can never have a non-zero entry in $\vec{v}$ associated to a zero value in $\Lambda$, from the conditions stated. Denote the vector $\vec{\lambda}$ as the vector of all the positive diagonal values of $\Lambda$. The condition imposed on $\vec{v}$ means that $\forall i \in \{1,\dots,n\}: \lambda_i \geq |v_i |$, which allows us to say:
\begin{equation}
\begin{aligned}
	\forall \vec{y} \neq \vec{0} \ , \exists \  \vec{x} \neq \vec{0}: \vec{y}^T\left( \Lambda - \vec{v}\vec{v}^T\right) \vec{y} &\geq  \vec{x}^T \left( \Lambda - \vec{\lambda}\vec{\lambda}^T \right) \vec{x} \\
    \vec{y}^T A \vec{y} &\geq  \vec{x}^T B \vec{x}.
\end{aligned}
\label{eq:inequality1}
\end{equation}
In particular, one could consider the vector $\vec{x}$ where $(\vec{x})_j = |(\vec{y})_j|$. By doing this, we can find bound the eigenvalues of $A$.
\begin{equation}
\begin{aligned}
	\vec{x}^T \left( \text{diag} (\vec{\lambda}) - \vec{\lambda}\vec{\lambda}^T \right) \vec{x} &= \sum_i x_i^2 \lambda_i - (\sum_i x_i \lambda_i)^2 \\
	&= \sum_i x_i^2 \lambda_i - (\sum_i x_i \lambda_i)(\sum_j x_j \lambda_j) \\
	&=  \sum_i x_i^2 \lambda_i (1-\lambda_i) - \sum_{i,  j\neq i} x_i \lambda_i x_j \lambda_j \\ 
	&=  \sum_i x_i^2 \lambda_i \sum_{j\neq i}\lambda_j - 2 \sum_{i,  j> i} x_i \lambda_i x_j \lambda_j \\
	&=  \sum_i \sum_{j> i} \lambda_i \lambda_j (x_i^2 - 2 x_i x_j +x_j^2 ) \\
	&=  \sum_i \sum_{j> i} \lambda_i \lambda_j (x_i - x_j )^2 .
\end{aligned}
\end{equation}
Since each $\lambda_i$ is by construction positive and non-zero, this means this expression is always positive, and is only zero when $(x_i-x_j)^2 = 0 \ \forall i,j$, which only happens for the null vector $\vec{0}$ and the $\vec{1}/\sqrt{n}$. 

Suppose $\exists j: \lambda_j > |v_j | $. In this case, Eq. \ref{eq:inequality1} is never saturated, unless $(\vec{y})_j$ is zero, which implies $\vec{x}^T B \vec{x} > 0$. This means $A$ is definite positive, and consequently that it has the rank of $\Lambda$. The other possible case is if $\forall \  j: \lambda_j = |v_j | $. Consider the matrix $D = diag((\vec{v})_j/|(\vec{v})_j|)$ which transforms $\vec{v} \rightarrow \vec{\lambda}$, $D\vec{v}=\vec{\lambda}$. Notice this matrix is a change of basis as $DD^T=\mathbb{1}$ and moreover that $DAD^{T} = B$ in this case. This allows us to map the eigenvectors of $B$ into the eigenvectors of $A$, proving the lower bound on the rank, as $\vec{1}/\sqrt{n}$ is the only non-zero vector than spans the null space of $B$. This means the rank of $A$ is bounded between the rank of $\Lambda$ and rank $\Lambda$ -1, where the latest equality only happens iff.  $\forall \  j: \lambda_j = |v_j | $.
\end{proof}

\begin{proposition}
Let $A \in \mathbb{R}^{n\times n}$ be symmetric and $B \in \mathbb{R}^{k\times n}$. If $A$ has full rank, then $BAB^T$ has rank equal to the rank of $B$.
\begin{equation*}
	\rank \ A = n \quad  \implies \quad \rank \  BAB^T = \rank \ B
\end{equation*}
\label{prop:fullrank}
\end{proposition}

\begin{proof}
Since $A$ is symmetric, consider its diagonalization $A = Q \Lambda Q^T$.
\begin{equation*}
	 BAB^T = B Q \Lambda Q^T B^T = B Q \sqrt{\Lambda} Q^T Q \sqrt{\Lambda} Q^T B^T = \tilde{B}\tilde{B}^T \implies \rank \  \tilde{B}\tilde{B}^T = \rank  \tilde{B} = \rank B,
\end{equation*}
where used properties of the rank, namely, $\rank B B^T = \rank B^T B = \rank B = \rank B^T$ and if $A$ is a $l\times m$ matrix with rank $m$ and $B$ is a $m\times n$ matrix, $\rank AB = \rank B$.
\end{proof}

\begin{proposition}
Let $A \in \mathbb{R}^{n\times n}$ be symmetric and $B \in \mathbb{R}^{k\times n}$. If $A$ has rank $n-1$, then:
\begin{equation*}
	\rank \ A = n-1 \quad  \implies \quad \rank \ B \geq \rank \ BAB^T \geq  \rank \ B -1 .
\end{equation*}
\label{prop:notfullrank}
\end{proposition}

\begin{proof}
Consider the same as before, but in the last step apply the Sylvester inequality for the rank of $\sqrt{A}B$:
\begin{equation*}
	  \rank \ B \geq \min \{ \rank \ \sqrt{A} , \rank \ B \} \geq  \rank \ \tilde{B} \geq  \rank \ B + \rank \ \sqrt{A} - n = \rank \ B -1.
\end{equation*}
\end{proof}

\begin{proposition}
Let $V,W$ be two vector spaces. Let $V_1$ and $V_2$ be two $m\times n$ matrices with vectors in $V$. If $\mathsf{span} (V_1) = \mathsf{span}(V_2)$ then $\mathsf{span}(V_1 W_1) = \mathsf{span} (V_2 W_1)$ for any $n\times p$ matrix $W_1$ with vectors in $W$.
\label{prop:vectorspan}
\end{proposition}

\begin{proof}
If $\mathsf{span} (V_1) = \mathsf{span}(V_2)$ then $\exists Q: V_2 = Q V_1$ and is invertible. The proof concludes by realising that $\mathsf{span} (QA) = \mathsf{span}(A)$ for any invertible $Q$, given the definition of a span of a set of vectors.
\end{proof}

\section{Private States - Proofs \label{appendix:privatestateproofs}}

\onecopy*

\begin{proof}
The proof in the indirect way comes trivially from calculating the privacy for the family of states $F_{GHZ}$, together with Corol.~\ref{thm:unitaryequivalence}, we get that difference up to a LU operation.
Doing the same as before, now we have a different $\mathfrak{Q}$ and $C$ matrices, as $\vec{n} = \vec{a}$. When looking at the matrix $\mathfrak{Q}$ note that the first coefficient $\lambda_0$ and $v_0$ are associated with the span of the family of states $F_{GHZ}$, as by varying them, we get all the states in $F_{GHZ}$. Doing the same analysis as before, we get that now $\vec{a} \in C$ and $dim_{\vec{a}} (C) = 1$, $i.e.$ it only appears once in the lines of $C$. Constructing the same table to analyze privacy:

\begin{table}[H]
\centering
\begin{tabular}{c|c|c|c|c}
	rank $\supp\mathfrak{Q}$ & rank $C_{\supp\mathfrak{Q}}$ & rank $C\mathfrak{Q}C^T$ & Reasoning & Privacy  \\ \hline
	any & 1 & $\leq 1$ & Trivial & Iff. $C_{\supp\mathfrak{Q}} = \vec{a}$ \\
	dim $\supp\mathfrak{Q}$ & $\geq 2$ & $\geq 2$ & Prop. \ref{prop:fullrank} & Never  \\
	dim $\supp\mathfrak{Q}$-1 & 2 & 1 or 2 & Prop. \ref{prop:notfullrank} & Check \\
	dim $\supp\mathfrak{Q}$-1 & $>2$ & $>1$ & Prop. \ref{prop:notfullrank} & Never  
\end{tabular}
\end{table}

Checking again the case of rank $\supp\mathfrak{Q}$ = dim $\supp\mathfrak{Q}-1$ with rank $C_{\supp\mathfrak{Q}}$ = 2. Using the same reasoning we find a vector $D\vec{w}$ in the orthogonal subspace of $D\vec{1}$, $\mathsf{span} P^\perp$, acting on two different lines of $C_{\supp\mathfrak{Q}}$:
\begin{equation}
\begin{aligned}
	D\vec{w} C_{\supp \mathfrak{Q}} &= 1/\sqrt{2} \left( \vec{h}_{\boldsymbol{\mathcal{N}}}^* (i) \pm \vec{h}_{\boldsymbol{\mathcal{N}}}^* (j) \right) \\
	&=\begin{cases}
	2/\sqrt{2} \left[ \vec{n} - \left(\vec{h}_{\boldsymbol{\mathcal{N}}} (i) + \vec{h}_{\boldsymbol{\mathcal{N}}} (j) \right) \right] \\
	2/\sqrt{2} \left(\vec{h}_{\boldsymbol{\mathcal{N}}} (i) - \vec{h}_{\boldsymbol{\mathcal{N}}} (j) \right) 
	\end{cases}.
\end{aligned}
\end{equation}
From here there are two options:
\begin{enumerate}[label=(\roman*)]
    \item $i$ or $j$ correspond to the GHZ state, and in that case we have that $\vec{h}_{\boldsymbol{\mathcal{N}}} (i) \pm \vec{h}_{\boldsymbol{\mathcal{N}}} (j) = \vec{0} \equiv \vec{h}_{\boldsymbol{\mathcal{N}}} (j) = \vec{0}$, meaning that $i=j$ and contradicting rank-$2$ of $C_{\supp\mathfrak{Q}}$;
    \item neither $i$ nor $j$ correspond to the GHZ state, and we recover that $C_{\supp\mathfrak{Q}}$ only has vectors $\vec{c}_i \prec \vec{a}$, which is the case for the non-existence of a private state.
\end{enumerate}
This means also this case is never private. Consequently, the only private state is the one spanned by a matrix with a unique non-zero entry correspondent to $\lambda_0-v_0^2$. Moreover, this is maximized for $v_0^2 = 0$, meaning a perfect GHZ state, up to a LU operation.

\end{proof}

\privatefamily*

\begin{proof}
This is a direct consequence of the invariance of the Hamming-weight introduced in Props.~\ref{prop:invarianceperm},\ref{prop:invariancepermS}. Looking at the terms that appear in the QFI, fix the encoding dynamics basis at $Z$, as we have chosen our distributed $s$-states (Def.~\ref{def:distdickestates}) in the computational basis. Notice this does not affect the generality of our proof, as again we are always an LU operation from any other basis (Corol.~\ref{thm:unitaryequivalence}). Using this fact, and the fact that the vectorial Hamming-weight is invariant under the distributed permutations, we get that:
\begin{equation}
\begin{aligned}
	\sum_{j\in\mu} Z_j \ket{\psi} &= h_{\mathcal{N}_\mu}(\ket{s({b_1}_\mu,a_\mu+{b_0}_\mu)}) \ket{\psi} = \left(a_\mu+{b_0}_\mu - {b_1}_\mu \right) \ket{\psi} ,\\
    \sum_{j\in\mu} Z_j \ket{\tilde{\psi}} &= h_{\mathcal{N}_\mu}(\ket{s(a_\mu+{b_1}_\mu,{b_0}_\mu)}) \ket{\tilde{\psi}} = \left(-a_\mu +{b_0}_\mu - {b_1}_\mu \right) \ket{\tilde{\psi}},
\end{aligned}
\end{equation}
Where we used that $\vec{b} = \vec{b}_0 + \vec{b}_1$ which correspond to the amount of 0's and 1's respectively. Note that $\vec{d}=\vec{b}_1$ as well. After assembling in the QFI matrix we get:
\begin{equation}
\begin{aligned}
	\boldsymbol{\mathcal{Q}} &= \vec{a} \vec{a}^T \left[ |\alpha^2| + |\beta|^2 - (|\alpha|^2-|\beta|^2)^2 \right],
\end{aligned}
\end{equation}
which means private. Moreover, it is maximised for $|\alpha|=|\beta|$.
\end{proof}

\onecopyplusancillaarb*

\begin{proof}
The inverse direction is a direct consequence from Thm.~\ref{prop:privatefamily}. The direct proof comes again from the QFI. Construct the same decomposition into $C\mathfrak{Q}C^T$ as before, where the basis of $\mathfrak{Q}$ are the families correspondent to $\boldsymbol{F}_{GHZ}$. 

\begin{table}[H]
\centering
\begin{tabular}{c|c|c|c|c}
	rank $\supp\mathfrak{Q}$ & rank $C_{\supp\mathfrak{Q}}$ & rank $C\mathfrak{Q}C^T$ & Reasoning & Privacy  \\ \hline
	any & 1 & $\leq 1$ & Trivial & Iff. $C_{\supp\mathfrak{Q}} = \vec{a}$ \\
	dim $\supp\mathfrak{Q}$ & $\geq 2$ & $\geq 2$ & Prop. \ref{prop:fullrank} & Never  \\
	dim $\supp\mathfrak{Q}$-1 & 2 & 1 or 2 & Prop. \ref{prop:notfullrank} & Check \\
	dim $\supp\mathfrak{Q}$-1 & $>2$ & $>1$ & Prop. \ref{prop:notfullrank} & Never  
\end{tabular}
\end{table}

Now, for the first time we have that $\dim \{\vec{c}_j \in C: \vec{c}_j= \vec{a} \} \equiv \mathfrak{Q}_{\vec{a}}$ is either 0 or larger than 1. It is larger than 1 if $\vec{b} = \vec{n}-\vec{a}$ has all entries divisible by two, and one can find the family of states such that $\vec{b}_0 = \vec{b}_1$. This is a consequence of the symmetric Hamming weight that remains the same when adding the same amount of 0s and 1s to a bitstring.

Looking at the case of rank $\supp\mathfrak{Q}$ = dim $\supp\mathfrak{Q}-1$ with rank $C_{\supp\mathfrak{Q}}$ = 2, we can find a way to construct all the others families. Using the same reasoning as before, let $D$ be the matrix  such that $D\vec{1}$ spans the null space of $\boldsymbol{\mathcal{Q}}$. Denote the orthogonal space of this vector by $P^\perp$. Let $W$ be the set of vectors given by $W^\pm = \{w| w= \vec{e}_i \pm \vec{e}_j , \text{sign} (D_{ii}) = \mp \text{sign}(D_{jj})\}$, and Prop.~\ref{prop:vectorspan} allows us to use this vectors to continue our proof, as $\textsf{span}(P^\perp)=\mathsf{span}(W^+\cup W^-)$. Looking first at $W^-$,the options let us check which vectors $\vec{c}_i - \vec{c}_j$ are proportional to $\vec{a}$:
\begin{equation}
\begin{aligned} 
	\vec{h}_{\boldsymbol{\mathcal{N}}}^* (i) - \vec{h}_{\boldsymbol{\mathcal{N}}}^* (j) &= \delta \vec{a}, \qquad \delta = \{  \pm 1, \pm 2 \} \\
	\vec{n} - 2\vec{h}_{\boldsymbol{\mathcal{N}}} (i) - \left( \vec{n} - 2 \vec{h}_{\boldsymbol{\mathcal{N}}} (j) \right) &=  \delta \vec{a}, \qquad \delta = \{  \pm 1, \pm 2 \}  \\
	\vec{h}_{\boldsymbol{\mathcal{N}}} (i) - \vec{h}_{\boldsymbol{\mathcal{N}}} (j)  &=  \tilde{\delta}  \vec{a}  , \qquad \tilde{\delta} = \{ \pm 1\} \\
	\implies \vec{h}_{\boldsymbol{\mathcal{N}}} (i) = \vec{a} + \vec{b}_j &\wedge \vec{h}_{\boldsymbol{\mathcal{N}}} (j) = \vec{b}_j ,
\end{aligned}
\end{equation}
where we used the fact that $\vec{h}_{\boldsymbol{\mathcal{N}}} (\cdot) $, is an positive integer vector and $gcd(\vec{a})=1$. Moreover, we introduce the notation $\vec{b}_i$ as a vector between $\vec{0} \preceq \vec{b}_i \prec \vec{b}$. Substituting for $\vec{h}_{\boldsymbol{\mathcal{N}}}^* (i)$ and $\vec{h}_{\boldsymbol{\mathcal{N}}}^* (j)$ and using $\vec{b}_i =\vec{b}-\vec{b}_j$ we observe: 
\begin{equation}
\begin{aligned} 
	\vec{h}_{\boldsymbol{\mathcal{N}}}^* (i) &= \vec{n} - 2( \vec{a} + \vec{b}_j) =  \vec{a} + \vec{b} - 2( \vec{a} + \vec{b}_j) = - \vec{a} + \vec{b} - 2 \vec{b}_j = - (\vec{a} + \vec{b} - 2 \vec{b}_i ) = -(\vec{n}-2\vec{b}_i)  \\
	\vec{h}_{\boldsymbol{\mathcal{N}}}^* (j) &= \vec{n} -2 \vec{b}_j  =  \vec{a} + \vec{b} - 2 \vec{b}_j.
\end{aligned}
\end{equation}
The equations for the subspace $W^+$ are:
\begin{equation}
\begin{aligned} 
	\vec{h}_{\boldsymbol{\mathcal{N}}}^* (i) + \vec{h}_{\boldsymbol{\mathcal{N}}}^* (j) &= \delta \vec{a}, \qquad \delta = \{ 0, \pm 1, \pm 2 \} \\
	\vec{n} - 2\vec{h}_{\boldsymbol{\mathcal{N}}} (i) + \left( \vec{n} - 2 \vec{h}_{\boldsymbol{\mathcal{N}}} (j) \right) &=  \delta \vec{a}, \qquad \delta = \{ 0, \pm 1, \pm 2 \}  \\
	\vec{n} - \left(\vec{h}_{\boldsymbol{\mathcal{N}}} (i) + \vec{h}_{\boldsymbol{\mathcal{N}}} (j) \right) &=  \tilde{\delta}  \vec{a}  , \qquad \tilde{\delta} = \{ 0, \pm 1\} \\
	\implies \left( \vec{h}_{\boldsymbol{\mathcal{N}}} (i) + \vec{h}_{\boldsymbol{\mathcal{N}}} (j)  = \vec{a} + \vec{b} \right) &\vee \left( \vec{h}_{\boldsymbol{\mathcal{N}}} (i)  + \vec{h}_{\boldsymbol{\mathcal{N}}} (j) = \vec{b} \right) \vee \left(  \vec{h}_{\boldsymbol{\mathcal{N}}} (i) + \vec{h}_{\boldsymbol{\mathcal{N}}} (j) = 2\vec{a} +\vec{b} \right).
\end{aligned}
\end{equation}
This means the overall solution are the vectors inside $C_{\supp\mathfrak{Q}}$, such that:
\begin{equation}
\begin{aligned} 
	\begin{cases}
	 \vec{c}_i = \vec{n} - 2 \left(\vec{a} + \vec{b}_i \right) \vee \vec{c}_i = \vec{n} - 2  \vec{b}_i, \\
	 \vec{c}_j = \vec{n} - 2 \left(\vec{a} + \vec{b}_j \right) \vee \vec{c}_j = \vec{n} - 2  \vec{b}_j,
	\end{cases}, \qquad  \vec{b}_i + \vec{b}_j = \vec{b}.
\end{aligned}
\label{eq:vecbibj}
\end{equation}
Substituting we actually see that $\vec{n} - 2 \left(\vec{a} + \vec{b}_i \right) = -( \vec{n} - 2  \vec{b}_j )$, meaning they are in fact spanned by a state in the same family (see Prop.~\ref{prop:samestate}). This means choosing only the combination $\vec{c}_i = \vec{n} - 2  \vec{b}_i \wedge \vec{c}_j = \vec{n} - 2  \vec{b}_j$ is sufficient. The set of states with the same vectorial Hamming-weight, are exactly the states belonging to the equivalence class, by definition. From the construction of matrix $\mathfrak{Q}$, we get that each line is associated with the following family of states:
\begin{equation}
\begin{aligned} 
	&\overline{F}_{GHZ_m} = \left\{ \alpha \ket{m} + \beta \ket{\overline{m}}, \alpha,\beta \in \mathbb{C}, |\alpha|^2+|\beta|^2=1 \right\} \\
	&\lambda_i = |\alpha_i|^2+|\alpha_{i+1}|^2 ,	v_i = 2\re{\alpha_i \alpha_{i+1}^*}, \alpha = \frac{\alpha_i + \alpha_{i+1}}{\sqrt{2}}, \beta = \frac{\alpha_i - \alpha_{i+1}}{\sqrt{2}} \\
	&\lambda_i = |v_i|  \implies \begin{cases} 
	\alpha = 1 , \beta = 0 , v_i/|v_i| = 1 \\
	\alpha = 0 , \beta = 1 , v_i/|v_i| = -1 
	\end{cases}.
\end{aligned}
\end{equation}
This means we can only take states $i$ such that, up to permutation of qubits (see Prop.~\ref{prop:permivariance}), that have $\vec{h}_{\boldsymbol{\mathcal{N}}}(i) = \vec{b}_i$, together with states $j$ that $\vec{h}_{\boldsymbol{\mathcal{N}}}(\overline{j}) = \vec{a} + \vec{b}_j$. Given that we can only choose $\vec{b}_i, \vec{b}_j$ such that $\vec{b}_i + \vec{b}_j = \vec{b}$ (Eq.~\ref{eq:vecbibj}) and $\vec{h}_{\boldsymbol{\mathcal{N}}}(\overline{j})  = \vec{n} - \vec{h}_{\boldsymbol{\mathcal{N}}}(j) =  \vec{n} -(\vec{a} + \vec{b}_j) = \vec{b}-\vec{b}_j$, then this are exactly the family of states given by $\mathcal{F}(\boldsymbol{\mathcal{N}},\vec{a},\vec{b}_i)$. By doing the same for each possible set of $\vec{b}_i,\vec{b}_j$, one generates each of the possible families of private states. Note however that one cannot create a superposition of states of different families, as that means the rank of $C_{\supp\mathfrak{Q}}$ would be larger than 2, which would always mean a non-private state (check table). Moreover, one recovers the fact that if $\vec{b}$ is divisible by 2, then for $\vec{b}_i = \vec{b}_j = \vec{b}/2$ is the case of a rank-1 $C_{\supp \mathfrak{Q}}$. Given that all other possibilities create either not rank-1 QFI matrices, or QFI matrices which are not proportional to $\vec{a}\vec{a}^T$, this are the only sets of private states. 
\end{proof}

Below we provide two additional propositions that allowed some simplifications on the theorems above.

\begin{proposition}\label{prop:permivariance}
Let $\sigma \in S_{\boldsymbol{\mathcal{N}}}$. Let $\ket{GHZ_m} = \alpha \ket{m} + \beta \ket{\overline{m}}$, where $m$ is a binary string with length $\norm{\vec{n}}_1$ and $\overline{(\cdot)}$ is the binary conjugation of $(\cdot)$.
\begin{equation}
\begin{aligned}
	h_{\boldsymbol{\mathcal{N}}} ( m ) = h_{\boldsymbol{\mathcal{N}}} \left( \sigma(m) \right) \Longleftrightarrow \mathcal{Q} (\ket{GHZ_m}) = \mathcal{Q} \left(\sigma(\ket{GHZ_m})\right)
\end{aligned}
\end{equation}
\end{proposition}

\begin{proposition}\label{prop:samestate}
Let $\ket{GHZ_m} = \alpha \ket{m} + \beta \ket{\overline{m}}$ and $\ket{GHZ_{\overline{m}}} = \alpha \ket{\overline{m}} + \beta \ket{m}$, where $m$ is a binary string with length $\norm{\vec{n}}_1$ and $\overline{(\cdot)}$ is the binary conjugation of $(\cdot)$.
\begin{equation}
\begin{aligned}
	\mathcal{Q} (\ket{GHZ_m}) = \mathcal{Q} \left(\ket{GHZ_{\overline{m}}}\right)
\end{aligned}
\end{equation}
\end{proposition}

Finally, the statement for when we have multiples of the vector of resources:

\multiplecopiesprivacy*

\begin{proof}
From Corol.~\ref{thm:unitaryequivalence} we get that we can fix dynamics as $Z$ dynamics, without loosing any generality, and work with the regular computational basis for the qubits. This comprises the statement up to LU operations. Then, starting from:
\begin{equation}
\begin{aligned}
	\sum_{j\in l} Z_j \ket{0_L}_l \equiv \boldsymbol{Z}_\mu^l \ket{0_L}_l &= \left(a_\mu+{b_0}_\mu^l - {b_1}_\mu^l \right) \ket{0_L}  = \left(a_\mu+\Delta b_\mu^l \right) \ket{0_L} ,\\
    \sum_{j\in l} Z_j \ket{1_L}_l \equiv \boldsymbol{Z}_\mu^l \ket{1_L}_l  &= \left(-a_\mu +{b_0}_\mu^l - {b_1}_\mu^l \right) \ket{1_L} = \left(-a_\mu+\Delta b_\mu^l \right)\ket{1_L},
\end{aligned}
\end{equation}
define $\ket{j} = \ket{ j_1  j_2 \cdots  j_d}$, such that $\boldsymbol{Z}_\mu^l \ket{j_L}_l = \left( (-1)^{j_l} a_\mu + \Delta b_\mu^l \right)  \ket{j_L}_l$ and note that:
\begin{equation}
\begin{aligned}
	\boldsymbol{Z}_\mu \ket{j} = \sum_{l = 1}^d \boldsymbol{Z}_\mu^l \ket{j} &= \left( h^*(j) a_\mu + \sum_{l=1}^d \Delta b_\mu^l \right)  \ket{j}.
\end{aligned}
\end{equation}
This means that, for $\ket{\psi} \in \mathfrak{F}$, the QFI will have the following form:
\begin{equation}
\begin{aligned}
	\mathcal{Q}_{\mu\nu} &\propto \bra{\psi} \boldsymbol{Z}_\mu \boldsymbol{Z}_\nu \ket{\psi} - \bra{\psi} \boldsymbol{Z}_\mu  \ket{\psi} \bra{\psi} \boldsymbol{Z}_\nu \ket{\psi} \\
    &= \sum_{j,m=0}^{2^d-1} \alpha_j^* \alpha_m \bra{j} \boldsymbol{Z}_\mu \boldsymbol{Z}_\nu \ket{m} -\sum_{j,m=0}^{2^d-1} \alpha_j^* \alpha_m \bra{j} \boldsymbol{Z}_\mu  \ket{m} \sum_{p,q=0}^{2^d-1} \alpha_p^* \alpha_q \bra{p}  \boldsymbol{Z}_\nu \ket{k} \\
    &= \sum_{j=0}^{2^d-1} |\alpha_j|^2 \left( h^*(j) a_\mu + \sum_{l=1}^d \Delta b_\mu^l \right) \left( h^*(j) a_\nu + \sum_{l=1}^d \Delta b_\nu^l \right) \\
    &\qquad \qquad - \sum_{j,p=0}^{2^d-1} |\alpha_j|^2 |\alpha_p|^2 \left( h^*(j) a_\mu + \sum_{l=1}^d \Delta b_\mu^l \right) \left( h^*(p) a_\nu + \sum_{l=1}^d \Delta b_\nu^l \right)  \\
    &= \left[ \sum_{j=0}^{2^d-1} |\alpha_j|^2 h^*(j)^2 - \sum_{j,p=0}^{2^d-1} |\alpha_j|^2 |\alpha_p|^2 h^*(j) h^*(p) \right] a_\mu a_\nu ,
\end{aligned}
\end{equation}
which makes the QFI matrix always proportional to $\vec{a}\vec{a}^T$, resulting always in a private state.

\end{proof}

\section{Private States Existence for General Hamiltonians \label{appendix:general}}

Let us start by repeating and proving Thm.~\ref{thm:existenceprivate}:

\existenceprivate*

\begin{proof}
Given the rewriting of the QFI matrix in Eq.~\ref{eq:generalQFI}, we already now from Thm.~\ref{thm:rank} that the rank of this matrix is given by the rank  $\supp\mathfrak{Q}-1$, similarly to what we have done before. By choosing the orthogonal subspace to the null space of $\supp \mathfrak{Q}$ we realize that the span of this space is generated by the vectors in $W = \{w| w= \vec{e}_{\vec{i}} - \vec{e}_{\vec{j}} \}$. Note that $W C_{\supp \mathfrak{Q}}$ generates a space of vectors $\mathcal{C} = \{\vec{c}| \vec{c}= \vec{c}_{\vec{i}} - \vec{c}_{\vec{j}} , \vec{c}_{\vec{i}}, \vec{c}_{\vec{j}} \in \mathcal{O}\}$. This is exactly the higher order discrete subspace generated by $\mathcal{O}$, $i.e.$ $\mathcal{C} \subsetsim \mathcal{O}^2_-$. This already proves the first part of the theorem about the existence of private states, by choosing states with $\alpha_{\vec{i}},\alpha_{\vec{j}}$ to be the only non-zero terms, associated to vectors $\vec{c}_{\vec{i}}-\vec{c}_{\vec{j}} = \vec{a}\in \mathcal{O}^2_-$.

The fact that only vectors in the equivalence class of $\mathcal{O}^2_-$ have private states is a consequence of $\mathsf{span}(\mathcal{C}) = \{ \alpha \vec{a}, \alpha \in \mathbb{R} \}$ in order for privacy to be verified. If $\alpha\vec{a} \not\in \mathcal{O}^2_-$ than one cannot build $\mathsf{span}(\mathcal{C})$, as every $\vec{c}_{\vec{i}}-\vec{c}_{\vec{j}}$ has to be proportional to $\vec{a}$ and $\vec{c}_{\vec{i}},\vec{c}_{\vec{j}}$ can only be chosen from $\mathcal{O}$.
\end{proof}

To prove Prop.~\ref{prop:symmetrichamiltonians} let us first introduce and prove some additional propositions that will allows us to prove it:

\begin{proposition}\label{prop:commutationpauli}
Let $X \in \mathcal{P}_n$ be a Pauli string and $\mathcal{H}\subseteq \mathcal{P}_n$ a subset of Pauli strings. Define an Hamiltonian as $H = \sum_{P \in \mathcal{H}} c_P P$, where all $c_P \neq 0$. Then:
\begin{equation*}
    [X,H]_\pm = 0 \Leftrightarrow [X,P]_\pm = 0 ,\quad \forall \ P \in \mathcal{H}.
\end{equation*}
\end{proposition}
\begin{proof}
The inverse implication comes trivially from the linearity of the (anti)commutator $[\cdot,\cdot]_\pm$. The direct implication comes as follows: divide $\mathcal{H}$ into two disjoint sets $\mathcal{H} = \mathcal{H}_+ \cup \mathcal{H}_-$ such that $[X,P]_\pm = 0, \ \forall \ P \in \mathcal{H}_\pm$. This is always possible given that two Pauli strings either commute or anticommute. Then:
\begin{equation*}
    \begin{aligned}
        [X,H]_\pm &= \sum_{P \in \mathcal{H}} c_P [X,P]_\pm \\
        &= \sum_{P \in \mathcal{H}_\pm} c_P [X,P]_\pm + \sum_{P \in \mathcal{H}_\mp} c_P [X,P]_\pm \\
        &= 0 +  \sum_{P \in \mathcal{H}_\mp} c_P [X,P]_\pm \\
        &= 2\sum_{P \in \mathcal{H}_\mp} c_P PX \\
        &= 2\sum_{\tilde{P} \in \mathcal{H}_\mp X} c_{\tilde{P}} \tilde{P} ,
    \end{aligned}
\end{equation*}
which can only be zero if all $c_{\tilde{P}}$ are 0, as the set of Pauli strings generates a basis for complex matrices, leading to a contradiction.
\end{proof}

\begin{proposition}\label{prop:symmetriceigenvalues}
Let $\mathcal{H}\subseteq \mathcal{P}_n$ be a subset of Pauli strings. Define an Hamiltonian as $H = \sum_{P \in \mathcal{H}} c_P P$, where all $c_P > 0$. Then:
\begin{equation*}
    \exists X \in \mathcal{P}_n : \{X, H \} = 0 \Leftrightarrow \{ \lambda_H \} = \{ + \lambda_H \} \cup \{ -\lambda_H\}
\end{equation*}
Meaning for each eigenvalue of the Hamiltonian $H$, $\lambda_H$, its symmetric counterpart, $-\lambda_H$, is also an eigenvalue of $H$.
\end{proposition}
\begin{proof}
The direct proof follows very simply. $X \in \mathcal{P}_n \implies X^2=\mathbb{1}$.
\begin{equation*}
\begin{aligned}
     \det(H-\lambda\mathbb{1}) &=  \det(H-\lambda\mathbb{1}) \det(X^2) \\
     &= \det(X)\det(H-\lambda\mathbb{1}) \det(X)\\
     &=\det(XHX-\lambda X\mathbb{1}X)\\
     &= \det(-H-\lambda\mathbb{1}),
\end{aligned}
\end{equation*}
which means the eigenvalues admit solutions for $\det(H-\lambda\mathbb{1})$ and $\det(H+\lambda\mathbb{1})$ and therefore if $\lambda_H \in$ ev, $-\lambda_H \in$ ev. The inverse proof comes as a consequence of the Pauli strings forming a basis for the space of complex matrices $\mathbb{C}^{2^n}\times\mathbb{C}^{2^n}$. Take $U$ to be a unitary operation that switches from positive eigenstates to the negative ones: $U\ket{\lambda_H^\pm} = \ket{\lambda_H^\mp}$. Then:
\begin{equation*}
\begin{aligned}
     U H \ket{\lambda_H^\pm} &= \pm \lambda_H U \ket{\lambda_H^\pm} = \pm \lambda_H \ket{\lambda_H^\mp} \\
     H U \ket{\lambda_H^\pm} &=  H \ket{\lambda_H^\mp} =  \mp \lambda_H \ket{\lambda_H^\mp}\\
     \implies& (UH + HU)\ket{\lambda_H^\pm} = 0, \quad \forall \ \ket{\lambda_H^\pm},
\end{aligned}
\end{equation*}
which means there is a unitary satisfying the anti-commutation relation. Given that $U$ is a complex matrix, one can find a decomposition in terms of Pauli strings. Using Prop.~\ref{prop:commutationpauli} and similar arguments for the proof of Prop.~\ref{prop:commutationpauli}, then one has to choose the decomposition among the set of anti-commuting Pauli strings only. This consequently means there is at least one Pauli string that anti-commutes with the Hamiltonian.
\end{proof}

This allows us to prove Prop.~\ref{prop:symmetrichamiltonians}:
\symmetrichamiltonians*
\begin{proof}
Follows directly from Prop.~\ref{prop:symmetriceigenvalues}. The stated conditions for each local Hamiltonian imply that the local eigenvalues $\lambda_\mu$ belong to a set such that if $+\lambda$ exists so does $-\lambda$. This means that if the vector $\vec{c}_i$ belongs to the discrete subset $\mathcal{O}$, so does the vector $-\vec{c}_i$. This consequently means that $\mathcal{O}^2_-$ where the vectors are built from $\vec{c}_i - \vec{c}_j$ will, in particular, include the vectors $\vec{c}_i-(-\vec{c}_i) = 2 \vec{c}_i$. So $\mathcal{O}\sim 2\mathcal{O} \subseteq \mathcal{O}^2_- \equiv \mathcal{O}\subsetsim \mathcal{O}^2_-$. 
\end{proof}

Finally, we are able to prove Corol.~\ref{corol:concentration}:

\concentration*
\begin{proof}
This is a consequence of the decomposition of $\mathcal{Q}=C\mathfrak{Q}C^T$ of Eq.~\ref{eq:generalQFI} and the fact that $\norm{\vec{a}}_2$ is maximum for $\vec{a}\in V(\overline{\mathcal{O}}^2_-)$. Among only the private states, $\vec{a}\in \overline{\mathcal{O}}^2_-$, one can calculate the maximal QFI as:

\begin{equation*}
    \max \mathcal{Q}(\vec{a}) = \max \vec{a}^T\mathcal{Q}\vec{a} =\norm{\vec{a}}_2^2 = \norm{\vec{c}_i - \vec{c}_j}_2^2 .
\end{equation*}

The privacy measure already incapsulates the fact that any private state will maximise information about the target function. This means the maximum information one can get about the maximum function is necessarily private.
\end{proof}

\section{QFI for Mixed States - Derivations \label{appendix:qfimixedstates}}

As said in the main text, to find the expression for the QFI in the mixed state case, we start by the most elementar expression for the QFI:

\begin{equation}
\begin{aligned}
	\mathcal{Q}_{\mu\nu} (\rho_{\thetab}) =  \Tr \mathcal{R}_{\rho_{\thetab}}^{-1} (\partial_{\theta_\mu}\rho_{\thetab}) \rho_{\thetab} \mathcal{R}_{\rho_{\thetab}}^{-1} (\partial_{\theta_\nu}\rho_{\thetab}).
\end{aligned}
\end{equation}

Under the assumptions of the main text regarding the superoperator $\mathcal{R}_{\rho_{\thetab}}$, which corresponds to the symmetric logarithmic derivative, we can derive it in terms of the density matrix eigenvectors:
\begin{equation}
\begin{aligned}
	\mathcal{R}_\rho^{-1} (\partial_{\theta_\mu}\rho_{\thetab})&= \sum_{j,k} \frac{2}{\lambda_j + \lambda_k} \bra{\mathcal{G}_j} \partial_{\theta_\mu}\rho_{\thetab}  \ket{\mathcal{G}_k} \ketbra{\mathcal{G}_j}{\mathcal{G}_k} .
\end{aligned}
\end{equation}

Using the notation $\ket{\mathcal{G}_k'} = U_\thetab \ket{\mathcal{G}_k}$:

\begin{equation}
\begin{aligned}
	\partial_{\theta_\mu}\rho_{\thetab}  &= \partial_{\theta_\mu} \sum_{k\in \supp} \lambda_k \ketbra{\mathcal{G}_k'}{\mathcal{G}_k'} \\
	&= \sum_{k\in \supp} \lambda_k \left[ \ketbra{\partial_{\theta_\mu} \mathcal{G}_k'}{\mathcal{G}_k'} + \ketbra{\mathcal{G}_k'}{\partial_{\theta_\mu} \mathcal{G}_k'} \right],
\end{aligned}
\end{equation}
where we assumed that the encoding is unitary and done over the density matrix.

\begin{equation}
\begin{aligned}
	D_{ij}^{\mu}\equiv \bra{\mathcal{G}_i'} \partial_{\theta_\mu}\rho_{\thetab}  \ket{\mathcal{G}_j'} &=\bra{\mathcal{G}_i'}  \sum_k \lambda_k \left[ \ketbra{\partial_{\theta_\mu} \mathcal{G}_k'}{\mathcal{G}_k'} + \ketbra{\mathcal{G}_k'}{\partial_{\theta_\mu} \mathcal{G}_k'} \right] \ket{\mathcal{G}_j'} \\
	&= \sum_k \lambda_k \left[ \braket{\mathcal{G}_i' | \partial_{\theta_\mu} \mathcal{G}_k'} \delta_{kj} + \delta_{ik} \braket{\partial_{\theta_\mu} \mathcal{G}_k' | \mathcal{G}_j'} \right] \\
	&=\id_{\supp} (j) \lambda_j \braket{\mathcal{G}_i' | \partial_{\theta_\mu} \mathcal{G}_j'} + \id_{\supp} (i) \lambda_i \braket{\partial_{\theta_\mu} \mathcal{G}_i' | \mathcal{G}_j'} \\
	&=\left[\id_{\supp} (i) \lambda_i - \id_{\supp} (j)\lambda_j \right] \braket{\partial_{\theta_\mu} \mathcal{G}_i' | \mathcal{G}_j'},
\end{aligned}
\end{equation}
where the last step was taking into consideration that:
\begin{equation}
\begin{aligned}
	\partial_{\theta_\mu} (\braket{ \mathcal{G}_i' | \mathcal{G}_j'}) = \partial_{\theta_\mu} (\braket{ \mathcal{G}_i | \mathcal{G}_j}) &= \partial_{\theta_\mu} (\delta_{ij}) = 0 \\
	&=  \braket{\partial_{\theta_\mu} \mathcal{G}_i' | \mathcal{G}_j'} + \braket{ \mathcal{G}_i' | \partial_{\theta_\mu} \mathcal{G}_j'} .
\end{aligned}
\end{equation}

Plugging it in the QFI we get that:
\begin{equation}
\begin{aligned}
	\mathcal{Q}_{\mu\nu} (\rho_\thetab) &=  \tr{ \mathcal{R}_{\rho_\thetab}^{-1} (\partial_{\theta_\mu}\rho_\thetab) \rho_\thetab \mathcal{R}_{\rho_\thetab}^{-1} (\partial_{\theta_\nu}\rho_\thetab)} \\
	&= \Tr \left[ \sum_{i,j} \frac{2}{\lambda_i + \lambda_j} D_{ij}^{\mu} \ketbra{\mathcal{G}_i'}{\mathcal{G}_j'} \cdot \sum_{k\in \supp} \lambda_k \ketbra{\mathcal{G}_k'}{\mathcal{G}_k'} \cdot \sum_{m,n} \frac{2}{\lambda_m + \lambda_n} D_{mn}^{\nu} \ketbra{\mathcal{G}_m'}{\mathcal{G}_n'} \right] \\
	&= \Tr \left[ \sum_{i,j,m,n} \frac{2D_{ij}^{\mu}}{\lambda_i + \lambda_j} \cdot \frac{2 D_{mn}^{\nu}}{\lambda_m + \lambda_n} \sum_{k\in\supp} \lambda_k \ket{\mathcal{G}_i'} \braket{\mathcal{G}_j' | \mathcal{G}_k'} \braket{\mathcal{G}_k' | \mathcal{G}_m'}\bra{\mathcal{G}_n'} \right] \\
	&= \Tr \left[ \sum_{ \substack{i,n \\ k \in \supp} } \frac{2D_{ik}^{\mu}}{\lambda_i + \lambda_k} \cdot \frac{2 D_{kn}^{\nu}}{\lambda_k + \lambda_n} \lambda_k \ket{\mathcal{G}_i'} \bra{\mathcal{G}_n'} \right] \\
	&= \sum_{\substack{n\in \supp \\ k\in\supp}} \frac{2D_{nk}^{\mu}}{\lambda_n + \lambda_k} \cdot \frac{2 D_{kn}^{\nu}}{\lambda_k + \lambda_n} \lambda_k +  \sum_{\substack{n\in \nul \\ k\in\supp}} \frac{2D_{nk}^{\mu}}{\lambda_n + \lambda_k} \cdot \frac{2 D_{kn}^{\nu}}{\lambda_k + \lambda_n} \lambda_k \\
	&=\sum_{\substack{n\in\supp \\ k\in\supp}}  \frac{2D_{nk}^{\mu} D_{kn}^{\nu} }{\lambda_n + \lambda_k}  \cdot \left( \frac{\lambda_k}{\lambda_k + \lambda_n} + \frac{\lambda_n}{\lambda_n + \lambda_k}  \right) + \sum_{\substack{n\in \nul \\ k\in\supp}} \frac{2D_{nk}^{\mu}}{\lambda_k} \cdot \frac{2 D_{kn}^{\nu}}{\lambda_k} \lambda_k \\
	&= \sum_{\substack{n \in \supp \\ k\in\supp}}  \frac{2D_{nk}^{\mu}  D_{kn}^{\nu} }{\lambda_n + \lambda_k} +  \sum_{\substack{n \in \nul \\ k\in\supp}}  \frac{4D_{nk}^{\mu}  D_{kn}^{\nu} }{\lambda_k}  \\
	&=\sum_{\substack{n \in \supp \\ k\in\supp}}  2\frac{(\lambda_n - \lambda_k)^2}{\lambda_n + \lambda_k} \braket{\partial_{\theta_\mu} \mathcal{G}_n' | \mathcal{G}_k'} \braket{ \mathcal{G}_k' | \partial_{\theta_\nu}\mathcal{G}_n'} + \sum_{\substack{n \in \nul \\ k\in\supp}}  4  \lambda_k \braket{ \mathcal{G}_n' | \partial_{\theta_\mu} \mathcal{G}_k'} \braket{ \partial_{\theta_\nu} \mathcal{G}_k' | \mathcal{G}_n'} \\
	&= \sum_{\substack{n \in \supp \\ k\in\supp}} 2\frac{(\lambda_n - \lambda_k)^2}{\lambda_n + \lambda_k} \re{\braket{\mathcal{G}_n | \boldsymbol{G}_\mu | \mathcal{G}_k} \braket{ \mathcal{G}_k | \boldsymbol{G}_\nu | \mathcal{G}_n}}+ \sum_{\substack{n \in \nul \\ k\in\supp}}  4  \lambda_k \re{ \braket{\mathcal{G}_k | \boldsymbol{G}_\mu | \mathcal{G}_n} \braket{ \mathcal{G}_n | \boldsymbol{G}_\nu | \mathcal{G}_k} }.
\end{aligned}
\label{eq:notfinalmixed}
\end{equation}

Let us retrieve the pure state equation from this. Assume the quantum state density matrix can be described by that of a pure state $\rho = \ketbra{\psi}{\psi}$. This means $\lambda_1 = 1$ and every other $\lambda_j = 0$. Moreover, since independent of the $\ket{\psi}$ choice, one can always find an orthonormal basis such that $\ket{\psi}$ is part of. This means in particular that $\exists \mathcal{B} = \{ \ket{\psi_j} \}_{j=1,\cdots,2^n}$, such that $\ket{\psi_1} = \ket{\psi}$ and $\sum_{j=1}^{2^n} \ketbra{\psi_j}{\psi_j} = \mathbb{1}$. Plugging $\rho$ into Eq.~\ref{eq:noisyqfi} we get:
\begin{equation}
\begin{aligned}
	\mathcal{Q}_{\mu\nu} (\rho_\thetab) &= 4 \re{ \bra{\psi_1 }  \boldsymbol{G}_\mu  \left(  \sum_{ k=2}^{2^n} \ketbra{\psi_k}{ \psi_k} \right)  \boldsymbol{G}_\nu  \ket{ \psi_1} } \\
	&=   4  \re{ \bra{\psi_1 }  \boldsymbol{G}_\mu  \left(  \mathbb{1} - \ketbra{\psi_1}{\psi_1} \right) \boldsymbol{G}_\nu  \ket{ \psi_1} } \\
	&=   4 \re{  \bra{\psi_1 }  \boldsymbol{G}_\mu  \boldsymbol{G}_\nu  \ket{ \psi_1} -  \bra{\psi_1 } \boldsymbol{G}_\mu  \ketbra{\psi_1}{\psi_1}   \boldsymbol{G}_\nu  \ket{ \psi_1} } .
\end{aligned}
\end{equation}

This allows us to get one additional way to calculate the QFI of a mixed state. Suppose the density matrix is a sum of orthogonal states (even if it is not, there is always a decomposition into orthogonal states by changing basis, consequence of a diagonalization of the matrix). Let:
\begin{equation}
\begin{aligned}
	\rho = \sum_{j \in \supp} \lambda_j \ketbra{\mathcal{G}_j}{\mathcal{G}_j},\quad \text{such that } \braket{\mathcal{G}_j | \mathcal{G}_k }= \delta_{jk}.
\end{aligned}
\end{equation}
Then:
\begin{equation}
\begin{aligned}
	\mathcal{Q}_{\mu\nu} (\rho_\thetab) &= \sum_{\substack{n \in \supp \\ k\in\supp}} \! 2\frac{(\lambda_n - \lambda_k)^2}{\lambda_n + \lambda_k} \re{\braket{\mathcal{G}_n | \boldsymbol{G}_\mu | \mathcal{G}_k} \braket{ \mathcal{G}_k | \boldsymbol{G}_\nu | \mathcal{G}_n}}+ \sum_{\substack{n \in \nul \\ k\in\supp}}  4  \lambda_k \re{ \braket{\mathcal{G}_k | \boldsymbol{G}_\mu | \mathcal{G}_n} \braket{ \mathcal{G}_n | \boldsymbol{G}_\nu | \mathcal{G}_k} } \\
	&= \sum_{\substack{n \in \supp \\ k\in\supp}} 2\frac{(\lambda_n - \lambda_k)^2}{\lambda_n + \lambda_k} \re{ a_{kn}^\mu a_{nk}^\nu } + \sum_{\substack{k\in\supp}}  4  \lambda_k  \re{ \bra{\mathcal{G}_k } \boldsymbol{G}_\nu \left[ \sum_{\substack{n\in\nul}} \ketbra{ \mathcal{G}_n}{\mathcal{G}_n} \right]  \boldsymbol{G}_\mu \ket{ \mathcal{G}_k}} \\
	&= \asum{\substack{n \in \supp \\ k\in\supp}} 2\frac{(\lambda_n - \lambda_k)^2}{\lambda_n + \lambda_k} \re{ a_{kn}^\mu a_{nk}^\nu } + \asum{\substack{k\in\supp}}  4  \lambda_k \re{ \bra{\mathcal{G}_k } \boldsymbol{G}_\nu \left[ \mathbb{1} - \asum{\substack{n\in\supp}} \ketbra{ \mathcal{G}_n}{\mathcal{G}_n} \right]  \boldsymbol{G}_\mu \ket{ \mathcal{G}_k}} \\
	&= \asum{\substack{n \in \supp \\ k\in\supp}} 2\frac{(\lambda_n - \lambda_k)^2}{\lambda_n + \lambda_k} \re{ a_{kn}^\mu a_{nk}^\nu } + \asum{\substack{k\in\supp}}  4  \lambda_k \re{ \bra{\mathcal{G}_k } \boldsymbol{G}_\nu \left[ \mathbb{1} - \ketbra{ \mathcal{G}_k}{\mathcal{G}_k} -\asum{\substack{n\in\supp \\ n\neq k}} \ketbra{ \mathcal{G}_n}{\mathcal{G}_n} \right]  \boldsymbol{G}_\mu \ket{ \mathcal{G}_k} }\\
	&= \asum{\substack{n \in \supp \\ k\in\supp}} 2\frac{(\lambda_n - \lambda_k)^2}{\lambda_n + \lambda_k} \re{ a_{kn}^\mu a_{nk}^\nu } + \asum{\substack{k\in\supp}}  \lambda_k \mathcal{Q}_{\mu\nu} (\ket{\mathcal{G}_k}) - \asum{\substack{n,k \in \supp \\ k\neq n } } 4  \lambda_k \re{ a_{kn}^\mu a_{nk}^\nu } \\
	&= \sum_{\substack{n,k \in \supp \\ n\neq k}} \left[ 2\frac{(\lambda_n - \lambda_k)^2}{\lambda_n + \lambda_k} - 4 \lambda_k \right] \re{ a_{kn}^\mu a_{nk}^\nu} + \sum_{\substack{k\in\supp}}  \lambda_k \mathcal{Q}_{\mu\nu} (\ket{\mathcal{G}_k}) \\
	&= \sum_{n>k \in \supp } 4\left[ \frac{(\lambda_n - \lambda_k)^2}{\lambda_n + \lambda_k} - (\lambda_n + \lambda_k) \right] \re{ a_{nk}^\nu a_{kn}^\mu} + \sum_{\substack{k\in\supp}}  \lambda_k \mathcal{Q}_{\mu\nu} (\ket{\mathcal{G}_k}) .
\end{aligned}
\label{eq:finalmixed}
\end{equation}

In the case the state density matrix is described by a sum of non-orthogonal vectors:
\begin{equation}
\begin{aligned}
	\rho = \sum_{j \in \supp} \lambda_j \ketbra{\phi_j}{\phi_j},\quad \text{such that } \braket{\phi_j | \phi_k } = \alpha_{jk} \leq 1,
\end{aligned}
\label{eq:nonorthogonalsetting}
\end{equation}
such that $\boldsymbol{G}_\mu \ket{\phi_j}$ is equal for every $\ket{\phi_j}$, then diagonalizing $\rho$ into vectors that are orthogonal between themselves, call it $\ket{\tilde{\phi}_j}$ still holds the fact that $\boldsymbol{G}_\mu \ket{\tilde{\phi}_j}$ is equal for every $\ket{\tilde{\phi}_j}$. If, more than that, $\re{ a_{nk}^\nu a_{kn}^\mu} \propto \mathcal{Q}_{\nu\mu} (\ket{\mathcal{G}_n}) $, then $\lambda_n$ should simply be scaled into a $\tilde{\lambda}_n$, but the information is still structurally the same, $i.e.$ it is the same function of parameters. 

This gives some intuition about how to analyze Eq.~\ref{eq:finalmixed} and when to use it over Eq.~\ref{eq:notfinalmixed}. In Eq.~\ref{eq:finalmixed} the first term corresponds to a distinguishability between the states under the encoding dynamics. If the generators of the dynamics are orthogonal to the noise, then this goes to zero. The second part of Eq.~\ref{eq:finalmixed} can be seen as the local informations of each states in a perfect scenario weighted by the corresponding weights. Note that the expression:
\[
4\left[ \frac{(\lambda_n - \lambda_k)^2}{\lambda_n + \lambda_k} - (\lambda_n + \lambda_k) \right] ,
\]
is always negative, meaning that if they are completely distinguishable, then we recover the notion of convexity of the QFI. If they are not completely distinguishable, then the amount of information is decreased by a factor proportional to $ \re{ a_{nk}^\nu a_{kn}^\mu} $. On the other hand, when the noise is only in the same direction of the generators of the encoding, then Eq.~\ref{eq:notfinalmixed} becomes simpler, as the second part goes to zero, and we are left simply with an expression that weights the distinguishability of the states, under the encoding dynamics.
\end{document}